\documentclass[letterpaper, 11pt]{article}

\usepackage[margin=1in]{geometry}

\usepackage[font=scriptsize]{caption}
\newcommand{\tb}{(2\beta+1)}

\usepackage{amsthm}
\usepackage{amsmath}
\usepackage{algorithm}
\usepackage[noend]{algpseudocode}
\newtheorem{theorem}{Theorem}[section]
\newtheorem{corollary}[theorem]{Corollary}
\newtheorem{lemma}[theorem]{Lemma}
\usepackage{cite}

\usepackage{graphicx}

\graphicspath{{hopsetsFig/}}

\theoremstyle{definition}

\usepackage{appendix}

\usepackage{amssymb}
\usepackage[colorinlistoftodos,prependcaption,textsize=small]{todonotes}

\usepackage{subcaption}

\usepackage{color}
\usepackage{tabulary}
\usepackage{adjustbox}
\usepackage{mathtools}
\DeclarePairedDelimiter\floor{\lfloor}{\rfloor}
\usepackage{tcolorbox}
\newcommand{\cedge}{(r_C,r_{C'})}
\newcommand{\pramell}{\lfloor {\log \kappa\rho}\rfloor + \lceil \frac{\kappa+1}{\kappa\rho}\rceil-1 }

\newcommand{\gmk}{\Gamma}
\newcommand{\dgk}{d^{(2\beta+1)}_{G_{k-1}}}
\newcommand\apdi{(1+\epsilon_{k-1} )\delta_i}
\newcommand\epsi{\left({1}/{\epsilon}\right)^i}

\newcommand\eps[1]{\left(\frac{1}{\epsilon}\right)^{#1}}

\newcommand\nfrac{n^{1+\frac{1}{\kappa}}}

\newcommand{\mcg}[1]{\mathcal{G}^{#1}}
\newcommand{\mch}[1]{\mathcal{H}^{#1}}
\newcommand{\mcv}[1]{\mathcal{V}^{#1}}
\newcommand{\partdd[1]}{\alpha\cdot (1/\epsilon)^{{#1}}}

\newcommand\ta[1]{5\cdot \alpha\cdot c(n) \cdot (1/\epsilon)^{#1}}
\newcommand{\partd}{\alpha\cdot \epsi}

\newcommand{\induchyp}{Assume that the claim holds for some $i\in [0,\ell-1]$, and prove it holds for $i+1$.}

\newcommand\congestmo{{\sf CONGEST} model}

\newcommand{\degi}{n^\frac{2^i}{\kappa}}

\newcommand{\rul}{(2\apdi+4R_i){\log n}}

\newcommand{\mindg}{(\epsilon/n)\cdot 2^k}
\newcommand{\maxdg}{ 2^{k+1}}

\newcommand{\str}{^*}

\newcommand{\krange}{ [k_0,\lambda]}

\usepackage{titlesec}
\titleformat{\subsubsection}[runin]
{\normalfont\normalsize\bfseries}{\thesubsubsection}{1em}{}

\usepackage[hypertexnames=false]{hyperref}
\usepackage{cleveref}
\crefname{lemma}{Lemma}{Lemmas}

\author{Michael Elkin$^1$
	 $\ $and Shaked Matar$^1$}
\date{$^1$Department of Computer Science, Ben-Gurion University of the Negev, Beer-Sheva, Israel.\\
	Email: \texttt{elkinm@cs.bgu.ac.il, matars@post.bgu.ac.il}}

\title{Deterministic PRAM Approximate Shortest Paths in Polylogarithmic Time and Slightly Super-Linear Work}

\begin{document}
	\begin{titlepage}
		
\maketitle
 \thispagestyle{empty}

\begin{abstract}
	\normalsize 
	We study a \emph{$(1+\epsilon)$-approximate single-source shortest paths} (henceforth, $(1+\epsilon)$-SSSP) in $n$-vertex undirected, weighted graphs in the parallel (PRAM) model of computation. A randomized algorithm with polylogarithmic time and slightly super-linear work $\tilde{O}(|E|\cdot n^\rho)$, for an arbitrarily small $\rho>0$, was given by Cohen \cite{coh94} more than $25$ years ago.
	Exciting progress on this problem was achieved in recent years \cite{ElkinN17Hop,ElkinN19,Li19,AndoniSZ19}, culminating in randomized polylogarithmic time and $\tilde{O}(|E|)$ work. However, the question of whether there exists a \emph{deterministic} counterpart of Cohen's algorithm remained wide open.
	
	In the current paper we devise the first \emph{deterministic} \emph{polylogarithmic-time} algorithm for this fundamental problem, with work $\tilde{O}(|E|\cdot n^\rho)$, for an arbitrarily small $\rho>0$. This result is based on the first efficient deterministic parallel algorithm for building hopsets, which we devise in this paper.

\end{abstract}

\end{titlepage}

\section{Introduction}

\subsection{Approximate Single-Source Shortest Paths} 
Consider a weighted undirected graph $G=(V,E,\omega)$ with\footnote{One can also extend our results to non-negative weights. To eliminate zero weight edges, one starts with contracting them. This can be done, e.g., by running the algorithm of Shiloach and Vishkin \cite{ShiloachV82}.} $\omega(e)> 0 $ for every $e\in E$, and a source vertex $s\in V$. We want to compute $(1+\epsilon)$-approximate shortest paths from $s$ to all other vertices (henceforth, $(1+\epsilon)$-SSSP problem), for an arbitrarily small constant $\epsilon >0$, in the parallel PRAM model of computation. This is one of the most basic and central graph algorithmic problems, and it has been intensively researched starting from late eighties. 
In particular, in their influential survey on parallel algorithms, Karp and Ramachandran \cite{KarpR90} list six fundamental graph-algorithmic problems for which there was no satisfactory solution known. The SSSP problem is one of them.

Early algorithms for this problem \cite{uy91,KleinS97} used randomization, and provided tradeoff between time $t$ and number of processors $p$ in which the product $t\cdot p$ is close to $n^3$. Improved tradeoffs were given by Shi and Spencer \cite{Spencer97,ss99}, whose randomized algorithm requires $\tilde{O}(t)$ time for a parameter $t, \ 1\leq t\leq n$, and $\tilde{O}(\frac{n^3}{t^2}+|E|)$ work. (The notation $\tilde{O}$ hides here a polylogarithmic dependence of $n$.)
Cohen \cite{Cohen97} devised a deterministic algorithm that achieves a similar tradeoff. Note however that for polylogarithmic time, the work complexity of her algorithm becomes $\tilde{O}(n^3)$. 

 Galil and Margalit \cite{GalilM97}, Alon et al. \cite{AlonGM97} and Zwick \cite{Zwick98} devised algorithms based on fast matrix multiplication for the $(1+\epsilon)$-approximate variant of the problem. These algorithms have polylogarithmic (parallel) time and work complexity $O(n^\omega)$, where $\omega<2.377\dots$ is the matrix multiplication exponent \cite{CoppersmithW87,Williams12,GallU18}.

In a major breakthrough, Cohen \cite{coh94} devised a \emph{randomized} $(1+\epsilon)$-approximate algorithm for this problem with polylogarithmic time $({\log n})^{O\left(\frac{\log {1}/{\rho}}{ \rho}\right)}$, and work complexity $\tilde{O}(|E|\cdot n^\rho)$, for any arbitrarily small constant parameters $\epsilon,\rho>0$. Her algorithm is based on \emph{hopsets}. See Section \ref{sec hop} for the definition of hopset.

\newcommand{\oeps}{(1+\epsilon)}
Improved constructions of hopsets, due to \cite{ElkinN17Hop,ElkinN19}, led to improved tradeoff for the problem: specifically, their algorithm has running time $({\log n})^{O(1/\rho)}$, and work $\tilde{O}(|E|\cdot n^\rho)$. Similarly to the algorithm of Cohen \cite{coh94}, the algorithms of \cite{ElkinN17Hop,ElkinN19} are randomized.

Most recently, using ideas from continuous optimization \cite{Sherman16,Madry13,BeckerKKL16}, Li \cite{Li19} and Andoni et al. \cite{AndoniSZ19} devised algorithms with $({\log n})^{O(1)}$ running time and $\tilde{O}(|E|)$ work for the $\oeps$-SSSP problem. Their algorithms are, however, randomized as well.

To summarize, the only known \emph{deterministic} PRAM algorithms for this central problem that use polylogarithmic time require at least $O(n^\omega)$ work 
\cite{Zwick02}. The question if there exists a \emph{deterministic} algorithm with polylogarithmic time and close to $|E|$ work complexity remained wide open, despite all this intensive research. 

In the current paper, we devise the first deterministic PRAM NC algorithm for building hopsets. As a consequence, we obtain a \emph{deterministic} polylogarithmic-time algorithm with slightly super-linear work for the $\oeps$-SSSP problem. Specifically, for any arbitrarily small constant parameter $\rho>0$, our algorithm has running time $({\log n})^{O(1/\rho)}$, and work complexity $\tilde{O}(|E|\cdot n^\rho)$.

Providing deterministic solutions for fundamental problems for which efficient randomized parallel algorithms are known is a central thread in Theoretical Computer Science. 
The importance of derandomizing PRAM algorithm was pointed out in a pioneering work by Karp and Pippenger \cite{KarpP83}. 
See also \cite{KarpW85,Luby86,KarloffR88,NisanW88,KarpR90,NaorN93,BergerRS94}.
In particular, derandomization of PRAM algorithms led to development of fundamental tools such as small-biased limited independence probability spaces \cite{NaorN93,AlonGHP92}. 

\subsection{Hopsets} \label{sec hop}

As was mentioned above, our algorithm for $\oeps$-SSSP is based on a novel deterministic PRAM NC algorithm for building hopsets. To our knowladge, our algorithm is the first NC algorithm for building hopsets with work smaller than $n^\omega$. 

Given a graph $G= (V,E,\omega)$ and a pair of parameters $\epsilon>0$ and $\beta =1,2,\dots,$ a graph $G' = (V,H,\omega')$ is said to be a \emph{$(1+\epsilon,\beta)$-hopset} for $G$ if for every pair of vertices $u,v\in V$, we have 
\begin{eqnarray}
d_G(u,v)\leq d_{G\cup G'}^{(\beta)}(u,v)\leq (1+\epsilon)d_G(u,v).
\end{eqnarray}
Here $ d_{G\cup G'}^{(\beta)}(u,v)$ stands for the \emph{$\beta$-bounded $u-v$ distance} in $G\cup G'$, i.e., the length of the shortest $u-v$ path in $G\cup G'$ that contains at most $\beta$ edges. The parameter $\beta$ is called the \emph{hopbound} of the hopset $G'$. 

Cohen \cite{coh94} devised a randomized PRAM algorithm that for any parameters $\epsilon,\rho >0$, and for any $n$-vertex, weighted graph $G= (V,E,\omega)$, builds a $\left(1+\epsilon,\beta_{Coh}=({\log n})^{O(\frac{1}{\rho}{\log \frac{1}{\rho}})}\right)$-hopset with $\tilde{O}(n^{1+\rho})$ edges, in $\tilde{O}(\beta_{Coh})$ time, and using $\tilde{O}(|E|\cdot n^\rho)$ work.
Her algorithm is based on \emph{pairwise covers}, and Cohen \cite{coh94} remarked that a deterministic NC algorithm that constructs the latter would suffice for derandomizing her hopset construction. 
However, more than quarter a century after Cohen's work \cite{coh94}, still no efficient parallel deterministic procedure for building these covers is known.

Elkin and Neiman \cite{ElkinN19} devised a different randomized algorithm with polylogarithmic depth for constructing hopsets. In their algorithm\footnote{Their algorithm also gave rise to first constructions of hopsets with constant hopbound $\beta$. However, this property is of little help for the \emph{single-source} shortest paths problem.}, $\beta_{EN}=({\log n})^{O(1/\rho)}$, improving $\beta_{Coh} = ({\log n})^{O(\frac{1}{\rho}{\log \frac{1}{\rho}})}$.
The algorithm of \cite{ElkinN19} is based on the \emph{superclustering-and-interconnection} approach, originated in the work of Elkin and Peleg \cite{ElkinP01} on near-additive spanners. 
In this approach, one identifies dense areas of the graph, i.e., areas that contain many clusters, and builds from them superclusters. The remaining clusters are then interconnected, i.e., shortest paths between them are inserted into the hopset. 

To implement their approach efficiently in PRAM, \cite{ElkinN19} used random sampling. Dense areas are more likely to be sampled, and as a result superclustered. In the current paper we show that the sampling step can be replaced by carefully constructed \emph{ruling sets}. 
For a pair of parameters $\alpha,\gamma\geq 1$, a set $U\subset V$ is said to be an \emph{$(\alpha,\gamma)$-ruling set} for a set $\widehat{U}\subseteq V$, if the following two conditions hold: 
\begin{enumerate}
	\item For every $u,v\in U$, $d_G(u,v)\geq \alpha$.
	\item For every $\widehat{u}\in \widehat{U}$, there exists a vertex $u\in U$ with $d_G(u,\widehat{u})\leq \gamma$. 
\end{enumerate} 

Ruling sets are long known to be a key derandomization tool in distributed algorithms \cite{GoldbergPS88,awerbuch1989network}.
As long as $\gamma =\Omega(\alpha\cdot {\log n})$, they can be efficiently constructed in parallel and distributed models \cite{GoldbergPS88}. We show that such ruling sets are sufficient for constructing hopsets. Note that these ruling sets are applicable for {\em unweighted} graphs, while hopsets are built for {\em weighted} graphs. Our algorithm builds appropriate unweighted cluster graphs, and constructs ruling sets for them.
As a result, we derandomize the result of \cite{ElkinN19}, and devise a deterministic NC algorithm that for any constant parameters $\rho,\epsilon> 0$, and any $n$-vertex graph $G= (V,E,\omega)$, constructs a $(1+\epsilon,\beta=\beta_{EN})$-hopset $H$ with $\tilde{O}(n^{1+\rho})$ edges, in time $\tilde{O}(\beta_{EN})= ({\log n})^{O(1/\rho)}$, and using work $\tilde{O}(|E|\cdot n^\rho)$.
Once the hopset is constructed, we conduct a Bellman-Ford exploration 
in the graph union the hopset $H$ to depth $\beta$. 
 This exploration requires additional time $\tilde{O}(\beta)$, and work complexity $\tilde{O}(|V|+|H|)$. Both are dominated by the respective complexities of constructing the hopset $H$. 

In addition to its main application to the $(1+\epsilon)$-SSSP problem, our deterministic algorithm for constructing hopsets has other applications. In particular, recently Elkin and Neiman \cite{EN20} devised a PRAM algorithm with polylogarithmic depth that
computes $(1+\epsilon)$-approximate shortest paths from all pairs in $S \times V$, for a set $S$ of at most $n^{0.313...}$ sources using $\tilde{O}(n^2)$ work. (They actually have a tradeoff between the number of sources and running time. See \cite{EN20} for more details.)
This algorithm relies on an efficient computation of a suitable hopset. In conjunction with the deterministic construction of hopsets that we devise in the current paper, the algorithm of \cite{EN20} can be made deterministic. Yet another possible application of our deterministic algorithm for constructing hopset can be found in \cite{EGN19}. That algorithm is for the same problem of computing paths for all pairs 
$(s,v) \in S \times V$ as above, but this time allowing both multiplicative error of $(1+\epsilon)$ and an additive error of $\beta \cdot W(s,v)$, where $\beta$ is a large constant that depends on $\epsilon$ and on some other parameters, but is independent of $n$, and $W(s,v)$ is the weight of the heaviest edge on a shortest $s-v$ path. It has polylogarithmic depth and work roughly $\tilde{O}(|E| + |S|\cdot n)$.
It has two randomized ingredients, one of which is an algorithm for constructing a hopset and another one is an algorithm for constructing a near-additive weighted spanner. Our algorithm in the current paper can be plugged there. Once the other ingredient (i.e., the construction of weighted near-additive spanner) is derandomized (and it is plausible that this can be accomplished via ideas similar to ones that we develop here), one would obtain a deterministic counterpart of the result of \cite{EGN19}. To summarize, our deterministic PRAM NC algorithm for building hopsets is valuable on its own right, and we believe that additional applications of it will be found in future. 

\subsection{Path-Reporting Hopsets}
Our additional contribution is in devising a \emph{path-reporting} variant of our hopsets. The basic variant of our hopset construction enables one to compute approximate distances, but not paths that implement these distances. 

Elkin and Neiman \cite{ElkinN18rout,ElkinN19} devised a different \textit{implicit} mechanism for making their hopset path-reporting. Their hopset enables, upon a query $u,v$, to retrieve a $\oeps$-approximate shortest path $P_{u,v}$ in $G$ between $u$ and $v$ in $O({\log n})$ PRAM time and $\tilde{O}(|P_{u,v}|)$ work. Their algorithm however cannot retrieve the entire $\oeps$-Shortest-Path-Tree (henceforth, SPT) within the desired resource bounds. (See the discussion at the beginning of Section 5.4 of \cite{ElkinN19}.)

Our new path-reporting mechanism enables us to compute a $(1+\epsilon)$-SPT in $G$ in polylogarithmic time and slightly super-linear work. 
This mechanism can be used also in conjunction with previous randomized algorithms for building hopsets \cite{coh94,ElkinN18rout,ElkinN19}. It simplifies the respective path-reporting variant of these hopsets, and more importantly, makes them \emph{explicit}. We believe that this mechanism is of independent interest.

We next shortly outline it. 
Our hopset is a union of logarithmically many single-scale hopsets\footnote{Informally, a {\em single-scale} hopset $H_k$ for a scale $k$ is a hopset that takes care of pairs of vertices $u,v$ such that $d_G(u,v) \in (2^k,2^{k+1}]$. See Section \ref{sec hopset const} for more details.}. Moreover, each edge $e= (u,v)$ in a scale-$k$ hopset $H_k$, for some $k\geq 2$, is implemented by at most $\beta = polylog(n)$ edges that either belong to the scale-$(k-1)$ hopset $H_{k-1}$, or to the original graph $G$. In our algorithm, we store for every edge $(u,v)\in H_k$ a path $P^{(k-1)}_{u,v}$ between $u$ and $v$ in $G\cup H_{k-1}$. 
This increases our storage by just a polylogarithmic factor.
When we need to retrieve paths (e.g., when building an approximate SPT), we conduct a \emph{peeling} process. First, we replace all hopset edges of the highest scale-$\lambda$ by edges of the original graph $G$ and of the scale-$(\lambda-1)$. Then we replace the latter by hopset edges of scale-$(\lambda-2)$, etc., up until all hopset edges are eliminated. This results in a polylogarithmic overhead (of $\tilde{O}(\beta)$) in the total work. 

This peeling process becomes far more elaborate when one combines it with the weight reduction of \cite{KleinS92, coh94,ElkinN19}. The latter reduction eliminates the dependency on the maximum edge weight from our results. From the technical viewpoint, it creates graphs $\mcg{1},\mcg{2},\dots,\mcg{\lambda}$, each of which is obtained from the original graph by contracting some edges, and grouping some sets of vertices into \emph{nodes}. 

To construct a $(1+\epsilon)$-SPT for the original graph $G$, our algorithm first computes a $(1+\epsilon)$-SPT in $G$ union the hopset, which in turn is the union of hopset $\mch{1},\mch{2},\dots,\mch{\lambda}$ of all these graphs $\mcg{1},\mcg{2},\dots,\mcg{\lambda}$. Each of the latter hopsets contains edges between the nodes, and also, for every node $\widehat{v}$ it contains (the so called "star") edges between its center $v$ and each other vertex $u\in \widehat{v}$. 
During the peeling process, we carefully replace the star edges by paths in $G$, and hopset edges $(\widehat{v}_1,\widehat{v}_2)$ between nodes $\widehat{v}_1$ and $\widehat{v}_2$ by paths between their respective centers $v_1$ and $v_2$ in $G$. In addition, we also find appropriate portal vertices $v_1'$ and $v_2'$ in $\widehat{v}_1$ and $\widehat{v}_2$, respectively, that enable us to integrate the path in our final $\oeps$-SPT for $G$, and connect the centers $v_1$ and $v_2$ to their respective portals.

\subsection{Related Work}

Improved randomized constructions of hopsets, which are based on Thorup-Zwick \cite{ThorupZ01,ThorupZ06} sampling hierarchy, were devised in \cite{ElkinN17Hop,HuangP19}. Existential lower bounds on hopsets were shown in \cite{AbboudBP17}. 

 Henzinger et al. \cite{HenzingerKN16} devised a distributed deterministic algorithm for building hopsets. However, the hopbound of their hopset is $2^{O(	\sqrt{{\log n}{\log {\log n}}})}$, and not at most polylogarithmic in $n$, as that of hopsets of \cite{coh94,ElkinN19,ElkinN17Hop} and of our hopset. 
As a result, any PRAM algorithm that would use their hopset is doomed to use at least this running time (of $2^{O(	\sqrt{{\log n}{\log {\log n}}})}$), while we achieve a polylogarithmic running time. 

Also, the current authors \cite{ElkinMatar} used ruling sets to derandomize the superclustering-and-interconnection approach of \cite{ElkinP01,ElkinN17spanners} in the context of near-additive spanners for unweighted graphs in the distributed \congestmo. The running time of the algorithm of \cite{ElkinMatar} is $O(n^\rho)$, for an arbitrarily small parameter $\rho>0$. Our current work follows the same approach, but does so in a far more complicated setting of hopsets for weighted graphs, and in a more restrictive PRAM model. 
Also, the running time of our current algorithm is polylogarithmic in $n$, as apposed to the polynomial time in \cite{ElkinMatar}. 

Finally, hopsets were found extremely useful for {\em dynamic} algorithms for approximate shortest paths problems \cite{Bernstein09,BernsteinR11,HenzingerKN16}. We believe that the techniques that we developed for our new deterministic algorithm for building hopsets will be useful in this context as well.

\subsection{Preliminaries}\label{sec pre}
Let $G=(V,E,\omega)$ be a weighted, undirected graph on $n$ vertices, where the minimal edge weight is $1$. If $(u,v)\notin E$, then $\omega(u,v) = \infty$. 
For convenience, we assume that each vertex $v\in V$ has a unique ID in the range $
\{ 0,1,\dots, n-1 \}$. We also assume that $n$ is a power of $2$. This assumption does not affect our ultimate result.

Throughout the algorithm, we construct clusters. Each cluster $C$ is centered around a designated center $r_C\in C$. The ID of the cluster $C$ is the ID of its center $r_C$. 
For a pair of clusters $C,C'$ in a graph $X$ 
let $d^{(h)}_X(C,C') = {\min \{ d^{(h)}_X(u,u') \ | \ u\in C \ and \ u'\in C' \} }$, where $h$ is a hopbound. 
If $h$ is not specified, then the hopbound is not bounded, i.e., $d_X(C,C') = d^{(\infty)}_X(C,C')$.

Denote by $\Lambda$ 
the aspect ratio of the graph, i.e., the ratio between the largest and the smallest distance in $G$. In the main part of this paper we will assume that $\Lambda = n^{O(1)}$. In Appendix \ref{sec reduc} we argue that using Klein-Sairam's weight reduction \cite{KleinS97}, one can get rid of this assumption.

Throughout the paper, when the logarithm base is unspecified, it is equal to $2$. 
For a pair of integers $i,j,$ where $ i\leq j$, the notation $ [i,j]$ stands for $\{i,i+1,\dots, j\}$.

\subsubsection{The Computation Model}
We consider the concurrent-read-exclusive-write (CREW) PRAM model (see, e.g., \cite{JaJa92}), in which all processors work in synchronous rounds.
In every round, each processor either writes to some memory cell associated with it, reads some memory cell or stays idle. Concurrent read is allowed, but only one processor can write to each cell in every round. To avoid access conflicts, vertices write on odd rounds and read on even rounds. The \textit{running time} of the algorithm is the number of rounds that the algorithm requires, also referred to as the \textit{depth} of the algorithm. The work of the algorithm is the number of read and write operations executed by all processors throughout the algorithm. A trivial upper bound on the work of the algorithm is the running time, multiplied by the number of processors.

We aim at using roughly $O(|E|\cdot n^\rho)$ processors, for some parameter $0<\rho< 1/2$. Every vertex and edge in the graph are simulated by $O(n^\rho)$ processors, and each processor is associated with $O(1)$ memory cells. Every cluster $C$ constructed by the algorithm is simulated by its center $r_C$. That is, the processors and memory of $r_C$ are utilized for the simulation of the cluster $C$. 
Throughout our algorithm, we add edges to a hopset $H$. Every edge that is added to the hopset is also allocated $O(n^\rho)$ processors. For every edge $(u,v)\in E\cup H$, the processors $p^{\langle u,v\rangle }_1,p^{\langle u,v\rangle }_1,\dots,p^{\langle u,v\rangle }_{n^\rho}$ and $p^{\langle v,u\rangle }_1,p^{\langle v,u\rangle}_1,\dots,p^{\langle v,u\rangle}_{n^\rho}$ simulate the (directed) edges $\langle u,v\rangle,\langle v,u\rangle$, respectively.

\subsection{Outline}
	Section \ref{sec hopset const} contains the details of our construction of hopsets. In Section \ref{sec analysis} we analyze the properties of the resulting hopset and the computational complexity of the algorithm, which are then summarized in Theorems \ref{theorem final hopset} and \ref{theorem compute dist}. 
	Section \ref{sec path-reporting} contains a path-reporting hopset construction. Its properties are summarized in Theorem \ref{theo peth reporting}.

	Appendix \ref{append explorations} contains a PRAM implementation of a BFS exploration in a virtual graph that is used by our algorithm. Appendix \ref{sec ruling} contains a PRAM implementation of the algorithm of \cite{awerbuch1989network,sew,KuhnMW18} for constructing ruling sets. In Appendices \ref{sec reduc} and \ref{append red path} we argue that the dependence of our results on the aspect ratio can be eliminated. These results are summarized in Theorem \ref{theo reduc} and \ref{theo path reduc}, respectively.
	
	Bibliography appears after the appendix. 
	
\section{The Algorithm}\label{sec hopset const}
In this section, we provide the details of the construction.
The  input for our algorithm is an undirected, weighted graph $G=(V,E,\omega)$, where $\omega(e)>0$ for all $e\in E$, and parameters $0<\epsilon<1/10$, $\kappa = 1,2,\dots$ and $0<\rho <1/2$. 
The parameter $\kappa$ governs the sparsity of the hopset and $\rho$ governs the work complexity. The hopbound parameter $\beta$ is a function of $n,\Lambda,\epsilon,\kappa,\rho$ (recall that $\Lambda$ stands for the aspect ratio of the graph), and is given by

\begin{equation}\label{eq beta def}
\beta =O\left( \frac{{\log {\Lambda}}{\log n} ({\log \kappa\rho} + 1/\rho) }{\epsilon} \right)^{\pramell}.
\end{equation}

Let $G=(V,E,\omega)$ be a weighted, undirected graph, with aspect ratio $\Lambda$. 
Let $k=0,1,\dots,\lceil{\log \Lambda}\rceil$. A set of edges $H_k$ with weights given by the weight function $\omega_{H_k}$ is said to be a \emph{$(1+\epsilon,\beta)$-hopset for the scale} $k$ if for every pair of vertices $u,v$ with $d_G(u,v)\in (2^k,2^{k+1}]$ we have that:
$$d_G(u,v) \leq d^{(\beta)}_{G_{k}} (u,v) \leq (1+\epsilon)d_G(u,v),$$
where $G_k= (V,E\cup H_k,\omega_k)$ and $\omega_k(u,v) = {\min \{\omega(u,v), \omega_{H_k}(u,v)\}}$ for every edge $(u,v)\in E\cup H_k$. 

The algorithm constructs a separate hopset $H_k$ for every scale $(2^0,2^1],(2^1,2^2],\dots,(2^{\lceil{\log \Lambda}\rceil-1},2^{\lceil{\log \Lambda}\rceil}]$. Note that for $k\leq \floor*{{\log {\beta}}} -1$, we can set $H_k=\emptyset$. This is since $2^{k+1}\leq \beta$, and thus for each pair of vertices $u,v\in V$ with $d_G(u,v)\leq 2^{k+1} $, the original graph $G$ already contains a shortest path between $u,v$ that uses at most $\beta$ edges. In other words, $d_G(u,v) = d^{(\beta)}_G(u,v)$. 
Denote $k_0 = \floor*{{\log {\beta}}} $ and $\lambda = \lceil {\log \Lambda}\rceil-1$. We construct a hopset $H_k$ for every $k\in\krange$.

During the construction of a hopset $H_k$, for some 
$k\geq k_0$, we will need to conduct explorations from certain vertices to depth $\delta \leq 2^{k+1}$. Therefore, we use the $(1+\epsilon_{k-1},\beta)$-hopset $H_{k-1}$ for the scale $(2^{k-1},2^k]$ in the construction of $H_k$. The value of $\epsilon_{k-1}$ will be determined in the sequel.

Instead of conducting explorations from the vertex $u\in V$ to depth $\delta$ in the original graph $G$, we conduct the explorations from $u$ to depth $(1+\epsilon_{k-1})\delta$ and $(2\beta+1)$ hops in the graph $G_{k-1}$.
The following lemma shows that the exploration executed in $G_{k-1}$ will reach all vertices that are within distance up to $\delta$ from $u$ in the original graph $G$.
 (This lemma is analogous to Lemma 3.9 in \cite{ElkinN19}.)
 We will later show that distances in $G_{k-1}$ are never shorter than those in the original graph $G$, i.e., that for every pair of vertices $u,v\in V$ we have $d_{G}(u,v)\leq d_{G_{k-1}}(u,v)$.

\begin{lemma}\label{lemma zeta}
	Let $u\in V$, and let $v$ be a vertex such that $d_G(u,v)\leq 2^{k+1}$. 
	Then, 
	$$d^{(2\beta+1)}_{G_{k-1}}(u,v) \leq (1+\epsilon_{k-1} )d_G(u,v).$$
\end{lemma}

\begin{proof}
	Let $\pi(u,v)=(u=u_0,u_1,\dots,u_t = v)$ be a shortest path in $G$ between $u,v$. Let $u_j$ be the last vertex on $\pi(u,v)$ such that $d_G(u,u_j)\leq 2^{k}$. Then, $d_G(u,u_{j+1})> 2^{k}$. See Figure \ref{fig gk} for an illustration. 
	Note that the length of $\pi(u,v)$ is at most 
	$2^{k+1}$. It follows that since $j$ is the maximal index such that $d_G(u,u_j)\leq 2^{k}$, we have $d_G(u_{j+1},v)\leq 2^{k}$. Hence we conclude that 
	$$ d^{(\beta)}_{G_{k-1}}(u,u_j)\leq (1+\epsilon_{k-1} )d_G(u,u_j) \qquad \text{and} \qquad
	d^{(\beta)}_{G_{k-1}}(u_{j+1},v)\leq (1+\epsilon_{k-1} )d_G(u_{j+1},v).$$
	Let $\pi(u,u_j)$ and $\pi(u_{j+1},v)$ be the shortest $\beta$-hop limited paths in $G_{k-1}$ from $u$ to $u_j$ and $u_{j+1}$ to $v$, respectively.
	\begin{figure}
		\centering
		\includegraphics[scale = 0.2]{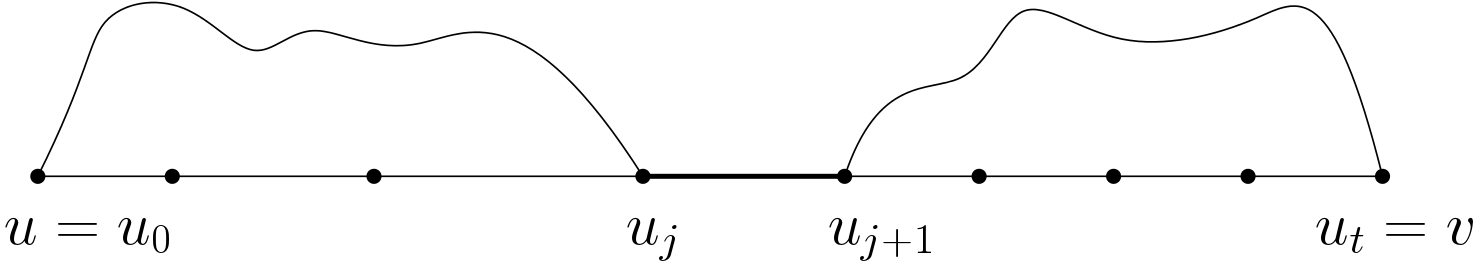}
		\caption{The path in $G_{k-1}$ from $u$ to $v$. The straight line represents a shortest path from $u$ to $v$ in the original graph $G$. The curved lines represent the paths in $G_{k-1}$ from $u$ to $u_j$ and $u_{j+1}$ to $v$ of length at most $(1+\epsilon_{k-1} )d_G(u,u_j)$ and $(1+\epsilon_{k-1} )d_G(u_{j+1},v)$, respectively. Each of them consists of at most $\beta$ edges. The thick line depicts the edge $(u_j,u_{j+1})$ from the graph $G$. By concatenating the path of $\beta$ edges from $u$ to $u_j$ with the edge $(u_j,u_{j+1})$ and the path of $\beta$ edges from $u_{j+1}$ to $v$, we obtain a path with at most $2\beta +1$ hops, of length at most $(1+\epsilon_{k-1} )d_G(u,v)$. 
		}
		\label{fig gk}
	\end{figure}
	In addition, the graph $G_{k-1}$ also contains the edge $(u_j,u_{j+1})$. Then, by concatenating the path $\pi(u,u_j)$ with the edge $(u_j,u_{j+1})$ and the path $\pi(u_{j+1},v)$, we obtain a path of at most $2\beta +1$ hops and length at most $(1+\epsilon_{k-1} )d_G(u,v)$.	
	 It follows that 
	$$d^{(2\beta+1)}_{G_{k-1}}(u,v) \leq (1+\epsilon_{k-1} ) d_G(u,v).$$
\end{proof}

\subsection{Constructing $H_k$} \label{sec constructing hk}
We are now ready to discuss the construction of a hopset $H_k$, for some 
$k\in\krange$. 
The algorithm works in phases. The input for each phase $i\in [0,\ell]$ is a collection of clusters $P_i$, a degree threshold parameter $deg_i$ and a distance threshold parameter $\delta_i$. The parameters $\ell,\{deg_i,\delta_i \ | \ i\in [0,\ell] \}$ will be specified in the sequel. For phase $0$, the input $P_0$ is the partition of $V$ into singleton clusters, i.e., $P_0= \{ \{ u\}\ |\ u\in V\}$. 

Our algorithm uses the \emph{superclustering-and-interconnection} approach (see, e.g., \cite{ElkinP01,ElkinN17spanners,ElkinN17Hop,ElkinN19,ElkinMatar}). Generally speaking,
in every phase $i$, we define clusters that have many other clusters in their vicinity as \textit{popular}.
Popular clusters are grouped into superclusters, which become the input for the next phase. 
Clusters that have not been superclustered in this phase are interconnected with the clusters in their vicinity. The set of superclusters is the input collection for the next phase. Intuitively, this approach allows us to defer work on dense areas of the graph to later phases of the algorithm. 

Recall that $\dgk(C,C') = {\min \{ \dgk(u,u') \ | \ u\in C \ and \ u'\in C' \} }$.
Formally, a pair of clusters $C,C'\in P_i$ are said to be \textit{neighbors} (or, equivalently, \emph{neighboring clusters}) if $\dgk(C,C')\leq \apdi$.
For a cluster $C\in P_i$, define
\begin{equation*}
\label{eq gamma}
	\Gamma(C) = \{ C'\ | \ \dgk(C,C')\leq \apdi, \ C\neq C' \}. 
\end{equation*}
A cluster $C\in P_i$ and its center are called \textit{popular} if $|\Gamma(C)|\geq deg_i$. Otherwise, they are called \textit{unpopular}.

Each phase $i$ is divided into two steps. In the \textit{superclustering step} of phase $i$, popular clusters are detected, and are grouped into superclusters. 
The set of new superclusters becomes the input of the next phase, i.e., it is the set $P_{i+1}$. 
Let $U_i$ be the set of clusters from $P_i$ that have not been superclustered in this phase. In the \textit{interconnection step} of phase $i$, clusters in $U_i$ are interconnected with their neighboring clusters that are also in $U_i$. For the last phase $\ell$, we will ensure that the number of clusters in $P_\ell$ is at most $deg_\ell$. Therefore, in this phase there will be no popular clusters, and thus the superclustering step will be skipped.

We now set the parameters $\{deg_i,\delta_i\ | \ i\in [0,\ell] \},\ell$.
The degree parameter $deg_i$ determines the number of edges added in the interconnection step of phase $i$. 
 In addition, the degree parameter determines how many clusters of $P_i$ are necessary to construct a supercluster in phase $i$, and thus determines also the number of phases of the algorithm. To achieve a polylogarithmic running time, the algorithm uses $\Theta(deg_i)$ processors to simulate each vertex and each edge in phase $i$. Since we aim at using $O(n^{\rho})$ processors to simulate each edge and vertex, we divide the phases into an exponential growth stage and a fixed growth stage\footnote{Recall that $\kappa$ and $\rho$ are input parameters of our algorithm. We assume $0<\rho <1/2$.}. Set $\ell = \pramell$. For the \textit{exponential growth stage}, that consists of phases $i\in [0,i_0 = \lfloor{\log \kappa\rho}\rfloor]$, we set $deg_i = \degi$. 
For the \textit{fixed growth stage}, that consists of phases
$i\in [i_1 =i_0+1 ,\ell ]$, we set $\deg_i = n^\rho$. 

For the distance parameter, for every $i\in [0,\ell]$ set $\delta_i = \partd $, where 
$\alpha = \epsilon^\ell \cdot 2^{k+1}$.

Let $C\in P_i$ be a cluster centered around a vertex $r_C\in C$. 
The radius of $C$ is defined to be
$Rad(C) = \max\{ d^{(i)}_{G_k}(r_C,v) \ | \ v\in C \}$. For the collection $P_i$, define 
\begin{equation}\label{eq rad def}
 	Rad(P_i) = \max\{Rad(C) \ | \ C\in P_i \} .
\end{equation}
Define recursively $R_0=0$, and $R_{i+1} = \rul+R_{i}$. We will show that for all $i\in [0,\ell]$, we have $Rad(P_i)\leq R_i$ (see Lemma \ref{lemma super path}). 
\subsubsection{Superclustering} \label{sec superclustering}
This section contains the details of the superclustering step for each phase $i\in [0,\ell-1]$. Recall that on phase $\ell$, the superclustering step is skipped.
Let $\tilde{G}_i = (P_i,\tilde{E})$ be an unweighted virtual graph where $\tilde{E}= \{ (C,C') \ | \ \dgk(C,C')\leq \apdi \}$. In other words, the graph $\tilde{G}_i$ is the graph on the set of supervertices (clusters) $P_i$, that contains an edge for every pair of neighboring clusters in $P_i$. Note that the edges of the virtual graph $\tilde{G}_i$ are not necessarily known to the processors.

The superclustering step of phase $i$ begins by detecting the popular clusters, i.e., the clusters that have at least $deg_i$ edges incident to them in the graph $\tilde{G}_i$.
This is done by a simulation of $deg_i+1$
 parallel BFS explorations executed from all supervertices (clusters) in $\tilde{G}_i$ to depth $1$ in the virtual graph $\tilde{G}_i$.
When the explorations terminate, each cluster $C\in P_i$ is associated with an array $m(C)$ that contains IDs of neighboring clusters of $C$. 
If the cluster $C$ has at least $deg_i$ neighboring clusters in $P_i$, this array will contain information regarding at least $deg_i$ of these clusters. Otherwise, the array $m(C)$ will contain the IDs and $(2\beta+1)$-hop bounded distance in $G_{k-1}$ to all the neighbors $C'$ of $C$. In other words, for every cluster $C'\in \Gamma(C)$, the cluster $C$ will know the ID of $C'$ as well as the distance $\dgk(C,C')$.

To implement these explorations, we use Algorithm \ref{alg parallel limited BF}, described in Appendix \ref{append explorations}. Its properties are summarized in Lemma \ref{lemma exporation v1}.

\textbf{Lemma \ref{lemma exporation v1}}\textit{
	Given a weighted undirected graph $G_{k-1} = (V,E\cup H_{k-1},\omega_{k-1})$ on $n$ vertices, 
	a set of clusters $P_i$, distance parameter $\apdi$, degree parameter $deg_i$ and a hopbound $2\beta+1$, Algorithm \ref{alg parallel limited BF} requires $O(\beta{\log n})$ time and $O((|E|+|H_{k-1}|)\cdot n^\rho)$ processors. When the algorithm terminates, every cluster $C\in P_i$ is associated with an array $m(C)$ such that:
 	}
	\begin{enumerate}
	 	\item\textit{If $C$ is \textit{popular}, then the first $deg_i+1$ cells of $m(C)$ contain information regarding $deg_i+1$ clusters from $P_i$. One of them is $C$, and the other are neighboring clusters of $C$.}
	 	
	 	\item \textit{ If $C$ is  \emph{unpopular} then $m(C)$ contains the identities and $(2\beta+1)$-hop bounded distances in $G_{k-1}$ between $C$ and all its neighboring clusters. In addition, the $(deg_i+1)$st  cell in $m(C)$ is empty. 
	 } 
	\end{enumerate}

When the algorithm terminates, each center $r_C$ of a cluster $C\in C$ checks weather the $(deg_i+1)$st cell in $m(C)$ is empty or not. If it is not empty, then $C$ is marked as popular. Otherwise, it is marked unpopular.

Let $W_i$ be the set of popular clusters. 
A subset $Q_i\subseteq W_i$ is selected to grow superclusters around. This is done in order to ensure that the number of superclusters formed in phase $i$ is significantly smaller than the number of clusters in $P_i$. 

To select the set $Q_i$, a $(3,2{\log n})$-\textit{ruling set} for $W_i$ with respect to (henceforth, w.r.t.) the graph $\tilde{G}_i$, is computed using the algorithm of \cite{awerbuch1989network,sew,KuhnMW18}. The implementation of the algorithm in the PRAM CREW model and its analysis are given in Appendix \ref{sec ruling}. 
The following corollary summarizes the properties of the set $Q_i$.

\textbf{Corollary \ref{coro ruling}}\textit{
	Given a weighted undirected graph $G_{k-1} = (V,E\cup H_{k-1},\omega_{k-1})$ on $n$ vertices, 
	a set of clusters $P_i$, a distance parameter $\apdi$, and a hopbound $2\beta+1$, Algorithm \ref{alg pram ruling} uses $O(|E|+|H_{k-1}|)$ processors and $O(\beta{\log^{2}n})$ time, and returns a $(3,2{\log n})$-ruling set $Q_i$ for the set $W_i$ w.r.t. the graph $\tilde{G}_i = (P_i,\tilde{E})$, where $\tilde{E} = \{ (C,C') \ | \ \dgk(C,C')\leq \apdi \}$.
}

To cover all popular clusters, a BFS exploration to depth $2{\log n}$ in the graph $\tilde{G}_i$ is simulated from all clusters in $Q_i$, using Algorithm \ref{alg parallel limited BF}. The details of the execution of this exploration, as well as the pseudocode of Algorithm \ref{alg parallel limited BF}, are given in Appendix \ref{append explorations}. 

Each cluster $C\in P_i$ that is detected by the exploration that has originated in a cluster $C'\in Q_i$ becomes superclustered into the supercluster $\widehat{C}'$ formed around $C'$. 
The center $r$ of the cluster $C$ adds to the hopset $H_k$ an edge to the center $r'$ of the cluster ${C}'$, and sets $\omega_{H_k}(r,r') = 2(\apdi+2R_i){\log n}$. The new edge is allocated $O(n^\rho)$ new processors. 
We show in Lemma \ref{lemma sup not shorter} 
that superclustering edges do not shorten distances with respect to the original graph $G$. 

For each new supercluster $\widehat{C}'$ formed around a cluster $C'\in Q_i$, its center is set to be the center of the cluster $C'$. The supercluster $\widehat{C}'$ contains all the vertices that belonged to $C'$, as well as all the vertices that belonged to other clusters that have been superclustered into $\widehat{C}'$.
The input collection for the next phase $P_{i+1}$ is set to be the set of new superclusters.

The following lemmas summarize the properties of the superclusters formed in phase $i$. 
Recall that the $i$-hop bounded radius of a cluster $C\in P_i$ in the graph $G_k$ is defined by \cref{eq rad def} to be
$Rad(C) = \max\{ d^{(i)}_{G_k}(r_C,v) \ | \ v\in C \}$, and that the 
radius of $P_i$ is set to be $Rad(P_i) = \max\{Rad(C) \ | \ C\in P_i \} $. Define recursively $\sigma_0 = 0$ and
 $\sigma_{i+1} = (4{\log n}+1)\sigma_i + 2 (2\beta+1){\log n}$, for all $i\in [1,\ell]$. Note that the sequence $(\sigma_i \ | \ i=0,1,\dots)$ is monotonically increasing.
For every cluster $C\in P_i$ centered around a vertex $r_C$, define the $\sigma_i$-bounded radius of the cluster $C$ in the graph $G_{k-1}$ to be 
\begin{equation*}
Rad^{(\sigma_i)}_{k-1}(C) = \max\{ d^{(\sigma_i)}_{G_{k-1}}(r_C,v) \ | \ v\in C \}.
\end{equation*}
For the collection $P_i$, define $Rad^{(\sigma_i)}_{k-1}(P_i) = \max\{Rad^{(\sigma_i)}_{k-1}(C) \ | \ C\in P_i \} $.
Also recall that $R_0=0$ and $R_{i+1} = \rul+R_i$ for all $i\in [0,\ell]$. 
The following lemma provides upper bounds for $Rad(P_i)$ and $Rad^{(\sigma_i)}_{k-1}(P_i)$.
\begin{lemma}\label{lemma super path}
	For $i\in [0,\ell]$,	we have $Rad(P_i) \leq R_i$ and also $ Rad^{(\sigma_i)}_{k-1}(P_i)\leq R_i$. 
\end{lemma}

\begin{proof}
	The proof is by induction on the index of the phase $i$. 
	For $i=0$, all clusters in $P_i$ are singletons and so $Rad(P_i)= Rad^{(\sigma_i)}_{k-1}(P_i) = 0 $, and also, by definition, $ R_0=0$. 
	
	Assume that the claim holds for some $i\in [0,\ell-1]$ and prove it for $i+1$. Let $\widehat{C}$ be a cluster in $P_{i+1}$, centered around a vertex $r_C$. Observe that $\widehat{C}$ was constructed around a cluster $C\in P_i$ (also centered around $r_C$) during phase $i$. 
	Let $u\in \widehat{C}$. The analysis splits into two cases.

	\textbf{Case 1:} The vertex $u$ is in $C$. 
	By the induction hypothesis, we have $d_{G_{k}}^{(i)}(r_C,u)\leq R_{i}$ and $d_{G_{k-1}}^{(\sigma_i)}(r_C,u)\leq R_i$. Therefore, 
	$d_{G_{k}}^{(i+1)}(r_C,u)\leq R_i\leq R_{i+1}$, and $d_{G_{k-1}}^{(\sigma_{i+1})}(r_C,u)
	\leq d_{G_{k-1}}^{(\sigma_{i})}(r_C,u)
	\leq R_i\leq R_{i+1}$. Thus the claim holds. 
	
	\begin{figure}
		\centering
		\includegraphics[scale = 0.25]{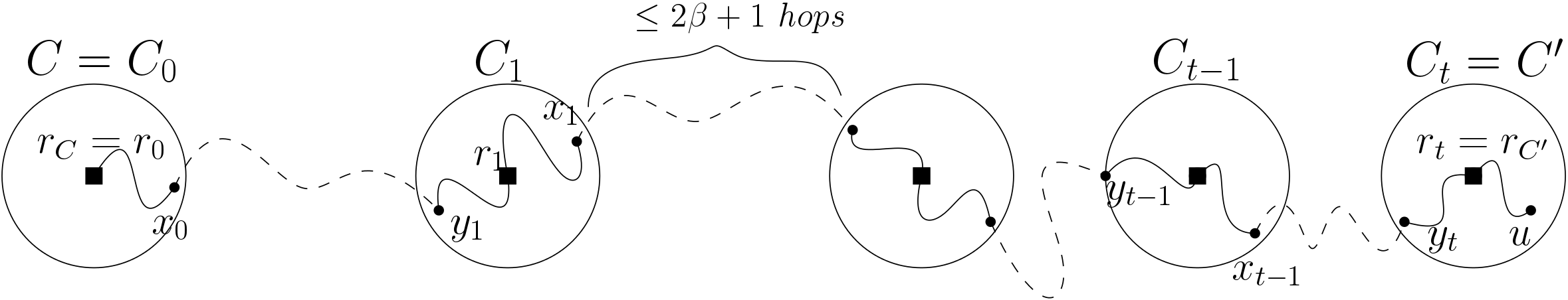}
		\caption{The path in $G_{k-1}$ from $r_C$ to $u$. The white circles represent the clusters along the path $P$. The squares represent cluster centers. The solid curved lines represent paths in $G_{k-1}$ from a center of a cluster to a vertex in the cluster. The dashed lines depict paths of at most $2\beta+1$ hops in $G_{k-1}$ between neighboring clusters.}
		\label{fig:superpath}
	\end{figure}

	\textbf{Case 2:} The vertex $u$ belongs to a cluster $C'\in P_i$, where $C'\neq C$. Let $r_{C'}$ be the center of the cluster $C'$. Since $u\in \widehat{C}$, we have that
	the cluster $C'$ was detected by the exploration to depth $2{\log n}$ in the graph $\tilde{G}_i$ that originated in $C$. It follows that an edge $\cedge$ was added to $H_k$. The weight of the edge is $2(\apdi+2R_i){\log n}$. By the induction hypothesis, we have that $d^{(i)}_{G_k}(r_{C'},u)\leq R_i$. It follows that $d^{(i+1)}_{G_k}(r_C,u) \leq 2(\apdi+2R_i){\log n}+R_i= R_{i+1}$, and therefore we have that $Rad(\widehat{C})\leq R_{i+1}$.
	
	Let $P= \{ C=C_0,C_1,\dots,C_t = C' \} $ be the shortest path from $C$ to $C'$ in the (unweighted) virtual graph $\tilde{G}_i$. Observe that $t\leq 2{\log n}$. 
	 For every $j\in [0,t-1]$, let $x_j\in C_j$ and $y_{j+1}\in C_{j+1}$ 
	 be the vertices such that $\dgk(x_j,y_{j+1})$ is minimal. See Figure \ref{fig:superpath} for an illustration. By definition of $\tilde{G}_i$, for every $j\in [0,t-1]$, we have $\dgk(x_j,y_{j+1})\leq \apdi$. 	 
	 For every cluster $C_j$ along the path $P$, denote by $r_j$ its center. By the induction hypothesis, for every $j\in [0,t-1]$ we have $d^{(\sigma_i)}_{G_{k-1}}(r_j,x_j) \leq R_i$, and for every $j\in [1,t]$ we have $d^{(\sigma_i)}_{G_{k-1}}(r_j,y_j)\leq R_i$. 
	 It follows that 
	 \begin{equation*}
	 \begin{array}{lclclclclclclc}\label{eq centers path}
	
	d^{(t(2\sigma_i+2\beta+1))}_{G_{k-1}}\cedge 
	& \leq &
	\sum_{j=0}^{t-1} 
	d^{(2\sigma_i+2\beta+1)}_{G_{k-1}}(r_j,r_{j+1}) 
\\	& \leq &
	\sum_{j=0}^{t-1} 
	d^{(\sigma_i)}_{G_{k-1}}(r_j,x_j)+
	d^{\tb}_{G_{k-1}}(x_j,y_{j+1}) + 
	d^{(\sigma_i)}_{G_{k-1}}(y_{j+1},r_{j+1}) 
\\	& \leq &
	t(\apdi+2R_i).
	 \end{array}
	 \end{equation*}
	 Recall that $t\leq 2{\log n}$. Observe that $2(2\sigma_i+2\beta+1){\log n} = \sigma_{i+1}-\sigma_i$. Therefore, we conclude that 
	 \begin{equation}\label{eq centers path fin}
	 \begin{array}{lclclclclclclc}
	 
	 d^{(\sigma_{i+1}-\sigma_i)}_{G_{k-1}}\cedge 
	 & \leq &
	 2(\apdi+2R_i){\log n}.
	 \end{array}
	 \end{equation}
	 
	 By the induction hypothesis, we also have $d^{(\sigma_i)}_{G_{k-1}}(r_{C'},u)\leq R_i$. It follows that

	 \begin{equation*}\label{eq u path fin}
	 \begin{array}{lclclclclclclc}
	 d^{(\sigma_{i+1})}_{G_{k-1}}(r_C,u)
	 & \leq & d^{(\sigma_{i+1}-\sigma_i)}_{G_{k-1}}\cedge + d^{(\sigma_i)}_{G_{k-1}}(r_t,u) 
	 & \leq &
	 2(\apdi+2R_i){\log n}+R_i &=& R_{i+1}.
	 \end{array}
	 \end{equation*}
	 Hence, $Rad^{(\sigma_i)}_{k-1}(\widehat{C})\leq R_{i+1}$.
	 
\end{proof}

From \cref{eq centers path fin} we derive the following lemma.
Recall that a superclustering edge $\cedge$ is added to $H_k$ in phase $i$ with weight $2(\apdi+2R_i){\log n}$. Thus, the next lemma shows that superclustering edges never shorten distances w.r.t. $G_{k-1}$.
\begin{lemma}
	\label{lemma sup not shorter}
	For every superclustering edge $(r_C,r_{C'})$ added to the hopset $H_k$ during phase $i$, we have $$d^{(\sigma_{i+1}-\sigma_{i})}_{G_{k-1}}(r_C,r_{C'})\leq 2(2R_i+\apdi){\log n}.$$
\end{lemma}

\textit{Remark:} For the correctness of the basic variant of our hopset algorithm, it is important that there exists a path in $G_{k-1}$ between $r_C$ and $r_{C'}$ with weight at most $R_i$. The number of hops in this path is not important. However, in Section \ref{sec path-reporting}, we use this bound to show that one can modify the current algorithm to obtain a path-reporting hopset.

Next, we show that all popular clusters in phase $i$ are superclustered in phase $i$. 
\begin{lemma}\label{lemma popular are clustered}
	Let $C\in P_i$ be a popular cluster. Then $C$ is superclustered into a cluster of $P_{i+1}$. 
\end{lemma}

\begin{proof}
	Since $C$ is popular, we know that $C\in W_i$. Recall that by Corollary \ref{coro ruling}, the set $Q_i$ is a $(3,2{\log n})$-ruling set for the set $W_i$ w.r.t. the graph $\tilde{G}_i$. Then, there exists a cluster $C'\in Q_i$ such that $d_{\tilde{G}_i}(C,C')\leq 2{\log n}$. Therefore, the BFS exploration executed in the graph $\tilde{G}_i$ from all clusters in $Q_i$ to depth $2{\log n}$ discovers the cluster $C$, and it becomes superclustered into a cluster of phase $i+1$. 
\end{proof}

Next, we show that the each supercluster formed in phase $i$ contains many clusters of $P_i$.

\begin{lemma}\label{lemma cluster size}
	For every $i\in [0,\ell-1]$, each supercluster constructed in phase $i$ contains at least $deg_i+1$ clusters from $P_i$. 
\end{lemma}

\begin{proof}
	Let $\widehat{C}$ be a cluster constructed in phase $i$, around a cluster $C\in Q_i$. Since $C\in Q_i\subseteq W_i$, we have that $|\Gamma(C)|\geq deg_i$. In addition, since $Q_i$ is $3$-separated w.r.t. the graph $\tilde{G}_i$, for every cluster $C'\in \Gamma(C)$, the cluster $C$ is the only neighbor of $C'$ in the set $Q_i$. Therefore, all clusters in $\Gamma(C)$ are detected by the BFS exploration originated in $C$. 
	Thus the supercluster formed around $C$ contains all clusters of $\Gamma(C)$, and the cluster $C$ itself. It follows that the new supercluster $\widehat{C}$ contains at least $deg_i+1$ clusters of $P_i$. 
\end{proof}

Finally, we prove that in phase $\ell$ there are no popular clusters, and so we do not need to create superclusters. The following two lemmas provide upper bounds to the size of $P_i$ in the exponential and the fixed growth stages, respectively. 
\begin{lemma}\label{lemma uperbound exp}
	For $i\in [0,i_0+1]$, the size of $P_i$ is at most 
	$ n^{1- \frac{2^i-1}{\kappa}} $. 
\end{lemma}
\begin{proof}
	The proof is by induction on the index of the phase $i$. For $i=0$, the claim is trivial since $|P_0|=n$. Assume that the claim holds for some $i\leq i_0$, and prove it for $i+1$. 
	
	By Lemma \ref{lemma cluster size}, each supercluster formed in phase $i$ contains at least $deg_i= \degi$ clusters from $P_i$. Together with the induction hypothesis, we have 
	
	$$|P_{i+1}|\quad \leq \quad 
	|P_i|\cdot deg_i^{-1} \quad \leq \quad
	n^{1- \frac{2^i-1}{\kappa}}\cdot n^{-\frac{2^i}{\kappa}} \quad = \quad n^{1- \frac{2^{i+1}-1}{\kappa}} .$$
\end{proof}

\begin{lemma}\label{lemma uperbound fixed}
	For $i\in [i_0+1, \ell]$, the size of $P_i$ is at most 
	$ n^{1 +\frac{1}{\kappa}- (i-i_0)\rho} $. 
\end{lemma}
\begin{proof}
	The proof is by induction on the index of the phase $i$. For $i=i_0+1 = \floor*{\log {\kappa\rho}}+1$, by Lemma \ref{lemma uperbound exp} we have that 
	$$|P_{i_0+1}| 
	\quad \leq \quad 
	n^{1- \frac{2^{i_0+1}-1}{\kappa}}
	\quad \leq \quad 
	n^{1- \frac{ {\kappa\rho}-1}{\kappa}}
	\quad = \quad
	n^{1+\frac{1}{\kappa} - \rho},$$ 
	and so the claim holds.
	Assume that the claim holds for some $i\in [i_0+1,\ell-1]$, and prove it for $i+1$. 
	
	By Lemma \ref{lemma cluster size}, we have that each supercluster formed in phase $i$ contains at least $deg_i = n^\rho$ clusters from $P_i$. Together with the induction hypothesis, we have 
	
	$$|P_{i+1}|
	\quad \leq \quad 
	|P_i|\cdot deg_i^{-1} 
	\quad \leq \quad
	n^{1 +\frac{1}{\kappa}- (i-i_0)\rho} \cdot n^{-\rho} 
	\quad = \quad 
	n^{1 +\frac{1}{\kappa}- (i+1-i_0)\rho}. $$
\end{proof}
Recall that $\ell = \pramell$. Also note that since $\rho<1/2$ we have that $\ell >i_0$, and so $deg_\ell = n^\rho$. 
By Lemma \ref{lemma uperbound fixed}, there are no popular clusters in phase $\ell$ since 

\begin{equation}\label{eq pl size}
|P_\ell | 
\quad \leq \quad
n^{1 +\frac{1}{\kappa}- (i_0 +\lceil
	\frac{\kappa+1}{\kappa\rho}\rceil -1-i_0)\rho} 
\quad \leq \quad
n^{\rho} 
\quad = \quad deg_\ell.
\end{equation}

\subsubsection{Interconnection}\label{sec intercon} In this section, we provide the technical details of the interconnection step for each phase $i\in [0,\ell]$. 

Recall that $U_i$ is the set of clusters from $P_i$ that have not been superclustered in this phase. In the interconnection step, 
every cluster $C\in U_i$ is interconnected with all clusters $C'\in \Gamma(C)\cap U_i$. Specifically, for each cluster $C\in U_i$, its center $r_C$ adds edges to all centers of clusters in $\Gamma(C)\cap U_i$. 
By Lemma \ref{lemma popular are clustered}, we have that all popular clusters have been superclustered in this phase. Thus, all clusters in $U_i$ are not popular.
By Lemma \ref{lemma exporation v1}, we have that for every cluster $C\in U_i$, its center $r_C$ maintains an array $m(C)$ such that for every neighboring cluster $C'\in \Gamma(C)$, the array $m(C)$ contains the ID of $C'$ and the $(2\beta+1)$-hop-bounded distance in $G_{k-1}$ between $C$ and $C'$. 
Then, $r_C$ only needs to learn which of the neighbors of $C$ are in $U_i$. Each center $r_{C'}$ of a cluster $C'\in U_i$ writes to its memory that $C'$ belongs to $U_i$. Then, each center $r_C$ of a cluster $C\in U_i$ checks which of its neighboring clusters are also in $U_i$.
Recall that $r_C$ has $O(n^\rho )$ processors associated with it. Thus, checking which clusters $C'\in m(C)$ belong to $U_i$ can be performed in $O(1)$ time. 
 The center $r_C$ of the cluster $C$ adds to $H_k$ an interconnection edge of weight $\dgk(C,C')+2R_i$ to every center of a cluster $C'\in \gmk(C)\cap U_i$. Every new interconnection edge is assigned $O(n^\rho)$ new processors.

In phase $\ell$, we set $U_\ell = P_{\ell}$, and the superclustering step is skipped. Thus the clusters do not possess the required information regarding their neighbors. Therefore, we execute $|P_\ell|$ parallel BFS explorations from all clusters in $P_\ell$ in the virtual graph $\tilde{G}_\ell$ to depth $1$ using Algorithm \ref{alg parallel limited BF} from Appendix \ref{append explorations}. By \cref{eq pl size}, we have that $|P_\ell|\leq n^\rho$. Hence, by 
Lemma \ref{lemma exporation v1}
these explorations use $O((|E|+|H_{k-1}|)n^\rho)$ processors and require $O(\beta{\log n})$ time.

Since  $|P_\ell|\leq n^\rho$, every cluster $C\in P_\ell$ has at most $n^\rho-1$ neighboring clusters in $\tilde{G}_{\ell}$, and there are no popular clusters in $P_\ell$. As a result, by Lemma \ref{lemma exporation v1} we have that every cluster $C\in P_\ell$ maintains an array $m(C)$ that contains the IDs and $(2\beta+1)$-hop bounded distance in $G_{k-1}$ between $C$ and all its neighboring clusters. Each center $r_C$ of a cluster $C$ adds to $H_k$ an interconnection edge of weight $\dgk(C,C')+2R_i$ to every center of a cluster $C'\in \gmk(C)\cap P_\ell$. Every new interconnection edge is assigned $O(n^\rho)$ new processors. 

The following lemmas summarize the properties of the interconnection step.

\begin{figure}
	\centering
	\includegraphics[scale=0.12]{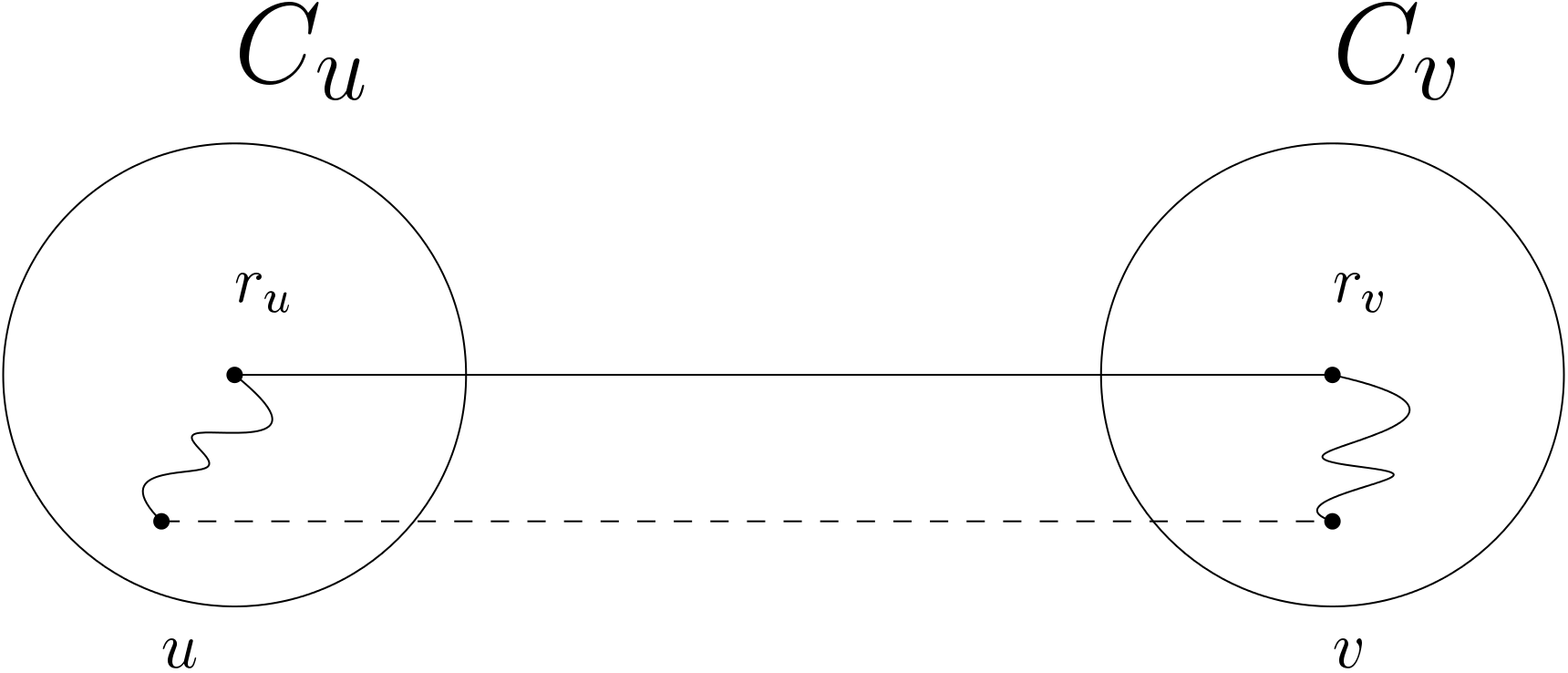}
	\caption{The path from $u$ to $v$ in $G_k$. The dashed line depicts the path from $u$ to $v$ in the original graph $G$. The curved lines represent paths of length at most $R_i$ and at most $i$ hops from $u,v$ to their respective cluster centers $r_u,r_v$. The straight line depicts the interconnection hopset edge between the cluster centers $r_u,r_v$. }
	\label{fig utovintercon}
\end{figure}

\begin{lemma}\label{lemma intercon path}
	Let $u,v$ be a pair of vertices with $d_G(u,v)\leq \alpha\cdot (1/\epsilon)^i$, such that $u\in C_u,\ v\in C_v$ and $C_u,C_v\in U_i$. Then, 
	$$d^{(2i+1)}_{G_k}(u,v) \leq (1+\epsilon_{k-1} )d_G(u,v)+4R_i.$$
\end{lemma}

\begin{proof}
	Denote $r_u,r_v$ the centers of the clusters $C_u,C_v$, respectively. See Figure \ref{fig utovintercon} for an illustration. By definition, we have $d_G(C_u,C_v)\leq d_G(u,v)\leq \alpha\cdot (1/\epsilon)^i=\delta_i$. Recall that $\alpha = \epsilon^\ell \cdot 2^{k+1}$, and $\epsilon<1$. Thus $d_G(C_u,C_v) \leq 2^{k+1}$. By Lemma \ref{lemma zeta} we have $\dgk(C_u,C_v)\leq (1+\epsilon_{k-1})d_G(C_u,C_v)$, 
	and so $\dgk(C_u,C_v)\leq (1+\epsilon_{k-1}) d_G(u,v)\leq \apdi$. Then, since $C_u,C_v\in U_i$, in the interconnection step of phase $i$ the centers $r_u,r_v$ add the edge $(r_u,r_v)$, with weight 
	$$\dgk(C_u,C_v)+2R_i\leq (1+\epsilon_{k-1} )d_G(C_u,C_v)+2R_i\leq (1+\epsilon_{k-1} )d_G(u,v)+2R_i.$$
	
	In addition, by Lemma \ref{lemma super path}, we have that $d_{G_k}^{(i)}(u,r_u)\leq R_i$ and $d_{G_k}^{(i)}(v,r_v)\leq R_i$. It follows that $$d_{G_k}^{(2i+1)}(u,v)\leq (1+\epsilon_{k-1} )d_G(u,v)+4R_i.$$
\end{proof}
Next, we show that interconnection edges added to the hopset $H_k$ do not shorten distances with respect to the graph $G_{k-1}$.

\begin{lemma}
	\label{lemma intercon not short}
	For every interconnection edge $(r_C,r_{C'})$ added to the hopset $H_k$ in phase $i\in [0,\ell]$, it holds that $d^{(2\sigma_i+\tb)}_{G_{k-1}}(r_C,r_{C'})\leq \dgk(C,C')+ 2R_i$.
\end{lemma}

\begin{proof}
	Consider an interconnection edge $\cedge$ added to the hopset $H_k$ in some phase $i\in [0,\ell]$, between a pair of clusters $C,C'$, where $r_C\in C$ and $r_{C'}\in C'$. This edge was added to the hopset with weight $\dgk(C,C')+2R_i$. 
	Let $u\in C$ and $u'\in C'$ be the vertices such that $\dgk(C,C')= \dgk(u,u')$.
	By Lemma \ref{lemma super path}, we have that $d^{(\sigma_i)}_{G_{k-1}}(r_C,u)\leq R_i$, and also 
	$d^{(\sigma_i)}_{G_{k-1}}(r_{C'},u')\leq R_i$. 
	It follows that 
	\begin{equation*}
	\label{eq intercon bound hops}
		d^{(2\sigma_i+\tb)}_{G_{k-1}}(r_C,r_{C'}) \ \leq \ 
		d^{(\sigma_i)}_{G_{k-1}}(r_C,u) + 
		\dgk(u,u')+
		d^{(\sigma_i)}_{G_{k-1}}(u',r_{C'})
		\ \leq \ 
		\dgk(C,C')+2R_i.
	\end{equation*}
\end{proof}

For every $i\in [0,\ell]$, define $U^{(i)} = \bigcup _{j = 0}^i U_j$. The next lemma shows that 
for every vertex $v$ there exists an index $i\in [0,\ell]$ such that $v$ belongs to a cluster that joins the set $U_i$. For notational purposes, define $U_{-1} = \emptyset$, and $U^{(i)} = \bigcup_{j=-1}^i U_i$ for all $i\in [0,\ell]$. 

\begin{lemma}
	\label{lemma ui partition}
	For every index $i\in [0,\ell]$, the set $P_{i} \cup U^{(i-1)}$ is a partition of $V$. 	
\end{lemma}

\begin{proof}
	The proof is by induction on the index of the phase $i$. For $i=0$, the claim is trivial since $P_{0}$ is a partition of $V$ into singleton clusters.
	
	Assume the claim holds for some index $i\in [0,\ell-1]$. Let $v\in V$. By the induction hypothesis, $v$ belongs to a cluster $C\in P_{i} \cup U^{(i-1)}$. 
	
	If $C\in U^{(i-1)}$, then, by definition, $C\in U^{(i)} $. Note that in this case, the cluster $C$ will not join a supercluster in phase $i$, and so the cluster $C$ is the only cluster containing the vertex $v$ in the set $U^{(i)}\cup P_{i+1}$.

	If $C\in P_i$, by the induction hypothesis, we know that there is no cluster $C'\in U^{(i-1)}$ that contains $v$. In phase $i$, the cluster $C$ has either been superclustered into a supercluster of $P_{i+1}$, or it has joined $U_i$. In any case, $C\in P_{i+1} \cup U^{(i)}$. Since every vertex $v\in V$ belongs to a single cluster $C\in P_{i+1} \cup U^{(i)}$, we have that $P_{i+1} \cup U^{(i)}$ is a partition of $V$. 
\end{proof}

Recall that in the last phase $\ell$, we skip the superclustering step and set $U_\ell =P_\ell$. Thus, Lemma \ref{lemma ui partition} implies that $U^{(\ell)}$ is a partition of $V$.

\section{Analysis of the Construction}\label{sec analysis}
In this section we analyze the size and stretch that the hopset $H= \bigcup_{k =k_0}^{\lambda}H_k$ guarantees, as well as the complexity of the construction.

\subsection{Analysis of the Size}

Let $k\in\krange$. We begin by analyzing the number of edges added to $H_k$ by every phase $i$ of the algorithm. 
Recall that in phase $i$, only centers of clusters in $P_i$ add edges to the hopset.

Consider a cluster $C\in P_i$. During phase $i$, 
if the cluster $C$ is selected to the ruling set $Q_i$ it will grow a supercluster of $P_{i+1}$. In this case, the center of $C$ will not add to $H_k$ any edges in phase $i$ of the algorithm. 
If $C\notin Q_i$, but it is detected by the exploration originated from a cluster in $Q_i$, it will join a supercluster. In this case, in phase $i$ the center of $C$ will add a single superclustering edge to $H_k$. 
If the cluster $C\notin Q_i$, and it is not detected by the exploration from $Q_i$, then $C$ joins the set $U_i$. In this case, the center $r_C$ of $C$ will add interconnection edges to all centers of clusters in $\gmk(C)$. By Lemma \ref{lemma popular are clustered}, we have that all popular clusters of phase $i$ have been superclustered in phase $i$. It follows that all clusters in $U_i$ are not popular. Thus, the number of interconnection edges added to the hopset $H_k$ by $r_C$ is at most $deg_i$.

Therefore, the number of superclustering edges added to the hopset $H_k$ in phase $i\in [0,\ell-1]$ is exactly 
$	|P_i|-|Q_i|-|U_i|.$
The number of interconnection edges added to to the hopset $H_k$ in phase $i\in [0,\ell]$ is at most $|U_i|\cdot deg_i.$ Note that by \cref{eq pl size}, we have that the size of $P_\ell$ is at most $deg_\ell$. Thus this upper bound also holds for phase $\ell$.

It follows that the number of edges added to $H_k$ in phase $i$ is at most 
\begin{equation}\label{eq hki edges}
	|P_i|-|Q_i|-|U_i|+ |U_i|\cdot deg_i = |P_i|-|Q_i| + |U_i|\cdot (deg_i-1).
\end{equation}

Next, we use the size of $P_{i+1}$ to bound the size of $U_i$. By Lemma \ref{lemma cluster size}, every cluster of $P_{i+1}$ contains at least $deg_i+1$ clusters of $P_i$. Therefore, we have
\begin{equation}\label{eq bound ui with pi1}
	|U_i|\leq |P_i| - |P_{i+1}|\cdot (deg_i+1).
\end{equation}

Recall that $|Q_i| = |P_{i+1}|$. By \cref{eq hki edges,eq bound ui with pi1} we have that the number of edges added to the hopset $H_k$ in phase $i$ is at most
\begin{equation*}\label{eq hki edges without ui}
\begin{array}{lclclclclcl}
	|P_i|-|Q_i| + |U_i|\cdot (deg_i-1) 
&\leq &
 	|P_i|-|P_{i+1}| + (|P_i| - |P_{i+1}|\cdot (deg_i+1))\cdot (deg_i-1) 
\\&=&
	|P_i|-|P_{i+1}| + |P_i|\cdot (deg_i-1)
	 - |P_{i+1}|\cdot (deg^2_i-1)
\\&=&
|P_i|\cdot deg_i
- |P_{i+1}|\cdot deg^2_i
.
\end{array}
\end{equation*}

Recall that in phase $\ell$ we do not form superclusters. Therefore, the number of edges added to the hopset by phase $\ell$ is at most $|P_\ell|\cdot (|P_\ell|-1)$. By \cref{eq pl size}, we have that $|P_\ell|\leq deg_\ell$, and so in phase $\ell$ we add at most $|P_\ell|\cdot deg_\ell$ edges to $H_k$. 
Hence, the number of edges added to $H_k$ by all phases of the algorithm is bounded by 
\begin{equation}\label{eq hk final}
\begin{array}{lclclclclcl}
|H_k| &\leq &
 \sum_{i=0}^{\ell-1}
(|P_i|\cdot deg_i
- |P_{i+1}|\cdot deg^2_i)
+|P_\ell|\cdot deg_\ell 
\\&=&
|P_0|\cdot deg_0 +\sum_{i=1}^{\ell}
(|P_i|\cdot (deg_i - deg^2_{i-1}))
\\&=&
\nfrac +\sum_{i=1}^{\ell}
(|P_i|\cdot (deg_i - deg^2_{i-1})).
\end{array}
\end{equation}

Recall that in the exponential growth stage, we have $deg_i = \degi$. Thus for $i\in [1,i_0]$ we have $deg_i = deg^2_{i-1}$. Recall also that in the fixed growth stage, we have $deg_i = n^\rho$. Note that $deg_{i_0} = n^{\frac{2^{\floor*{\log \kappa\rho}}}{\kappa}}\geq n^\frac{\rho}{2}$. It follows that for $i\in [i_0+1,\ell]$ we also have $deg_i \leq deg^2_{i-1}$. By \cref{eq hk final}, the size of $H_k$ satisfies: 
\begin{equation}\label{eq hk concluding}
\begin{array}{lclclclclcl}
|H_k| &\leq &
\nfrac +\sum_{i=1}^{\ell}
(|P_i|\cdot (deg_i - deg^2_{i-1}))
&\leq &
\nfrac .
\end{array}
\end{equation}

Recall that the ultimate hopset $H$ is the union of the hopsets $H_k$ constructed for each scale $k\in \krange$.
By \cref{eq hk concluding}, we have that the size of the hopset $H$ is at most: 
\begin{equation}\label{eq hk size}
\begin{array}{lclclclclcl}
|H| & = & \lceil{\log \Lambda}\rceil\cdot \nfrac .
\end{array}
\end{equation}

\subsection{Analysis of the Running Time and Work}

In this section, we analyze the computational complexity of the construction. 
Observe that for the construction of the hopset $H_k$, the algorithm uses the graph $G_{k-1}$, that contains all edged of $E$ and all edges of the hopset $H_{k-1}$. Edges of the hopsets $H_{k-2},H_{k-3},\dots$ are not used explicitly in the construction of $H_k$. Recall that for every scale $k\in \krange$, by \cref{eq hk concluding} we have $|H_k|\leq \nfrac$. 
Therefore, in the construction of the hopset $H_k$, the algorithm uses processors that simulate edges of the original graph $G$, as well as $O(\nfrac\cdot n^\rho)$ processors that simulate edges of the hopset $H_{k-1}$.
 When the construction of $H_k$ terminates, the processors that simulate the edges of $H_{k-1}$ are reallocated to simulate the edges of $H_{k+1}$. Hence, in every time point, our algorithm uses up to $O((|E|+\nfrac)\cdot n^\rho)$ processors to simulate edges and vertices.

\begin{lemma}\label{lemma rt} 
	Our algorithm uses $O((|E|+\nfrac)\cdot n^\rho)$ processors, and its running time is $O(({\log \Lambda})({\log \kappa\rho}+1/\rho)\beta{\log ^{2 }n})$. 
\end{lemma}

\begin{proof}
	
	We begin by analyzing the running time of a single phase $i$ in the construction of $H_k$, for some $k\in \krange$.

	\textbf{Superclustering:} By Lemma \ref{lemma exporation v1}, detecting popular clusters using $O((|E|+\nfrac)\cdot n^\rho)$ processors requires $O(\beta {\log n})$ time. Constructing a $(3,2{\log n})$ ruling set for the graph $\tilde{G}_i$ using $O(|E|+\nfrac)$ processors requires additional $O(\beta{\log^{2} n})$ time, by Corollary \ref{coro ruling}. 
	Forming superclusters requires executing a BFS exploration to depth $2{\log n}$ in the graph $\tilde{G}_i$. By Corollary \ref{coro bfs}, this requires $O(|E|+\nfrac)$ processors and $O(\beta{\log^2 n})$ time. Therefore, overall, the execution of each superclustering step requires $O((|E|+\nfrac)\cdot n^\rho)$ processors and $O(\beta\cdot {\log^2 n })$ time.

	\textbf{Interconnection:}
	In phases $i\in [0,\ell-1]$, by Lemma \ref{lemma exporation v1}, each center $r_C$ of a cluster $C\in U_i$ knows the centers of the (up to) $n^\rho$ clusters it needs to connect to, and the distances to them. Thus, for these phases, the interconnection step requires $O(n^{1+\rho})$ processors and $O(1)$ time. 
	
	For the concluding phase $\ell$, Algorithm \ref{alg parallel limited BF} is executed as a part of the interconnection step. By Lemma \ref{lemma exporation v1}, this exploration requires $O((|E|+\nfrac)\cdot n^\rho)$ processors and $O(\beta{\log n})$ time. 
	
	It follows that the overall running time of a single phase $i\in [0,\ell]$ is $O(\beta{\log ^2 n })$. Recall that the algorithm constructs hopsets for at most $\lceil{\log \Lambda}\rceil$ scales, and in each scale it executes $\ell+1$ phases.
	
	Hence, the algorithm uses $O((|E|+\nfrac)\cdot n^\rho)$ processors and requires $O(({\log \Lambda})({\log \kappa\rho}+1/\rho)\beta{\log ^{2 }n})$ time.
\end{proof}

\subsection{Analysis of the Stretch}

 We begin by providing an upper bound on $R_i$. Recall that
 by definition, we have 
 $R_0 = 0$, and 
 $$R_{i+1} = \rul +R_i= 2(1+\epsilon_{k-1} )\alpha\cdot\epsi{\log n}+(4{\log n}+1)R_i.$$ 

 \newcommand{\al}{2(1+\epsilon_{k-1} )\alpha{\log n}}

 \newcommand{\be}{(4{\log n}+1)}

 \begin{lemma}
 	\label{lemma bound ri}
 	For all $i\in [0,\ell]$, we have $R_i = \al\cdot \sum_{j=0}^{i-1}\eps{j}\cdot \be^{i-1-j}$.
 \end{lemma}

\begin{proof}
	The proof is by induction on the index of the phase $i$. 
	For $i=0$, both sides of the equation are equal to $0$. 

	Assume that the claim holds for some index $i\in [0,\ell-1]$, and prove it for $i+1$. By definition and the induction hypothesis we have 	
	\begin{equation*}
	\begin{array}{lclclclclclclc}
	R_{i+1} &=&2(1+\epsilon_{k-1} )\alpha\cdot\epsi{\log n}+(4{\log n}+1)R_i \\
		
		&=&2(1+\epsilon_{k-1} )\alpha\cdot\epsi{\log n}+(4{\log n}+1)\cdot \al\cdot \sum_{j=0}^{i-1}\eps{j}\cdot \be^{i-1-j} \\

	&=& \al\cdot \sum_{j=0}^{i}\eps{j}\cdot \be^{i-j} .
	
	\end{array}
	\end{equation*}
	
\end{proof}

 Next, we provide an explicit bound on $R_i$. Assume that $\epsilon <\frac{1}{2\be}$. This assumption will not affect our final result (see Section \ref{sec rescale} for details). 
 For all $i\in [0,\ell]$,

 \begin{equation}\label{eq explicit ri}
 \begin{array}{lclclclclclclc}
	 R_{i} 
	 &=& \al\cdot \sum_{j=0}^{i-1}\eps{j}\cdot \be^{i-1-j}\\
	 &=& \al\cdot \be^{i-1} \sum_{j=0}^{i-1}\eps{j}\cdot {\be^{-j}} \\
	 
	 &\leq& \al\cdot \be^{i-1} \left[\frac{
	 \left(\frac{1}{\epsilon\be}\right)^i}{\frac{1}{\epsilon\be} -1}\right] \\
 
 	 &\leq& \al\cdot \be^{i-1} \cdot \left(\frac{1}{\epsilon\be}\right)^i
 	 \cdot 
 	\frac{\epsilon\be}{1-\epsilon\be} \\
 	
 	&\leq& 4\cdot (1+\epsilon_{k-1} )\alpha{\log n}\cdot \eps{i-1}
.
 \end{array}
 \end{equation}
For the sake of brevity, we denote $c(n) = 4\cdot (1+\epsilon_{k-1} ){\log n}$.

We are now ready to analyze the stretch and hopbound guarantees that the hopset $H_k$ provides for pairs of vertices $u,v\in V$ with $d_G(u,v)\in (2^{k},2^{k+1}]$. 
Recall that $\alpha = \epsilon^\ell \cdot 2^{k+1}$. Define $h_0 =1$, and $h_i = (1/\epsilon+2)\cdot (h_{i-1}+1)+2i+1$. 

\begin{lemma}\label{lemma stretch}
	Let $u,v$ be a pair of vertices with $d_G(u,v) \leq \partd $, such that all vertices of a shortest path $\pi (u,v)$ in $G$ between them are clustered in $U^{(i)}$, for some $i\leq \ell$. Then it holds that 
	$$d_{G}(u,v)\ \leq\ d^{(h_i)}_{G_k}(u,v)
	\ \leq \ (1+\epsilon_{k-1} +5\cdot c(n)(i-1)\epsilon)\cdot d_G(u,v) +\ta{i-1}.$$
\end{lemma}
 
\begin{proof}
	
	For the left-hand side of the equation, observe that Lemmas \ref{lemma sup not shorter}, \ref{lemma intercon not short} imply that for any $k\in \krange$, distances in $G_{k}$ are never shorter than distances in $G_{k-1}$. Since $H_{k_0-1} = \emptyset$, we conclude that for every pair of vertices we have 
	\begin{equation*}
	\label{eq not short}
	d_G(u,v)= d_{G_{k_0-1}}(u,v) \leq d_{G_k}(u,v).
	\end{equation*}
	
	We continue with the proof of the right-hand side of the equation. 
	The proof is by induction on the index $i$. 
	For $i=0$, consider a pair of vertices $u,v$ such that $\{u \},\{v \}\in U_0$ and also  $d_G(u,v)\leq \partd = \delta_0$.
	By Lemma \ref{lemma zeta}, we have that $\dgk(u,v)\leq (1+\epsilon_{k-1})d_G(u,v)\leq (1+\epsilon_{k-1})\delta_0$. Since $\{u \},\{v \}\in U_0$, the edge $(u,v)$ was added to $H_k$ with weight $\dgk(u,v)\leq (1+\epsilon_{k-1})d_G(u,v)$ in the interconnection step of phase $0$.
	Thus, $d^{(1)}_{G_k}(u,v) \leq (1+\epsilon_{k-1} )d_G(u,v)$. In addition, since $d_G(u,v)\leq \alpha$, we have 
	\begin{equation*}\begin{array}{llllclclclc}
	d^{(1)}_{G_k}(u,v)
	 &\leq& (1+\epsilon_{k-1} )d_G(u,v) \\
	&\leq & (1+\epsilon_{k-1} )d_G(u,v) +5\cdot c(n)\cdot \epsilon (\alpha-d_G(u,v)) \\
	
	&= & (1+\epsilon_{k-1} -5\cdot c(n)\cdot \epsilon)d_G(u,v) +5\cdot c(n)\cdot \epsilon\alpha .
	\end{array}
	\end{equation*}
	Hence the claim holds for the base case. 
	
	Let $i\in [0,\ell-1]$. 
	Assume that the claim holds for all $i'\in [0,i]$, and prove it for $i+1$.

	Consider first a pair of vertices $x,y$ such that there exists a shortest path between them $\pi(x,y)$, and all vertices on this path are $U^{(i)}$-clustered. (We do not require $d_G(x,y)$ to be bounded.) We divide this path into segments, $L_1,\dots,L_q$, each of length no more than $\delta_i$. For convenience, we imagine that the vertices of $\pi(x,y)$ appear from left to right, where $x$ is the leftmost vertex and $y$ is the rightmost vertex. For a segment $L_j$, denote by $a_j,b_j$ its leftmost and rightmost vertices, respectively. For the first segment $L_1$, set $a_1 = x$. Given a left endpoint $a_j$ of a segment $L_j$, the right endpoint $b_j$ of the segment is set to be the farthest vertex to the right of $a_j$, that is within distance at most $\delta_i$ of $a_j$ in $G$. Observe that it is possible that $a_j=b_j$, if the closest vertex to the right of $a_j$ is within distance more than $\delta_i$ from it, or if $a_j = y$. Given a right endpoint $b_j\neq y $, the next left endpoint $a_{j+1}$ is the closest vertex to $b_j$ from the right. See Figure \ref{fig xyPath} for an illustration. 
		
	\begin{figure}
		\begin{center}
			\includegraphics[scale=0.3]{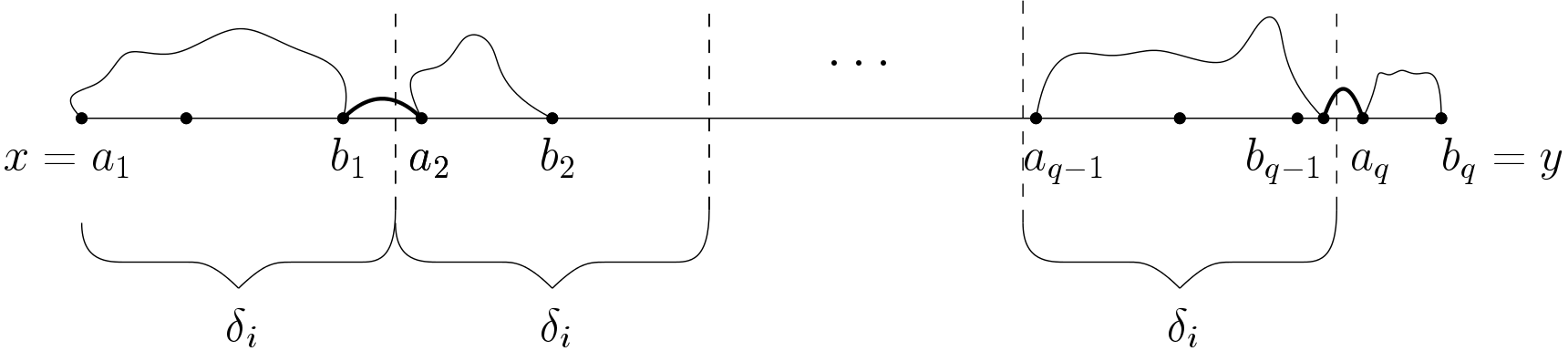}
			\caption{The bounded hop path between $x,y$ in $G_{k}$. In the figure, the straight line depicts the shortest path $\pi(x,y)$ in $G$ from $x$ to $y$. The dashed vertical lines represent the partition of $\pi(x,y)$ into segments. The thick curved lines represent edges of the original graph $G$, and the curved lines represents bounded hop paths from the first to last vertex in each segment.}
			\label{fig xyPath}
		\end{center}
	\end{figure}
	
	Observe that for the last segment $L_q$, we have that $b_q = y$. 
	Note that for every $i\in [1,q-1]$ we have $d_G(a_j,a_{j+1}) > \delta_i$. It follows that the number of segments $q$ is at most: 
	\begin{equation}\label{eq bound q}
	q \quad\leq\quad \left\lceil \frac{d_G(x,y)}{\delta_i}\right\rceil 
	\end{equation}
	
	Also note that for every $j\in [1,q]$, we have $d_G(a_j,b_j) \leq \delta_i$. Thus the induction hypothesis is applicable to all these segments, i.e., $$d^{(h_{i})}_{G_k}(a_j,b_j) \leq (1+\epsilon_{k-1} +5\cdot c(n) \cdot (i-1)\epsilon)\cdot d_G(a_j,b_j)+\ta{i-1}$$ 
	
	For each $j\in [1,q]$, let $\pi'(a_j,b_j)$ be the shortest $h_i$-bounded $a_j-b_j$ path in $G_k$.
	Let $\pi'(x,y)$ be a path in $G_k$ constructed in the following manner. 
	For $j\in [1,q]$, if $a_j \neq b_j$, we replace the segment between $a_j,b_j$ in
	$\pi(x,y)$ with the path $\pi'(a_j,b_j)$ of at most $h_i$ hops and at most $(1+\epsilon_{k-1} +5\cdot c(n)(i-1)\epsilon)\cdot d_G(a_j,b_j)+\ta{i-1}$ length. In addition, the path $\pi'(x,y)$ contains all edges $b_j,a_{j+1}$ from $G$, for $j\in [1,q-1]$, i.e., all edges between segments. Therefore, the
	 length of the new path $\pi'(x,y)$ is at most $(1+\epsilon_{k-1} +5\cdot c(n)(i-1)\epsilon)\cdot d_G(x,y)+ q \cdot \ta{i-1} $, and the number of hops in $\pi'(x,y)$ is at most $qh_{i} + q-1$. Recall that $\delta_i =\partd$. By \cref{eq bound q}, we have 
	
\begin{equation}\label{eq xy dist hop}
\begin{array}{lllll}
		d^{\left(\left\lceil \frac{d_G(x,y)}{\delta_i}\right\rceil\cdot (h_i+1)-1\right)}_{G_k}(x,y)
		&	\leq& 
			(1+\epsilon_{k-1} +5\cdot c(n)(i-1)\epsilon)\cdot d_G(x,y)+ \quad \left\lceil \frac{d_G(x,y)}{\delta_i}\right\rceil \cdot \ta{i-1} \\
			
		&	\leq& 
		(1+\epsilon_{k-1} +5\cdot c(n)(i-1)\epsilon+ 
		\frac{\ta{i-1}}{\partd}
		)\cdot d_G(x,y)+ \ta{i-1} \\
		
		&	= & 
		(1+\epsilon_{k-1} +5\cdot c(n)\cdot i\cdot \epsilon
		)\cdot d_G(x,y)+ \ta{i-1} \\
\end{array}
\end{equation}

Now, consider a pair $u,v\in V$ with $d_G(u,v) \leq \partdd[i+1]=\delta_{i+1}$, such that all vertices of a shortest path $\pi (u,v)$ in $G$ between them are clustered in $U^{(i+1)}$.

	Let $w_1,w_2$ be the first and last $U_{i+1}$-clustered vertices of the path $\pi(u,v)$. Let $C_1, C_{2}\in P_{i+1}$ be the clusters of $w_1,w_2$, respectively. See Figure \ref{fig uvPath} for an illustration. 
	Note that since $w_1,w_2$ lie on a path of length at most $\partdd[i+1]$, we have $d_G(w_1,w_2)\leq \partdd[i+1]$. By Lemma \ref{lemma intercon path}, we have that 
	
	\begin{equation}\label{eq w1w2}
	d_{G_k}^{(2(i+1)+1)}(w_1,w_2)\leq (1+\epsilon_{k-1} )d_G(w_1,w_2)+4R_{i+1}.
		\end{equation} 

\begin{figure}
	\begin{center}
		\includegraphics[scale=0.25]{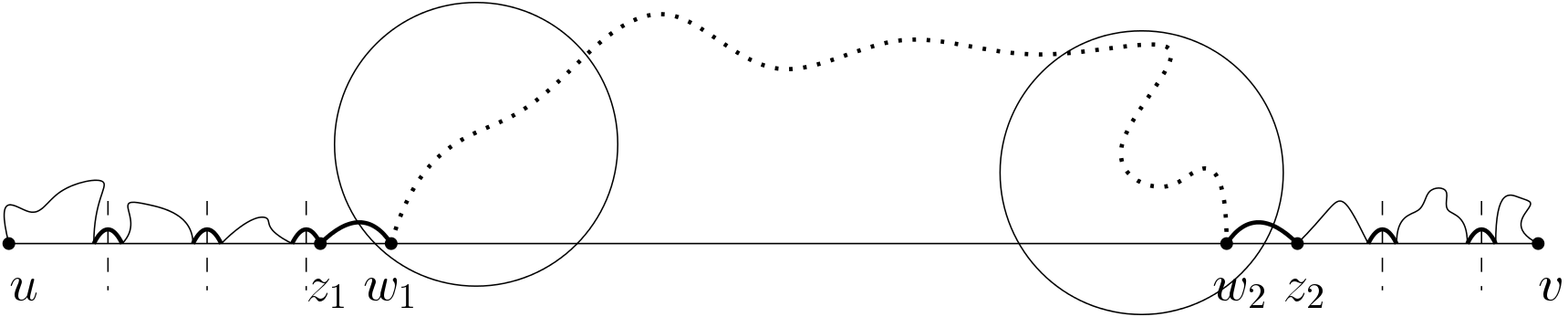}
		\caption{The path between $u,v$ in $G_{k}$. In the figure, the straight line depicts the shortest path $\pi(u,v)$ in $G$ from $u$ to $v$. The vertices $w_1,w_2$ are the first and last $U_i$-clustered vertices, and $z_1,z_2$ are their neighbors on the subpaths of $\pi(u,v)$ from $u$ to $w_1$ and from $w_2$ to $v$, respectively. The curved dotted line depicts the path from $w_1$ to $w_2$ obtained by Lemma \ref{lemma intercon path}. The solid thick lines depict edges from the original graph $G$, and the solid curved lines depict paths obtained by the induction hypothesis.} 
		\label{fig uvPath}
	\end{center}
\end{figure}
	
	Let $z_1,z_2$ be the neighbors of the vertices $w_1,w_2$ on the subpaths of $\pi(u,v)$ from $u$ to $w_1$ and from $w_2$ to $v$, respectively. Note that $z_1,z_2$ are both $U^{(i)}$-clustered.
	
	The path $\pi'(u,v)$ is constructed in the following manner. By \cref{eq xy dist hop}, there is are paths in $G_k$ from $u$ to $z_1$ and from $z_2$ to $v$, each with a multiplicative stretch of $(1+\epsilon_{k-1} +5\cdot c(n)\cdot \epsilon\cdot i)$ and an additive error of $\ta{i}$, using up to $p_1 =\left(\left\lceil \frac{d_G(u,z_1)}{\alpha\cdot (1/\epsilon)^{i}}\right\rceil\cdot (h_{i}+1)-1\right) $ and $p_2 = \left(\left\lceil \frac{d_G(z_2,v)}{\alpha\cdot (1/\epsilon)^{i}}\right\rceil\cdot (h_{i}+1)-1\right)$ hops, respectively. By \cref{eq w1w2}, there is a path from $w_1$ to $w_2$ in $G_{k}$ of length at most $(1+\epsilon_{k-1} )\cdot d_G(w_1,w_2)+4R_{i+1}$, using up to $2(i+1)+1$ hops. In addition, we use the edges $(z_1,w_1)$ and $(w_2,z_2)$ of the original graph $G$. Since the edges $(z_1,w_1),(w_2,z_2)$ lie on a shortest path, we know that $d_G(z_1,w_1)= \omega(z_1,w_1)$, and $d_G(w_2,z_2) = \omega (w_2,z_2)$. 
	It follows that 
	
	\begin{equation}\label{eq dist uv}
	\begin{array}{cllllclclclc}
	
	d^{(p_1+p_1+2i+5)}_{G_k}(u,v) &\leq &
	d^{(p_1)}_{G_k}(u,z_1)+
	d^{(1)}_{G_k}(z_1,w_1)+
	d^{(2i+3)}_{G_k}(w_1,w_2)+
	d^{(1)}_{G_k}(w_2,z_2)+
	d^{(p_1)}_{G_k}(z_2,v)\\
	&\leq &

	(1+\epsilon_{k-1} + 5\cdot c(n) \cdot i\cdot \epsilon)\cdot d_G(u,z_1)
	+ \ta{i-1} +
	\omega(z_1,w_1)\\&&+
	(1+\epsilon_{k-1} )\cdot d_G(w_1,w_2)+4R_{i+1}+
	\omega(w_2,z_2)\\&&+	
	(1+\epsilon_{k-1} + 5\cdot c(n) \cdot i\cdot \epsilon)\cdot d_G(z_2,v)
	+ \ta{i-1} \\
	
	&\leq &
	(1+\epsilon_{k-1} + 5\cdot c(n) \cdot i\cdot \epsilon)\cdot d_G(u,v) + 4R_{i+1}+
	10\cdot \alpha\cdot c(n) \cdot (1/\epsilon)^{i-1}.
	\end{array}
	\end{equation}
	
	Recall that by \cref{eq explicit ri} we have that 
	$R_{i+1}\leq c(n)\cdot \alpha\cdot (1/\epsilon)^{i}$, and also that $\epsilon \leq 1/10$ (in fact, we even require $\epsilon < \frac{1}{2(4{\log n}+1)}$). 
	Therefore, we have
	\begin{equation}\label{eq beta}
	\begin{array}{llllll}
	
		4R_{i+1}+
		10\cdot \alpha\cdot c(n) \cdot (1/\epsilon)^{i-1}&\leq&
		4c(n)\cdot\alpha\cdot (1/\epsilon)^{i}+
		10\cdot \alpha\cdot c(n) \cdot (1/\epsilon)^{i-1}\\&=&
		c(n)\cdot \alpha \cdot (1/\epsilon)^{i}(4+10\epsilon)
\\ &\leq&
	5c(n)\cdot \alpha \cdot (1/\epsilon)^{i}.

	\end{array}		
	\end{equation}

The number of hops in the path $\pi'(u,v)$ satisfies:

	\begin{equation}\label{eq dist hop}
	\begin{array}{lllclclclclc}
	p_1+p_1+2i+3&\leq &
	\left(\left\lceil \frac{d_G(u,z_1)}{\alpha\cdot (1/\epsilon)^{i}}\right\rceil
	+
	\left\lceil \frac{d_G(z_2,v)}{\alpha\cdot (1/\epsilon)^{i}}\right\rceil\right)\cdot (h_{i}+1)-2+2(i+1)+3\\ 
	&\leq &(1/\epsilon+2)\cdot (h_{i}+1)+2(i+1)+1\qquad=\qquad h_{i+1}.\\ 
	\end{array}
	\end{equation}

	For the last inequality, recall that $d_G(u,z_1)+d_G(z_2,v)\leq d_G(u,v)\leq \delta_{i+1}=\alpha\cdot \eps{i+1}$. 
	It follows by \cref{eq dist uv,eq beta,eq dist hop}, that the distance between $u,v$ in $G_k$ satisfies:
	
	\begin{equation*}
		d^{h_{i+1}}_{G_k}(u,v) \leq 
		(1+\epsilon_{k-1} + 5\cdot c(n) \cdot i\cdot \epsilon)\cdot d_G(u,v) + 
		5c(n)\cdot \alpha \cdot (1/\epsilon)^{i}.		
	\end{equation*}
\end{proof}

We now provide an upper bound on $h_\ell$.
\begin{lemma}
	\label{lemma hi bound}
	The recursive equation $h_0 = 1$, and $h_i = (1/\epsilon+2)\cdot (h_{i-1}+1)+2i+1$ solves to 
	$$h_i \leq (1/\epsilon+5)^{i}.$$	
\end{lemma}
\begin{proof}
	The proof is by induction on the index $i$. For $i=0$, both sides of the equation are equal to $1$, and so the claim holds. 
	
	Assume that the claim holds for some index $i\in [0,\ell-1]$, and prove it for $i+1$. 
	
	By definition, we have 
	\begin{equation*}
	\begin{array}{lllclclclclc}
	h_{i+1} &=&(1/\epsilon+2)\cdot (h_{i}+1)+2(i+1)+1

	 &=&
	(1/\epsilon+2)\cdot h_{i}+1/\epsilon+2i+5.
	\end{array}
	\end{equation*}
	Note that $ 1/\epsilon+2i+5 \leq h_{i}+ 2$.
	This is since 
	\begin{equation*}
	\begin{array}{lllclclclclc}
	h_{i} & = & (1/\epsilon+2)\cdot (h_{i-1}+1)+2i+1\\
	& = & 1/\epsilon\cdot h_{i-1}+2h_{i-1}+1/\epsilon+2i+3
	\\
	& \geq & 1/\epsilon+2i+3
	\end{array}
	\end{equation*}
	
	Together with the induction hypothesis, it follows that we have 
	\begin{equation*}
	\begin{array}{lllclclclclc}
	h_{i+1} &=&
	(1/\epsilon+2)\cdot h_{i}+1/\epsilon+2(i+1)+3 &\leq& (1/\epsilon+3)h_i +2 &\leq& (1/\epsilon+5)^{i+1}. 
	\end{array}
	\end{equation*}

\end{proof}
It follows that the hopbound of the hopset $H$ is 
\begin{equation}
\label{eq betaell}
 h_\ell= (1/\epsilon+5)^{\ell}.
\end{equation}


\subsection{Rescaling}\label{sec rescale}
In this section, we rescale $\epsilon$ to obtain our ultimate stretch guarantee.

Set now $\epsilon' = 20({\log n})\epsilon (\ell+1)$. 
The condition $\epsilon <\frac{1}{2\be}$ now translates to the condition $\epsilon'=O(\ell)$, which we replace by a stronger 
condition
$\epsilon'<1$. 

Recall by \cref{eq betaell} the hopbound of the hopset $H$ is $(1/\epsilon+5)^{\ell}$. The hopbound is now rescaled to be

\begin{equation*}
\label{eq betaell rescaled}
h_{\ell} = (1/\epsilon+5)^{\ell} = 
\left( \frac{20(\ell+1){\log n}}{\epsilon'} +5 \right)^\ell 
= O\left( \frac{\ell{\log n}}{\epsilon'} \right)^\ell.
\end{equation*}

Next, we complete the analysis of the stretch guarantee.

\begin{corollary} 
	\label{coro stretch1}
	Let $u,v$ be a pair of vertices with $d_G(u,v) \in (2^k,2^{k+1}]$. Then it holds that 
	$$d^{(h_{\ell})}_{G_k}(u,v)\leq (1+\epsilon_{k-1})(1+\epsilon')d_G(u,v).$$
\end{corollary}

\begin{proof}
	Recall that $\alpha = 2^{k+1}\epsilon^\ell$, and that $c(n) = 4(1+\epsilon_{k-1})\cdot {\log n}$. Then, 
	By Lemma \ref{lemma stretch}, we have

	\begin{equation*}\label{eq temp1}
	\begin{array}{lllclclclclc}
	d^{(h_{\ell})}_{G_k}(u,v) 
	&\leq& (1+\epsilon_{k-1} +5\cdot c(n)(\ell-1)\epsilon)\cdot d_G(u,v) +\ta{\ell-1} \\
	
	&=& (1+\epsilon_{k-1} +5\cdot c(n)(\ell-1)\epsilon)\cdot d_G(u,v) +10c(n)\cdot 2^{k}\epsilon \\
	
	&\leq& (1+\epsilon_{k-1} +5\cdot c(n)\epsilon (\ell+1))\cdot d_G(u,v) \\

	&=& (1+\epsilon_{k-1} +5\cdot 4(1+\epsilon_{k-1})\cdot {\log n}\cdot \epsilon (\ell+1))\cdot d_G(u,v) \\

	&=& (1+\epsilon_{k-1})(1+20 {\log n}\cdot \epsilon (\ell+1))\cdot d_G(u,v) \\

	&=& (1+\epsilon_{k-1})(1+\epsilon')\cdot d_G(u,v) .

	\end{array}
	\end{equation*}

\end{proof}
Recall that $\epsilon_{k-1}$ is the value such that $G_{k-1}$ provides stretch $(1+\epsilon_{k-1})$, and that $k_0 = {\lfloor {\log \beta}\rfloor}$ (see \cref{eq beta def}) and $\lambda = \lceil {\log \Lambda}\rceil$. 
For $k<k_0$, note that $\epsilon_k = 0$. We now bound $\epsilon_k$ for $k\geq k_0$. 
\begin{lemma}\label{lemma bound ek}
	For $k\in [k_0,\lambda]$ we have 
	$$ 1+\epsilon_{k}\leq (1+\epsilon')^{k}.$$
\end{lemma}

\begin{proof}
	The proof is by induction on the index of the scale $k$. For $k=k_0$, we have $1+ \epsilon_{k_0} = (1+\epsilon_{k_0-1})(1+\epsilon') =1+ \epsilon'$ since $\epsilon_{k_0-1} = 0$. Thus the claim holds. 
	
	Assume that the claim holds for some $k\in [k_0,\lambda-1]$ and prove it for $k+1$. By definition and the induction hypothesis, we have 
	$$1 + \epsilon_{k+1} = (1+\epsilon_{k})(1+\epsilon') \leq (1+\epsilon')^{k}(1+\epsilon') = (1+\epsilon')^{k+1}.
	$$
\end{proof}

Observe that Lemma \ref{lemma bound ek} implies that the overall stretch of our hopset is at most
$(1+\epsilon')^\lambda$. 

Define $\epsilon'' = 2\lambda\epsilon'$. The condition $\epsilon'<1$ is replaced by a stronger condition, $\epsilon''<1$. It follows that our hopset guarantees a stretch of 

$$1+\epsilon_\lambda = 
\left(1+\frac{\epsilon''}{2\lambda}\right)^\lambda \leq 1+\epsilon''. 
$$ 

The hopbound $h_\ell$ now translates to 

\begin{equation}
\label{eq betaell rescaled2}
h_\ell = O\left( \frac{\lambda\ell{\log n}}{\epsilon''} \right)^\ell = O\left( \frac{\lambda({\pramell}){\log n}}{\epsilon''} \right)^{\pramell} .
\end{equation}
We set $\beta = h_\ell$. 
 Write now $\epsilon =\epsilon''$. Observe that this setting of $\beta$ is consistent with \cref{eq beta def}. See Lemma \ref{lemma rt} for the time and work complexities, and \cref{eq hk size} for the size. We conclude:

\begin{theorem}\label{theorem final hopset}
	Given a weighted undirected graph $G=(V,E,\omega)$ on $n$ vertices with aspect ratio $\Lambda$, and parameters $0<\epsilon<1$, $\kappa =2,3,\dots$, and $0<\rho<1/2$, our algorithm deterministically computes a $(1+\epsilon,\beta)$-hopset $H$ of size at most $\lceil{\log \Lambda}\rceil\cdot \nfrac$ in $O({\log \Lambda}({\log \kappa\rho}+1/\rho)\beta{\log^2 n})$ time in the PRAM CREW model using $O((|E|+\nfrac)\cdot n^\rho)$ processors, where 
	$$\beta =O\left( \frac{{\log {\Lambda}}{\log n} ({\log \kappa\rho} + 1/\rho) }{\epsilon} \right)^{\pramell}.$$
\end{theorem}

Note also that for the result to be meaningful, we must have $\kappa,1/\rho = O({\log n})$, and thus ${\log (\kappa\rho)}+1/\rho = O({\log n})$. 
Hence, $\beta = \left(\frac{{\log\Lambda}{\log n}}{\epsilon}\right)^{O({\log \kappa \rho} +1/\rho)}$, and the running time is bounded by $\left(\frac{{\log\Lambda}{\log n}}{\epsilon}\right)^{O({\log \kappa \rho} +1/\rho)}$ as well.

Using the reduction provided in Appendix \ref{sec reduc}, one can eliminate the dependence of our result on the aspect ratio $\Lambda$ (see Theorem \ref{theo reduc}). Specifically, we show that one can obtain a $(1+\epsilon,\beta)$-hopset of size at most 
$O\left(\nfrac\cdot {\log n} \right)$, using $O\left(n^\rho{\log n}\left(|E| + \nfrac\cdot {\log n }\right)\right) $ processors and $
O(({\log \kappa\rho}+1/\rho)\beta{\log^3 n}) $ time, where 
 $$\beta =O\left( \frac{{\log^2 n} ({\log \kappa\rho} + 1/\rho) }{\epsilon} \right)^{\pramell}.$$

Given a $(1+\epsilon,\beta)$-hopset $H$ for a graph $G=(V,E)$, one can execute a Bellman-Ford exploration from a vertex $v\in V$ limited to $\beta$ hops. Such an exploration requires $O(\beta{\log n})$ time and uses $O(1)$ processors to simulate every vertex and every edge of $E\cup H$. For every vertex $u\in V$, this exploration provides the $\beta$-hop bounded distance between $u$ and the source $v$, which is at most $(1+\epsilon)d_G(u,v)$. Therefore, our hopset can be used to solve the $(1+\epsilon)$\textit{-approximate-single-source-shortest distance} (aSSSD) problem. Moreover, given a set of sources $S$, one can execute $|S|$ parallel Bellman-Ford explorations, each limited to $\beta$ hops, and solve the $(1+\epsilon)$\textit{-approximate-multiple-source-shortest distance} (aMSSD) problem. Executing $|S|$ Bellman-Ford explorations limited to $\beta$ hops can therefore be performed in $O(\beta{\log n})$ time, using $O(|S|)$ processors to simulate every vertex and every edge of $E\cup H$.

\begin{theorem}\label{theorem compute dist}
	Given a weighted undirected graph $G=(V,E,\omega)$ on $n$ vertices, parameters $0<\epsilon<1$, $\kappa =2,3,\dots$, and $0<\rho<1/2$, and a set of sources $S\subseteq V$,
	our deterministic algorithm computes $(1+\epsilon)$-approximate-distances for all pairs of vertices in $S\times V$ 
	in $O(({\log \kappa\rho}+1/\rho)\beta{\log^3 n}) $
	time in the PRAM CREW model using 	
	$O\left(n^\rho{\log n}\left(|E| + \nfrac\cdot {\log n}\right)\right) $ processors, where 
		$$\beta =O\left( \frac{\log n} {\epsilon} \right)^{O({\log \kappa\rho}+\frac{1}{\rho})}.$$
\end{theorem}

In particular, for the \textit{single}-source $(1+\epsilon)$-approximate shortest distance problem, one can set $\rho = 1/\kappa$, and obtain a deterministic polylogarithmic-time algorithm (the time is 
$\left(\frac{{\log n}}{\epsilon}\right)^{O(1/\rho)}$), 
with $O(|E|\cdot n^\rho)$ processors. To our knowladge, this is the first deterministic polylogarithmic time algorithm for this problem that employs less than $O(n^\omega)$ processors \cite{Zwick98}. (Here, $\omega$ is the matrix multiplication exponent.) 

\section{Path-Reporting Hopsets}\label{sec path-reporting}

In this section, we modify our algorithm and show that
the modified algorithm constructs \textit{path-reporting} hopsets. Roughly speaking, a hopset is called path-reporting if it can be used to retrieve actual paths, and not just approximate distances. Unlike previous path-reporting hopsets \cite{ElkinN17Hop,ElkinN19}, our construction enables us to retrieve a $(1+\epsilon)$-SPT (see below).

Given a graph $G= (V,E)$, a source $s\in V$ and a parameter $0< \epsilon <1$, 
we compute in polylogarithmic time a \textit{$(1+\epsilon)$-approximate-single-source-shortest-path} tree (henceforth, $(1+\epsilon)$-SPT) $T= (V,E_T)$ with $E_T\subseteq E$, rooted at $s$.
For every vertex $v\in V$, the distances in the tree $T$ will satisfy 
$$d_T(s,v)\leq (1+\epsilon)d_G(s,v).$$
Moreover, every vertex $v\in V$ will know its parent with respect to $T$ and also the distance $ d_T(s,v)$ between $v$ and the source $s$.

We compute a $(1+\epsilon)$-SPT rooted at $s$ in the following way. First, a $(1+\epsilon,\beta)$-hopset $H$ for $G$ is constructed. 
A Bellman-Ford exploration is executed from $s$ in the graph $\mathcal{G}= (V,E\cup H)$ to depth $\beta$. 
Let $\mathcal{T}$ be the shortest path tree in $\mathcal{G}$ obtained by the exploration. For every vertex $v\in V$, the fields $p(v),d(v)$ are initialized to be the parent of $v$ with respect to the tree $\mathcal{T}$, and the distance $d_{\mathcal{T}}(s,v)$ between $s$ and $v$ in $\mathcal{T}$, respectively. 
Note that some of the edges of $\mathcal{T}$ may be edges of the hopset $H= \bigcup_{k\in \krange} H_k$. 

The algorithm proceeds for $\lambda-k_0+1$ iterations. The input for every iteration $j\in [0,\lambda-k_0]$ is a tree $\mathcal{T}_{\lambda-j}$ rooted at the source vertex $s$. Moreover, we will show that the tree $\mathcal{T}_{\lambda-j}$ contains only edges of the graph  ${\mathcal{G}}_{\lambda - j} = (V,E\cup\bigcup_{k'\in [k_0,\lambda-j]}H_{k'})$.
 The input $ \mathcal{T}_\lambda$ for the first iteration, i.e., $j=0$, is the tree $\mathcal{T}$.
For every iteration $j\in [0,\lambda-k_0]$, denote $k = \lambda-j$. 
During every iteration $j\in [0,\lambda-k_0]$, edges of the hopset $H_{k}$ that belong to $\mathcal{T}_k$ are removed from $\mathcal{T}_k$. Each hopset edge $(v,v')$ in $\mathcal{T}_k$ is replaced by a path from $G_{k-1}$ between the vertices $v$ and $v'$, with weight no greater than $\omega_{H_k}(v,v')$. The details of this edge-replacing process are discussed later in the sequel. At the end of this process, each vertex $v\in V\setminus \{ s\}$ is left with a single 
parent. The graph that contains all the edges from vertices in $V\setminus \{ s\}$ to their respective parents is the input $\mathcal{T}_{k-1}$ for the next iteration. 
We will later show that for every $k\in [k_0,\lambda ]$, the graph $\mathcal{T}_{k}$ is indeed a tree rooted at $s$, and that all its edges belong to $ E\cup \bigcup_{j\in [k_0,k]}H_{j}$. 

Note that a vertex $v\in V$ may belong to a path from $G_{k-1}$ for more than one hopset edge in $H_k$. The vertex $v$ will update its parent and distance estimate according to the edge that provides it with the smallest distance estimate. In particular, during this iteration, vertices $v\in V$ that have $(p(v),v)\in E\cup \bigcup_{k'\in [k_0,k-1]}H_{k'}$, update their distance estimate and parent only when it improves upon their existing distance estimate. 

When this process terminates,
each vertex $v\in V$ knows a parent $p(v)$ such that $(p(v),v)\in E$. The output of the procedure is the tree $T = \mathcal{T}_{k_0-1}$. We will show that this tree contains only edges from the original graph $G$. We note that when the edge-replacing process terminates, some vertices $v\in V$ may have a distance estimate $d(v)$ that is slightly higher than the actual distance $d_T(s,v)$. Therefore, we use a \textit{pointer-jumping} algorithm, described in Section \ref{sec pointer}, to compute distances in $T$. 
This completes the high-level description of the procedure. The pseudocode of the algorithm is given in Algorithm \ref{alg path-report}.

\begin{algorithm}
	\caption{Path Reporting Overview}
	\label{alg path-report}
	\begin{algorithmic}[1]
		\Statex \textbf{Input:} a weighted, undirected graph $G=(V,E,\omega)$, a source $s\in V$
		\Statex \textbf{Output:} a $(1+\epsilon)$-approximate shortest path tree $T$ rooted at $s$, such that every vertex $v\in V$ knows its parent $p(v)$ w.r.t. $T$, and $d(v)= d_T(s,v)$
		\State compute $(1+\epsilon,\beta)$-hopset $H= \bigcup_{k \in \krange}H_k$ with weights $\omega_H$ for $G$
		\State $\mathcal{G}= (V,E\cup H,\omega_{\mathcal{G} })$ where $\omega_{\mathcal{G} }(u,v) = {\min \{ \omega(u,v), \omega_H(u,v) \} } $ for all $(u,v)\in E\cup H$
		\State execute a Bellman-Ford exploration from $s$, limited to $\beta$ hops in the graph $\mathcal{G}$, and let $\mathcal{T}= \mathcal{T}_\lambda$ be the resulting tree
		\For { $k = \lambda,\lambda-1,\dots, k_0$ }
			\State Let $\mathcal{T}_{k-1}$ be the tree obtained by replacing edges of $\mathcal{T}_k$ that belong to $H_k$ with edges of $E\cup H_{k-1}$
		\EndFor
		\State $T = \mathcal{T}_{k_0-1}$
		\State execute \textit{pointer-jump} algorithm to obtain $d(v) = d_T(s,v)$ for all $v\in V$
	\end{algorithmic}
\end{algorithm} 

\newcommand{\pedge}{(p(v),v)}
\newcommand{\uve}{(u,v)}

\subsection{Replacing Hopset Edges}\label{sec replace}
In this section, we describe the execution of every iteration $j\in [0,\lambda-k_0]$. Denote $k= \lambda-j$.

We begin by introducing a property of hopset edges. A hopset edge $(u,v)$ that belongs to some hopset $H_k$ is said to have the \textit{memory} property, if it is associated with an array $A(u,v)$ that contains a path $\pi_{G_{k-1}}(u,v)$ from the graph $G_{k-1}$ between the vertices $u$ and $v$, with weight at most $\omega_k\uve$. Recall that $\omega_k$ is the weight function of the graph $G_k$, i.e., for every edge $\uve\in E\cup H_k$ the weight $\omega_k(u,v)$ is the minimum between the weight of $\uve$ in the original graph $G$, and in the hopset $H_k$.
The path $P = \pi_{G_{k-1}}(u,v)$ is also required to contain at most $\sigma$ hops, for some parameter $\sigma$ that will be specified in the sequel. For every vertex $x$ along $P$, the array $A(u,v)$ also contains the distance of $x$ from the endpoints $u$ and $v$, along the path $P$. The path $P$ is also referred to as the \textit{memory-path} of the edge $(u,v)$. 

We assume that all hopset edges of $H$ possess the memory property. Moreover, for edges that belong to more than one hopset among $H_{k_0},H_{k_0+1},\dots,H_{\lambda}$, we assume that they satisfy the memory property w.r.t. all hopsets to which they belong.
In Section \ref{sec mem prop} we show 
how to modify the construction of hopsets in order to guarantee this property, and specify the value of $\sigma$.

For convenience, we imagine that the path $s-v$ between the source $s$ and some vertex $v$ in $\mathcal{T}_k$, goes from left to right, and so also go the paths $\pi(p(u),u)$, for every vertex $u$ and its parent $p(u)$.
Consider a vertex $v\in V$ such that $(p(v),v)\in H_k$, and let $A\pedge=\langle p(v)=x_0,x_1,\dots,x_t= v\rangle$ be the memory path of $\pedge$. 
The vertex $v$ is responsible for replacing the edge $(p(v),v)$ in $\mathcal{T}_k$ with the memory path of $\pedge$. 
Specifically, $v$ will set its parent to be $x_{t-1}$. (There are processors associated with the vertex $v$, and these processors will perform these operations.)
It will also inform all vertices $x_1,x_2,\dots,x_{t-1}$ that they belong to the memory path $A\pedge$.
For this end, we use a global array $M$ of size $\sigma\cdot n$, such that every vertex $v\in V$ has $\sigma$ cells in $M$ associated with it. We will soon provide more details about the way that the algorithm manipulates with the array $M$.

For every index $i\in [1,t-1]$ let $d_{\pi}(p(v),x_i)$ be the distance from $p(v)$ to $x_i$ on the path $\pi  = A(p(v),v)$. 
Let $d'_{v}(x_i)$ be the distance estimate that the vertex $v$ has for the vertex $x_i$, i.e., $d'_{v}(x_i)= d(p(v))+d_{\pi}(p(v),x_i)$. This estimate corresponds to a path obtained by concatenating the $s-p(v)$ path in $\mathcal{T}_k$ with the $p(v)-x_i$ subpath of the path $\pi$. Let $p'_v(x_i) = x_{i-1}$, i.e., $p'_v(x_i)$ is the neighbor of $x_i$ along the path $\pi$, that is closer to $p(v)$ than $x_i$.

For every index $i\in [1,t-1]$, the vertex $v$ writes the triplet $\langle x_i, d'_{v}(x_i), p'_{v}(x_i)=x_{i-1}\rangle$ to the array $M$. Then, the vertex $v$ updates its parent $p(v)$ to be its left neighbor on the path $\pi\pedge$, i.e., $x_{t-1}$. See Figure \ref{fig:openedge} for an illustration. 
Observe that at this time point, for every vertex $v\in V$, the edge $(p(v),v)$ is either an edge from the original graph $G$, or it belongs to some hopset $H_{k'}$ where $k'<k$. This is since the input tree $\mathcal{T}_k$ contains edges from the set $E\cup \bigcup_{k'\in [k_0,k]}H_{k'}$. During the current iteration, every edge of $H_{k}$ in $\mathcal{T}_k$ is deleted, and the only edges that are added to the tree $\mathcal{T}_k$ are edges from the graph $G_{k-1}$.

\begin{figure}
	\centering
	\includegraphics[scale=0.25]{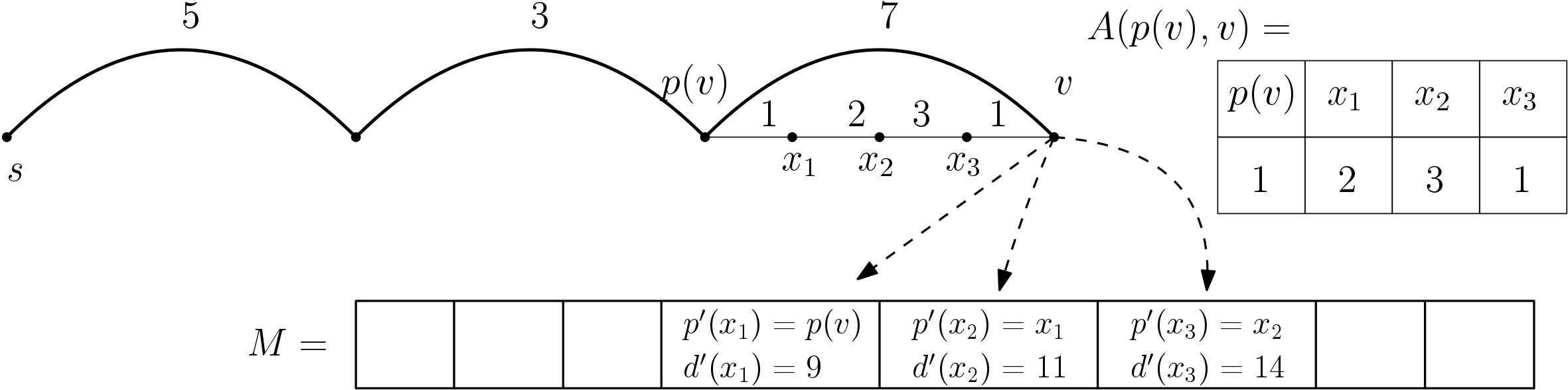}
	\caption{Removing hopset edges from the tree $\mathcal{T}_k$. In the figure, the thick curved lines depict edges of $\mathcal{T}_k$. The straight lines depict the edges of $G_{k-1}$ in the memory path of $\pedge$. The numbers above each edge represent its weight. The edge $(p(v),v)$ is an edge of the hopset $H_k$. The vertex $v$ writes to $M$ the three elements in $M$, which it computes using the array $A\pedge$, and the distance estimate $d(p(v))$. }
	\label{fig:openedge}
\end{figure}

Some vertices $x\in V$ may appear on more than one hopset edge. The array $M$ is sorted according to IDs. Ties are broken according to distance estimates. Then, every vertex $x\in V$ uses a binary search to find the first entry $M[ind]$ in $M$ that contains its ID. If the entry $M[ind]$ contains a distance estimate that is smaller than $d(x)$, then $x$ updates its distance estimate and its parent $p(x)$ accordingly. Otherwise, it ignores the new information. 

Observe that each edge added to the tree $\mathcal{T}_k$ by this procedure is an edge of the graph $G_{k-1}$. It follows that at this point, every edges of the tree $\mathcal{T}_k$ either belongs to $E$, or to some hopset $H_{k'}$ where $k'<k$.
Let $\mathcal{T}_{k-1}$ be the graph obtained from $\mathcal{T}_k$ when the $j$th iteration terminates. 

This completes the description of the $j$th iteration of the edge-replacing procedure. 	

Next, we show that after every iteration $j\in [0,\lambda-k_0]$ of the algorithm, the tree $\mathcal{T}_{\lambda-j-1}$ is a $(1+\epsilon)$-SPT rooted at $s$. Note also that all its edges belong to $E\cup \bigcup_{k' \in [k_0,\lambda-j-1]}H_{k'}$.
We begin by showing that at the end of every iteration $j$, $\mathcal{T}_{\lambda-j-1}$ is indeed a tree. Recall that every vertex (other than $s$) has only one parent in $\mathcal{T}_{\lambda-j-1}$. Therefore, the number of edges in $\mathcal{T}_{\lambda-j-1}$ is $n-1$. The following lemma shows that there are no cycles in $\mathcal{T}_{\lambda-j-1}$, and as a result, it follows that $\mathcal{T}_{\lambda-j-1}$ is a tree. For notational purposes, we introduce another iteration to the algorithm, 
$j = \lambda-k_0+1$. As $H_{k_0-1}=\emptyset$, this iteration does nothing.

\begin{lemma}\label{lemma no cycles}
	At the beginning of every iteration $j\in [0,\lambda-k_0+1]$, for every vertex $x\in V\setminus \{s\}$ we have $d(x)>d(p(x))$.
\end{lemma}
\begin{proof}
	Assume inductively that the claim holds at the beginning of iteration $j$, for some $j\in [0,\lambda-k_0]$, and prove that it holds also at the end of the $j$th iteration .
	Note that at the beginning of iteration $0$, the tree $\mathcal{T}_{\lambda}$ and the distance estimates are a result of a Bellman-Ford exploration in the graph $\mcg{}$, and so the claim holds. (Recall that the minimal edge weight is $1$, and therefore all edge weights are positive.)
	
	Consider some iteration $j\in [0,\lambda-k_0]$, and let $k = \lambda-j$.
	Consider a vertex $x\in V\setminus \{s\}$. The proof splits into two cases:

	\textbf{Case 1:} 
		$x$ did not change its distance estimate and its parent during the $j$th iteration. Then, by the induction hypothesis, at the beginning of the $j$th iteration, the distance estimate of $p(x)$ is smaller than the distance estimate of $x$. Recall that distance estimates never increase throughout the algorithm. It follows that at the end of the $j$th iteration, we also have $d(x)> d(p(x))$.
		
	\textbf{Case 2:}
		 $x$ has changed its distance estimate \textit{or} its parent during the $j$th iteration. Let $\pedge\in H_{\lambda-j}$ be the edge according to which $x$ changed its fields, and let $x'$ be the left neighbor of $x$ along the memory path of $\pedge$. Note that $x'$ becomes the new parent of $x$.
		 
		 If $x'=p(v)$ then $x$ sets its parent to be $p(v)=x'$, and its estimate to be $d(p(v))+\omega(p(v),x) = d(x') + \omega(p(v),x)>d(x')$.
		 
		 Since the distance estimate of $x'$ never increases, at the end of the iteration we have $d(x)>d(x')$. 
		 
		 Otherwise, $x'\neq p(v)$. Therefore, the vertex $v$ wrote to $M$ the distance estimate (and parent) that $x'$ gets according to the memory path $\pedge$. This distance estimate is smaller than the estimate given to $x$ by $v$. It follows that at the end of the iteration, we have 
		 $d(x)>d(x')$. 
\end{proof}

\begin{lemma}\label{lemma still spt}
	After $\lambda-k_0$ iterations of the algorithm, $T = \mathcal{T}_{k_0-1}$ is a $(1+\epsilon)$-SPT rooted at $s$. 
	Moreover, all edges of $\mathcal{T}_{k_0-1}$ belong to the original graph $G$. 
\end{lemma}

\begin{proof}
	For the first assertion of the lemma, recall that at the beginning of the first iteration, the tree $\mathcal{T}_{\lambda}$ is a $(1+\epsilon)$-SPT rooted at $s$. 
	By Lemma \ref{lemma no cycles} and since there are exactly $n-1$ edges in every subgraph $\mathcal{T}_\lambda,\mathcal{T}_{\lambda-1},\dots,\mathcal{T}_{k_0-1}$, we conclude that every subgraph $\mathcal{T}_\lambda,\mathcal{T}_{\lambda-1},\dots,\mathcal{T}_{k_0-1}$, is a spanning tree of $\mcg{}_{\lambda},\mcg{}_{\lambda-1},\dots,\mcg{}_{k_0-1}$, respectively. 
	Since distances can only decrease during the algorithm, we have that at the beginning of every iteration $j\in [0,\lambda-k_0+1]$, the tree $\mathcal{T}_{\lambda-j}$ is a $(1+\epsilon)$-SPT of $\mcg{}_{\lambda-j}$ rooted at $s$.

	For the second assertion of the lemma, observe that during every iteration 
	$j\in [0,\lambda-k_0]$, every edge of the hopset $H_{\lambda-j}$ that belonged to $\mathcal{T}_{\lambda-j}$ is eliminated. Every edge that is added to $\mathcal{T}_{\lambda-j-1}$ during the $j$th iteration belongs to the graph $G_{\lambda-j-1}$, i.e., it is either an edge of $E$ or an edge of $H_{\lambda-j-1}$. 
	It follows that after $\lambda-k_0$ iterations, the tree $\mathcal{T}_{k_0-1}$ does not contain any hopset edges, and so it contains only edges of the original graph $G$. 
\end{proof}

\subsection{Computing Exact Distance Estimates}
\label{sec pointer}
When the algorithm terminates, some vertices $v\in V$ may have an estimate $d(v)>d(p(v))+\omega(p(v),v)$. To ensure that for every vertex $v$ we have $d(v)=d(p(v))+\omega(p(v),v)$, we use a \textit{pointer-jumping} procedure (see, e.g., \cite{ShiloachV82}). As a result of this procedure, for every vertex $v\in V$, the estimate that it will have will be equal to its distance in $T= \mathcal{T}_{k_0-1}$ to the root $s$ of ${T}$.

 For the source vertex $s$, set $p(s)=s$ and $\omega(s,s)= 0$. For every vertex $v\in V$,
define $q(v) = p(v)$. In addition, every vertex maintains a field $d'(v)$ that contains the distance in ${T}$ from $v$ to $q(v)$. This field is initialized to contain $ d'(v)= \omega\pedge$. Then, for ${\log n}$ iterations, each vertex updates:
\begin{equation*}
\begin{array}{clclcl}
	d'(v) = d'(v)+d'(q(v)),\\
	q(v) = q(q(v)).
\end{array}
\end{equation*}

Next, we show that after ${\log n}$ iterations, we have $d'(v) = d_{{T}}(s,v)$. We will consider the initialization step as iteration $0$ of the procedure. For every iteration $j\in [0,{\log n}]$, denote $q_j(v)$ and $d'_j(v)$ the values of the fields $q(v),d'(v)$ at the end of the $j$th iteration, respectively. 

\begin{lemma}\label{lemma triv}
	At the end of every iteration $j\in [0,{\log n}]$, we have 
	$ d'_j(v) = d_{T}(q_j(v),v)$. 
\end{lemma}
\begin{proof}
	The proof is by induction on the index of the iteration $j$. For $j=0$, recall that $q(v)= p(v)$ and $d'_0(v) = d'(v) = \omega\pedge = d_{T}\pedge$, and so the claim holds.

	Assume that the claim holds for some $j\in [0,{\log n}-1]$ and prove it for $j+1$. Let $v\in V$. 
	At the end of iteration $j+1$, we have $d'_{j+1}(v) = d'_{j}(v)+d'_j(q(v))$ and also $q_{j+1}(v) = q_j(q_j(v))$. By the induction hypothesis, we have 
	$$ d'_{j+1}(v) = d'_{j}(v)+d'_j(q(v)) =
	d_{T}(q_j(v),v) + d_{T}(q_j(q_j(v)),q_j(v)) = d_{T}(q_j(q_j(v)),v) = d_{T}(q_{j+1}(v),v).$$
	
\end{proof}
After ${\log n}$ iterations, each vertex $v\in V$ updates $d(v) = d'_{{\log n}}(v)$. 
Observe that, for every vertex $v\in V$ we have $q_{{\log n}}(v) = s$. Therefore, Lemma \ref{lemma triv} implies that when the procedure terminates, we have $d(v) = d_{{T}}(s,v)$.

\subsection{Constructing the Path-Reporting Hopset }
\label{sec mem prop}
In this section we modify our algorithm for building hopsets, so that every hopset edge will satisfy the \textit{memory} property (see Section \ref{sec replace}).

Let $k\in \krange$, and let $i\in [0,\ell]$. 
We say that a vertex $v$ that belongs to a cluster $C\in P_i$ has a \textit{cluster-memory}, if it is associated with arrays $CP(v)$ and $CD(v)$ that contain a path $P$ from $v$ to the center $r_C$ of the cluster $C$, and the distances of each vertex in $CP(v)$ from the center of the cluster along the path $CP(v)$, respectively.
Moreover, the path $P$ is required to be contained in $ E\cup H_{k-1}$. 
See Figure \ref{fig:message} for an illustration.
 If $v$ does not belong to a cluster in phase $i$, we say that it has a cluster-memory \emph{vacuously}.

Recall that a hopset edge $(u,v)$ is added to $H_k$ in phase $i$, because some exploration that has originated, w.l.o.g., in the cluster of $u$ has reached the cluster of $v$. Specifically, the vertex $u$ wrote some information regarding its own exploration to its memory cells. This information was then read by some other vertices, which in turn wrote some information regarding $u$'s exploration to their memory cells, etc. For convenience, we will say that when a vertex $u\in V$ initiates an exploration, it sends a message $m$ along edges of the graph. If a vertex $v$ was discovered by this exploration, and added an edge to $u$, we say that the message $m$ has reached $v$. Let $d$ be the weight of the hopset edge $(u,v)$.
We require the message $m$ to contain a path from $u$ to $v$ with weight at most $d$.

When the vertex $u$ first writes the message $m$ to its memory, it writes the triplet
$ \langle u, \langle u \rangle, \langle 0 \rangle\rangle$.
The first field in $m$ is the ID of the vertex in which the message originated. The second field in $m$, is a list $\mathcal{L}_{P}$ that contains the path of $G_{k-1}$ that this message has traversed. The third field in $m$, is a list $\mathcal{L}_{dist}$ of the distances the message has traversed before arriving at each vertex of the list, such that the $j$th distance in $\mathcal{L}_{dist}$ corresponds to the $j$th vertex on $\mathcal{L}_{P}$. See Figure \ref{fig:message} for an illustration. 

\begin{figure}
	\centering
	\includegraphics[scale=0.22]{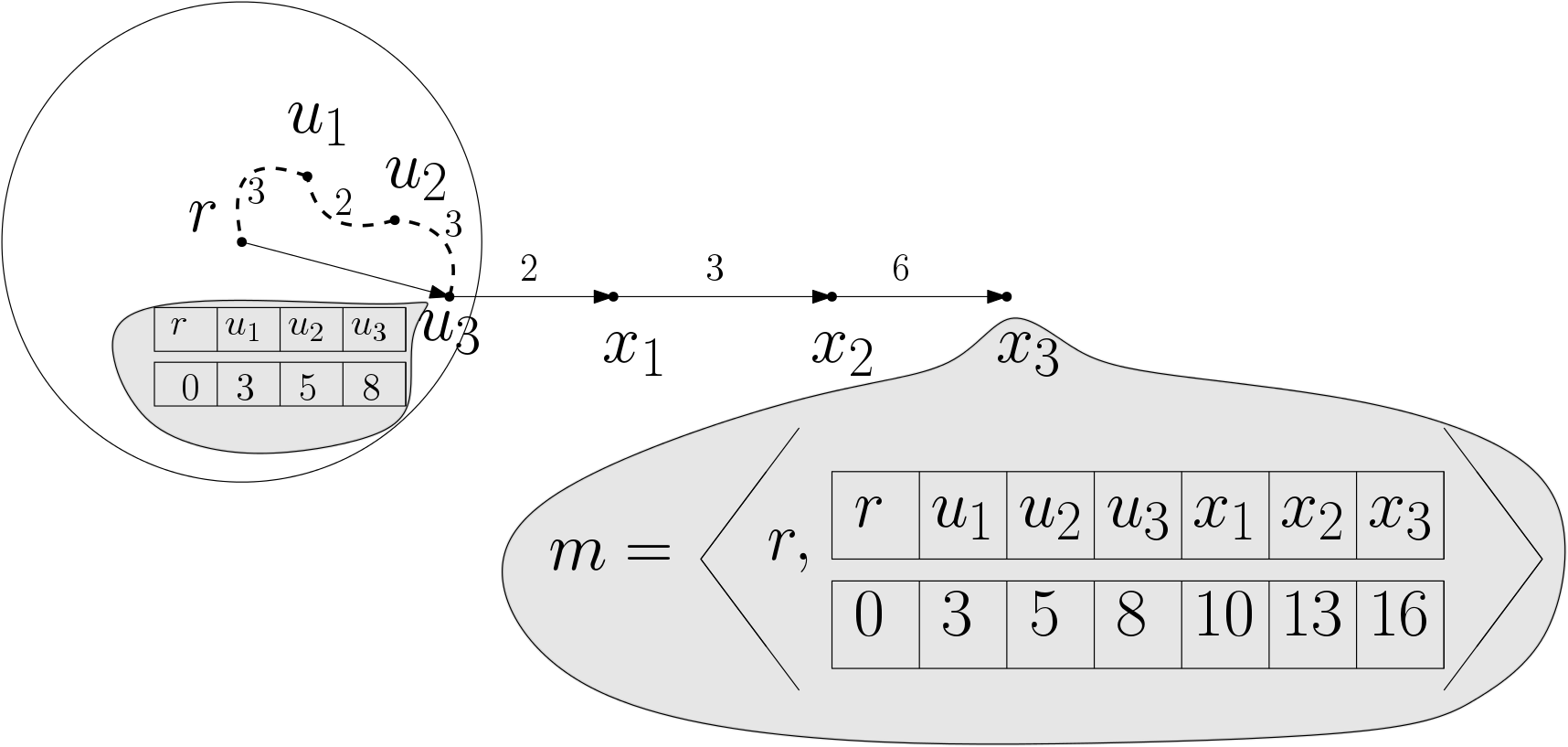}
	\caption{The path of a message from a cluster center to a vertex $x_3$. The solid arrows represent the path that the message traversed. The dashed lines represent the path in $G_{k-1}$ between $u_3$ and its cluster center $r$, that is stored in the array $CP(u_3)$. The small arrays depict the arrays $CP(u_3),CD(u_3)$. From $r$, the message reaches the vertex $u_3$, which appends $CP(u_3)$ and $CD(u_3)$ to the message. Then, the message is delivered to $x_1,x_2,x_3$, which in turn append their IDs and the total weighted length that the message has traveled to get to them. }
	\label{fig:message}
\end{figure}

Every vertex $x$ that receives a message $m$ adds the required information to the path and distance lists in $m$. Every vertex $x$ that receives this message from its neighbor $y$ in $G_{k-1}$ adds to $m$ the weight $\omega_{k-1}(y,x)$ to the weight field of $m$, and its ID to the path field, before delivering it further.
When a message $m = \langle x, \mathcal{L}_{P}, \mathcal{L}_{dist} \rangle $ originated at a vertex $x$ is delivered from a vertex $z\in C$ to the center $r$ of $C$ along a hopset edge $(z,r)$, it is the responsibility of the vertex $z$ to update the information in $ \mathcal{L}_{P}, \mathcal{L}_{dist} $ according to $CP(z),CD(z)$. 
Therefore, $z$ concatenates the path between $z$ and $r$ (i.e., $CP(z)$) to the field $\mathcal{L}_{P}$ in the message $m$. Additionally, it also computes for every vertex $w$ along $CP(z)$ the distance from the source $x$, according to $CD(z)$ and the message $m$. The list of distances is concatenated to the list $\mathcal{L}_{dist}$. 
Similarly, when a message $m = \langle x, \mathcal{L}_{P}, \mathcal{L}_{dist} \rangle $ is delivered to a vertex $z\in C$ from the center $r$ of $C$ along a hopset edge $(z,r)$, it is the responsibility of the vertex $z$ to update the information in $ \mathcal{L}_{P}, \mathcal{L}_{dist} $ according to $CP(z),CD(z)$.

Next, we explain how every vertex $v$ that belongs to a cluster $C\in P_i$ maintains the arrays $CP(v)$ and $CD(v)$. 
Consider a vertex $v\in V$. In phase $0$, the cluster of $v$ is a singleton, and so it writes $CP(v) = \langle v \rangle$ and $CD(v) = \langle 0\rangle$, and therefore it has a cluster-memory in phase $0$.

Assume inductively that every vertex $v\in V$ has a cluster-memory in phase $i$, for some $i\in [0,\ell-1]$. We will show that $v$ also has a cluster-memory in phase $i+1$. If $v$ does not belong to a cluster in $P_{i+1}$, then the claim holds vacuously. Consider the case where $v$ belongs to a cluster $\widehat{C}\in P_{i+1}$, formed in phase $i$ around a cluster $C\in P_i$. Let $r_C$ be the center of $C$, and let $C'\in P_i$ be the cluster such that $v\in C'$. If $C= {C'}$, than by the induction hypothesis $v$ has a cluster-memory also in phase $i+1$. Otherwise, the center $r_{C'}$ of the cluster $C'$ has added the superclustering edge $\cedge$ to the hopset. Note that the memory property of the edge $\cedge$ relies only on edges of $G_{k-1}$, and on the fact that all vertices have cluster-memory in phase $i$. Therefore, the edge $\cedge$ satisfies the memory property, i.e., it is associated with an array $A\cedge$ that contains the vertices of a path between $r_C$ and $r_{C'}$, and the distance of all vertices along this path from its endpoints. Since the vertex $v$ has a cluster-memory in phase $i$, the arrays $CP(v)$ and $CD(v)$ contain the details of a path to $r_{C'}$. The vertex $v$ adds to $CP(v)$ and $CD(v)$ the information from the array $ A\cedge$, and updates the distances accordingly. Therefore, the vertex $v$ has cluster-memory in phase $i+1$.

 Note that the length of the vertex arrays ($CP(\cdot)$ and $CD(\cdot)$) are dominated by the length of the edges arrays ($A(\cdot,\cdot)$). 
 It is left to bound the length of the array $A\cedge$ for every hopset edge.

 Recall that $\sigma_0 = 0$ and
 $\sigma_{i+1} = (4{\log n}+1)\sigma_i + 2 (2\beta+1){\log n}$, for all $i=1,2,\dots,\ell$.

 Lemma \ref{lemma intercon not short} implies that the number of hops along the path that implements an interconnection edge added to the hopset in phase $i$ of the algorithm, for $i\in [0,\ell]$, is at most $2\sigma_i+\tb$, and that the weight of this path is at most the weight of the edge in the hopset.
 
 For superclustering edges, Lemma \ref{lemma sup not shorter} implies that the number of hops along the path that implements a superclustering edge added to the hopset in phase $i$ of the algorithm, for $i\in [0,\ell-1]$, is at most $\sigma_{i+1}-\sigma_i$, and that the weight of this path is at most the weight of the edge in the hopset.

Note that $2\sigma_\ell+\tb> \sigma_{\ell}-\sigma_{\ell-1}$, and therefore, the maximal number $\sigma$ of edges in an array of a hopset edge is set to be $\sigma = 2\sigma_\ell+\tb$. We now
 provide an explicit bound on $\sigma_i$, for all $i\in [0,\ell]$. 

 \begin{lemma}\label{lemma exp bound sigma}
 	For every $i\in [0,\ell]$, we have $$\sigma_i = 2\tb{\log n}\cdot \sum_{j=0}^{i-1}(4{\log n}+1)^j.$$
 \end{lemma}
 \begin{proof}
 	The proof is by induction on the index of the phase $i$. For $i=0$, both sides of the equation are equal to $0$. 
 	
 	\induchyp\ By definition and the induction hypothesis, 
 	we have 
 	
 	\begin{equation*}
 	\begin{array}{lclcllcclclclclclc}
 	\sigma_{i+1} &=& 2 (2\beta+1){\log n} +(4{\log n}+1)\sigma_i 
 	\\&=& 2 (2\beta+1){\log n} +(4{\log n}+1)\left[ 2\tb{\log n}\cdot \sum_{j=0}^{i-1}(4{\log n}+1)^j\right]
 	\\&=& 2 (2\beta+1){\log n} +\left[ 2\tb{\log n}\cdot \sum_{j=0}^{i-1}(4{\log n}+1)^{j+1}\right]
 	\\&=& 2\tb{\log n}\cdot \sum_{j=0}^{i}(4{\log n}+1)^{j}
 	.
 	
 	\end{array}
 	\end{equation*}
 	
 \end{proof}
 Observe that Lemma \ref{lemma exp bound sigma} implies that 
 \begin{equation*}
 \begin{array}{lclcllcclclclclclc}
 \sigma_i = 2\tb{\log n}\cdot 
 \left[ \frac{(4{\log n}+1)^i-1}{4{\log n}} \right]\leq 
 \frac{1}{2}\tb\cdot 
 {(4{\log n}+1)^i}.
 \end{array}
 \end{equation*}
 
 Recall that by \cref{eq betaell rescaled2} we have $\beta =O\left( \frac{\lambda\ell{\log n}}{\epsilon} \right)^\ell$.
 It follows that the parameter $\sigma$ satisfies 
 \begin{equation}
 \label{eq sig bound}
 \sigma = 2\sigma_\ell+2\beta +1 
 = O(\sigma_\ell) = O( \beta\cdot 
 {(4{\log n}+1)^\ell}).
 \end{equation}

Observe that the current variant of the algorithm differs from the algorithm described in Section \ref{sec hopset const} only by the number of processors it uses. By using $O(\sigma n^\rho)$ processors to simulate every edge and every vertex, one can satisfy the memory property and maintain the same running time as in the variant from Section \ref{sec hopset const}. 
This is summarized in the following theorem.

\begin{theorem}\label{theorem final pr-hopset}
	Given a weighted undirected graph $G=(V,E,\omega)$ on $n$ vertices with aspect ratio $\Lambda$, and parameters $0<\epsilon<1$, $\kappa =2,3,\dots$, and $0<\rho<1/2$, our algorithm deterministically computes a path-reporting $(1+\epsilon,\beta)$-hopset $H$ of size at most $\lceil{\log \Lambda}\rceil\cdot \nfrac$ in $O({\log \Lambda}({\log \kappa\rho}+1/\rho)\beta{\log^2 n})$ time in the PRAM CREW model using $(|E|+\nfrac)\cdot \beta\cdot n^\rho \cdot 
	O({\log n})^{\pramell}$ processors, where 
$$\beta =O\left( \frac{{\log {\Lambda}}{\log n} ({\log \kappa\rho} + 1/\rho) }{\epsilon} \right)^{\pramell}.$$
\end{theorem}

Note that $\beta\cdot O({\log n})^{\pramell}= \left(\frac{{\log \Lambda}{\log n} }{\epsilon}\right)^{O({\log \kappa \rho}+1/\rho)}$. Thus, the number of processors is

$$\left(|E|+\nfrac\right)\cdot n^\rho\cdot\left(\frac{{\log \Lambda} {\log n}}{\epsilon}\right)^{O({\log \kappa\rho}+1/\rho)}.$$

See also the discussion that follows Theorem \ref{theorem final hopset}. It is applicable to Theorem \ref{theorem final pr-hopset} as well.

 \subsection{Complexity Analysis}
 In this section, we provide the analysis of the work and time required to compute approximate shortest paths from a single source in the graph $G=(V,E)$.

 By Theorem \ref{theorem final pr-hopset}, the hopset $H$ can be constructed in $O({\log \Lambda}({\log \kappa\rho}+1/\rho)\beta{\log^2 n})$ time using $O((|E|+\nfrac)\cdot \sigma n^\rho)$ processors.

 The Bellman-Ford exploration that computes $\mathcal{T} $ requires $O(\beta{\log n})$ time using $O(|E|+|H|)$ processors.
 
 The edge replacing procedure executes $\lambda-k_0+1 = O(\lambda)$ iterations. In each iteration $j\in [0,\lambda-k_0]$, each vertex $v$ that has $(p(v),v)\in H_{\lambda-j}$ updates its parent. It also uses the array $A\pedge$ to compute distance and parent estimates to all vertices along the memory path of $\pedge$. These estimates are written to the array $M$. Recall that the length of $A\pedge$ is $O(\sigma)$. This can be executed in $O(1)$ time, using $O(\sigma)$ processors to simulate every vertex $v\in V$. 
 
 Recall that the length of the array $M$ is $\sigma n$. 
 Sorting $M$ using $O(\sigma)$ processors for every vertex $v\in V$ can be performed in $O({\log |M|}) = O({\log (\sigma n)})$ time (see, e.g., \cite{AjtaiKS83}). 
 
 Each vertex $v\in V$ uses binary search to find the smallest element in $M$ that concerns $v$ in $O({\log (\sigma n)})$ time, using a single processor. 
 It then updates its parent and distance estimate in constant time. 
 
 It follows that every iteration of the edge replacing procedure can be performed using $O(\sigma)$ processors to simulate every edge of $E\cup H$ and every vertex in $V$. 
 
 Finally, the exact distances in the tree $T$ are computed using the \textit{pointer-jumping} algorithm. This requires $O({\log n})$ time and $O(1)$ processors to simulate every vertex.
 
 Recall that by \cref{eq sig bound}, we have $\sigma = O( \beta\cdot 
 {(4{\log n}+1)^\ell})= \left(\frac{{\log \Lambda}{\log n }}{\epsilon}\right)^{O(\ell)}$. 
 It follows that the running time of the algorithm is dominated by $O({\log \Lambda}({\log \kappa\rho}+1/\rho)\beta{\log^2 n})$, which is the time required to compute the hopset $H$, and the number of processors required for the simulation of every vertex of $V$ and edge in $E\cup H$ is $O(n^\rho\cdot \sigma) = O( \beta\cdot 
 {(4{\log n}+1)^\ell}\cdot n^\rho)$.

 The following theorem summarizes the properties of the path-reporting algorithm.

 \begin{theorem}
	\label{theo peth reporting}
	Given a weighted undirected graph $G=(V,E,\omega)$ on $n$ vertices with aspect ratio $\Lambda$, a source vertex $s\in V$ and parameters $0<\epsilon<1$, $\kappa =2,3,\dots$, and $0<\rho<1/2$, our algorithm deterministically computes a $(1+\epsilon)$-SPT for $G$ rooted at $s$ in $O({\log \Lambda}({\log \kappa\rho}+1/\rho)\beta{\log^2 n})$ time in the PRAM CREW model using $(|E|+\nfrac)\cdot  n^\rho \cdot \left(\frac{{\log \Lambda} {\log n}}{\epsilon}\right)^{O({\log \kappa \rho}+1/\rho)}$
	 processors, where 
	$$\beta =O\left( \frac{{\log {\Lambda}}{\log n} ({\log \kappa\rho} + 1/\rho) }{\epsilon} \right)^{\pramell}.$$
 \end{theorem}
 
 In Appendices \ref{sec reduc} and \ref{append red path}, based on \cite{KleinS97}, we argue that the dependence on $\Lambda$ in this result can be eliminated,
while keeping the running time and work complexity of our algorithm essentially intact.




\newpage
	
\begin{appendices}
\begin{center}
{\huge{\bf Appendix}}
\end{center}
%
%
%
%
%

	\section{Simulating Parallel BFS Explorations in the Virtual Graph $\tilde{G}_i$}
		\label{append explorations}
	This section contains the details for the explorations in the virtual graph $\tilde{G}_i$ (see Section \ref{sec superclustering} for its definition).
	Given a set of source clusters $S\subseteq P_i$, an upper bound 
	 $x\leq n^\rho+1$
	on the number of parallel explorations that traverse any vertex and a distance threshold parameter $d$, each cluster $C\in P_i$ will learn the IDs and distances to (up to) $x$ sources $C'\in S$ (including $C$ itself, if $C\in S$), that are within (unweighted) distance at most $d$ from $C$ in the virtual graph $\tilde{G}_i$. 
	
	We note that our algorithm for constructing hopsets uses the exploration algorithm only in the case where $1\leq x \leq n^\rho+1$ and $d=1$ (for the popular clusters detection), and in the case where $x=1$ and $d\geq 1$ (for simulating a single, multiple source BFS exploration in $\tilde{G}_i$). We provide here a general algorithm, but prove the correctness only for the two scenarios that our construction uses (i.e., the scenario where $x=1$ and the scenario where $d=1$). 

	\subsection{Overview}
	We begin with an intuitive overview of the algorithm. At the beginning of the algorithm, each cluster writes to its memory whether or not it belongs to the set of sources $S$. The algorithm proceeds in pulses. In every pulse $p\in [1,d]$, each cluster $C\in P_i$ aggregates the knowledge that its neighbors in the virtual graph $\tilde{G}_i$ obtained so far, regarding sources in $S$. The cluster $C$ then writes to its memory the IDs and distances to the closest $x$ sources it has learned about so far. (Recall that the cluster $C$ is simulated by the processors that simulate its center $r_C$, and that the memory used for the simulation of $r_C$ is also used for the simulation of $C$.)

	Each pulse is divided into three parts. 
	 In the \textit{distributing} part, each vertex $v$ that belongs to a cluster $C\in P_i$ copies the information that its cluster $C$ possesses. Note that in the first pulse, this information is whether $C$ is in $S$ or not. 
	Then, in the \textit{propagation} part, the information obtained by each vertex propagates to its neighbors in $G_{k-1}$. Specifically, for $2\beta+1$ steps, each vertex $u\in V$ reads the information obtained by its neighboring vertices in $G_{k-1}$, and updates its own information accordingly.
	Finally, in the \textit{aggregation} part, each cluster $C\in P_i$ accumulates the information obtained by vertices $v\in C$, and updates its information accordingly. 
	This completes the overview of the algorithm.

	\subsection{Technical Details and Complexity Analysis}\label{append technical}
	This section provides the technical details and analysis of running time and work of the algorithm. The pseudocode of the algorithm is given in Algorithm \ref{alg parallel limited BF}.

	 \begin{algorithm}
		\caption{Parallel Limited BFS Exploration}
		\label{alg parallel limited BF}
		\begin{algorithmic}[1]
			\Statex \textbf{Input:} a weighted, undirected graph $G_{k-1}=(V,E\cup H_{k-1},\omega_{k-1})$, sets of clusters $P_i$, $S\subseteq P_i$, distance and hop threshold parameters $ \apdi,2\beta+1$, number of explorations parameter $x$, depth parameter $d$
			
			\For {every vertex $v\in V$\label{step in1}} \textbf{ in parallel }
			\State \label{step in2} $m(v)= \ $\textit{an array of length $(deg(v)+1)\cdot x $}
			\State \label{step in3} $\mathcal{L}(v)= \ $\textit{a list of length $ x $}
			\EndFor
			
			\For {\label{step in4}every cluster $C\in P_i$}\textbf{ in parallel }
			\State \label{step in5}$m(C)= \ $\textit{an array of length $|C|\cdot x$}
			\If {\label{step in6}$C\in S$} add the tuple $\langle C, 0 \rangle$ as the first record of $m(C)$
			\EndIf
			\EndFor
			\For {$d$ iterations}
			\Statex $\qquad$\textbf{\underline{distribution-part:}}
			
			\For {each vertex $v$ that belongs to a cluster $C\in P_i$ \textbf{in parallel}}
			\State	copy the first $x$ records in $m(C)$ to $m(v)$ 
			\EndFor
			
			%
			\Statex $\qquad$\textbf{\underline{propagation-part:}}
			\For {$2\beta+1$ steps }
			\For {every vertex $u\in V$ \textbf{in parallel}}
			\State copy $\mathcal{L}(u)$ to $m(u)$
			\State let $x_1,x_2,\dots,x_{deg(u)}$ be the neighbors of $u$ in $G_{k-1}$
			\For {$t\in [1,deg(u)]$ \textbf{in parallel}}
			\State copy $\mathcal{L}(x_t)$ to $m(u)$ and add $\omega_{k-1}(u,x_t)$ to the distance value of each record
			\State remove all records with distance value grater than $\apdi$ from $m(u)$ 
			\EndFor

			\State apply Algorithm \ref{alg find best} on $m(u)$
			\State copy the first $x$ elements of $m(u)$ to $\mathcal{L}(u)$
			
			\EndFor
			\EndFor

			\Statex $\qquad$ \textbf{\underline{aggregation-part:}}
			\For {each vertex $v$ that belongs to a cluster $C\in P_i$ \textbf{in parallel}}
			\State	copy the list $\mathcal{L}(v)$ to $m(C)$ 
			\State \label{step sort}{apply Algorithm \ref{alg find best} on $m(C)$}
			
			\EndFor
			\EndFor
			
		\end{algorithmic}
	\end{algorithm} 
	
	\begin{algorithm}
		\caption{Sort Array}
		\label{alg find best}
		\begin{algorithmic}[1]
			\Statex \textbf{Input:} an array $m[]$ that contains entries $\langle C, d\rangle$. 
			\State sort $m[]$ according to IDs. Break ties according to distances.
			\For {every entry $ind$ in $m[]$ \textbf{in parallel }}
			\If {$m[ind-1]$ (if exists) has the same source as $m[ind]$}
			\State set $m[ind] = \langle null, \infty\rangle$
			\EndIf
			\EndFor 
			\State sort $m[]$ according to distances. Break ties by IDs.
		\end{algorithmic}
	\end{algorithm}

	At the beginning of the algorithm, each vertex $v\in V$ is allocated an empty list $\mathcal{L}(v)$ that can contain up to $x$ elements and an array $m(v)$ of length $(deg(v)+1)\cdot x$, where $deg(v)$ is the degree of the vertex $v$ in the graph $G_{k-1}$. Each cluster $C\in P_i$ is allocated an array $m(C)$ of length $|C|\cdot x$. 
	Each center $r_C$ of a cluster $C\in S$ writes $\langle C,0\rangle$ to $m(C)$. The second element of the pair $\langle C,dist\rangle$ is referred to as the \textit{distance value} of the pair.
	
	Recall that we use at least $x$ processors to simulate each vertex and each edge of $G_{k-1}$, and that the number of edges in $G_{k-1}$ is $|E|+|H_{k-1}|$. 
	We now describe the distribution, propagation and aggregation parts for each pulse $p\in [1,d]$.
	
	\textbf{Distribution part.}
	During the distribution part of pulse $p$, each vertex $u$ that belongs to a cluster $C\in P_i$ copies the first $x$ records in the array $m(C)$ to its list $\mathcal{L}(u)$. 
	This can be executed in $O(1)$ time using $x$ processors for each vertex. 
	
	\textbf{Propagation part.} 
	The propagation part is composed of $2\beta+1$ steps. In each step $j\in [1,2\beta+1]$, the processors that simulate each vertex $v\in V$ copy the contents of the list $\mathcal{L}(v)$ to the array $m(v)$. In addition, the processors that simulate each edge $(v,u)$ incident to $v$ in the graph $G_{k-1}$ copy the contents of the list $\mathcal{L}(u)$ to $m(v)$, and add the weight $\omega_k(v,u)$ to the distance value of each copied pair. If a record in $m(v)$ now has a distance value greater than $\apdi$, the record is deleted. Since each vertex and each edge are simulated by at least $x$ processors, this requires $O(1)$ time. 
	
	Observe that the array $m(v)$ might contain multiple entries regarding the same source. Therefore, the array $m(v)$ is now sorted according to the source ID. Ties are broken by distance. For each index $ind\in [1,|m(v)|]$, we check if the entry $m(v)[ind-1]$ (if exists) has the same source as the entry $m(v)[ind]$. If the answer is positive, then the entry $m(v)[ind]$ is deleted, i.e., replaced with $\langle null, \infty\rangle$. Then, the array $m(v)$ is sorted again, now according to distances. See Algorithm \ref{alg find best} for the pseudocode of this sorting procedure. The smallest $x$ elements in $m(v)$ are copied to the list $\mathcal{L}(v)$. This completes the description of the propagation part.

	Observe that for every vertex $v\in V$, we have that the length of the array $m(v)$ is $(deg(v)+1)x$, i.e., polynomial in $n$. Therefore, sorting it using $(deg(v)+1)x$  processors that simulate $v$ and its edges requires $O({\log n})$ time \cite{AjtaiKS83}. Copying the smallest $x$ elements in the sorted array $m(v)$ requires $O(1)$ time, using the processors that simulate the vertex $v$. It follows that the propagation part can be executed in $O(\beta{\log n})$ time, using $x$ processors to simulate each edge in $G_{k-1}$, i.e., using $O((|E|+|H_{k-1}|)\cdot x)$ processors. 
	
	\textbf{Aggregation part.} For each cluster $C\in P_i$, each vertex $u\in C$ copies the contents of the list $\mathcal{L}(u)$ into the array $m(C)$. To remove duplicates and find the smallest elements in $m(C)$, the array is sorted as described in the propagation part above (see Algorithm \ref{alg find best}). 
	Observe that the size of the array $m(C)$ is $|C|\cdot x$, i.e., polynomial in $n$. Therefore, it can be sorted in $O({\log n})$ time using $|C|\cdot x$ processors (see \cite{AjtaiKS83}). It follows that the aggregation part can be executed in $O({\log n})$ time, using $O((|E|+|H_{k-1}|)\cdot x)$ processors. 
	
	To summarize, each pulse $p\in [1,d]$ can be executed in $O(\beta{\log n})$ time using $O((|E|+|H_{k-1}|)\cdot x)$ processors. 
	
	\begin{corollary}
		\label{coro exporation rt}
		Given a weighted undirected graph $G_{k-1} = (V,E\cup H_{k-1},\omega_{k-1})$ on $n$ vertices, 
		sets of clusters $P_i$, $S\subseteq P_i$, distance and hopbound parameters $\apdi,2\beta+1$, 
		number of explorations parameter $x$ and a depth parameter $d$,
		 Algorithm \ref{alg parallel limited BF} can be executed in $O(d\beta{\log n})$ time using $O((|E|+|H_{k-1}|)\cdot x)$ processors.
	\end{corollary}

	\subsection{Correctness}
	
	In this section, we prove the correctness of the two variants of the algorithm we use. 
	The first variant is used to detect popular clusters from $P_i$. 
	The second variant is used for BFS explorations to depth $d\geq 1$ in $\tilde{G}_i$. 

	\subsubsection{Variant 1: Detecting Popular Clusters}
	\label{append var1}
	In this section, we show that by setting $d=1$, $x=deg_i+1$ and $S=P_i$, Algorithm \ref{alg parallel limited BF} can be used for the detection of popular clusters.

	For a vertex $u\in V$ and an index $j\in [0,2\beta+1]$, let $\mathcal{N}^j(u)$ be the set of clusters $C\in P_i$ such that $d_{G_{k-1}}^{(j)}(u,C)\leq \apdi$, where $ d_{G_{k-1}}^{(j)}(u,C) = {\min \{ d_{G_{k-1}}^{(j)}(u,v) | v\in C \} }$. Let $\mathcal{N}^j[x](u)$ be the subset of the $x$ closest reachable within $j$ hops clusters to $u$, from the set $\mathcal{N}^j(u)$, i.e., the set of (up to) $x$ clusters $C$ from $\mathcal{N}^j(u)$ with minimal $d^{(j)}_{G_{k-1}}(u,C)$. Ties are broken according to the cluster ID. For convenience, we refer to the initialization, i.e., \cref{step in1,step in2,step in3,step in4,step in5,step in6} of Algorithm \ref{alg parallel limited BF}, as step $0$ of the propagation part. 
	\begin{lemma}\label{lemma induc exploration}
		For every index $j\in [0,2\beta+1]$, when step $j$ of the propagation part terminates, for every vertex $u\in V$ the list $\mathcal{L}(u)$ contains the IDs and the $j$-hop distance from all clusters in $\mathcal{N}^j[x](u)$.
	\end{lemma}

	\begin{proof}		
		The proof is by induction on the index of the step $j$. Consider a vertex $u\in V$. For $j=0$,
		if the vertex $u$ belongs to a cluster $C\in P_i$, then the 
		set $\mathcal{N}^0[x](u)$ contains the cluster $C$. Otherwise, by definition, the set $\mathcal{N}^0[x](u)$ is empty. Indeed, if $u$ belongs to a cluster $C$, then the list $\mathcal{L}(u)$ contains the entry $\langle C,0\rangle$. Otherwise, it is empty, and so the claim holds for the base case. 
		
		We assume that the claim holds for some $j\in [0,2\beta]$ and prove that it holds also for $j+1$. 
		Let $u$ be a vertex in $V$, and let $C\in \mathcal{N}^{j+1}[x](u)$. If $u\in C$, then the list $\mathcal{L}(u)$ contains the element $\langle C,0\rangle$, and it will never be removed from $\mathcal{L}(u)$. 
		Consider the case where $u\notin C$. Let $v\in C$ such that $d_{G_{k-1}}^{(j+1)}(u,v) = d_{G_{k-1}}^{(j+1)}(u,C)$, and let $\pi(v,u)= \langle v=v_0,v_1,\dots,v_{s-1},v_s=u\rangle$ be a path between $u,v$ with weight $ d_{G_{k-1}}^{(j+1)}(u,C)$, where $s\leq j+1$.
		
		By the induction hypothesis, the list $\mathcal{L}(v_{s-1})$ of the vertex $v_{s-1}$ contains the IDs and the $j$-hop distance from all clusters in $\mathcal{N}^j[x](v_{s-1})$. 
		If $\langle C, d^{(j)}_{G_{k-1}}(v_{s-1},C)\rangle\notin \mathcal{L}(v_{s-1})$, then by the induction hypothesis $C\notin \mathcal{N}^j[x](v_{s-1})$. Therefore, there are at least $x$ other clusters $C'$ on the list $\mathcal{N}^j[x](v_{s-1})$, with distance $d^{(j)}_{G_{k-1}}(v_{s-1},C')$ smaller than $ d^{(j)}_{G_{k-1}}(v_{s-1},C)$ (we assume that there are no equalities, as ties are broken according to the clusters ID). For each such cluster $C'$, let $\pi(C',v_{s-1})$ be a path with up to $j$ hops of weight $d^{(j)}_{G_{k-1}}(v_{s-1},C')$. See Figure \ref{fig otherpath} for an illustration.
		By concatenating the edge $(v_{s-1},u)$ to every such path $\pi(C',v_{s-1})$, we obtain a path from $C'$ to $u$ of at most $j+1$ hops, and weight smaller than 
		$d^{(j+1)}_{G_{k-1}}(u,C)$. Thus $C\notin \mathcal{N}^{j+1}[x](u)$, contradiction. 
		
		\begin{figure}
			\centering
			\includegraphics[scale =0.12]{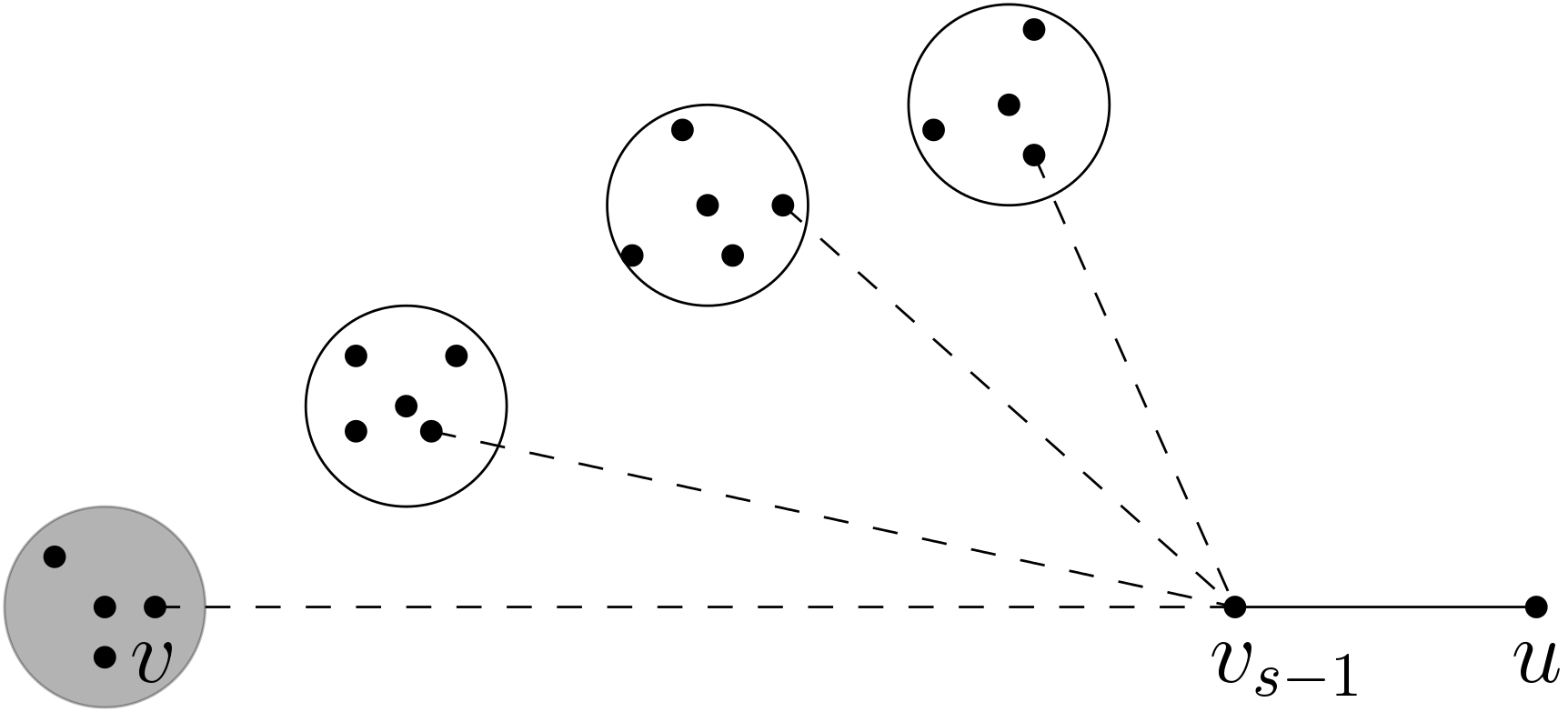}
			\caption{The path in ${G_{k-1}}$ from the cluster $C\in \mathcal{N}^{j+1}[x](u)$ 
				to $u$. The gray circle depicts the cluster $C$. The vertex $v_{s-1}$ is the neighbor of $u$ on the path from $C$ to $u$. The dashed lines depict paths of at most $j$ hops from clusters to $v_{s-1}$. The white circles depict clusters that are in $\mathcal{N}^{j}[x](v_{s-1})$. }
			\label{fig otherpath}
		\end{figure}
		
		Then, $\langle C, d^{(j)}_{G_{k-1}}(v_{s-1},C)\rangle \in \mathcal{L}(v_{s-1})$. Note that $d^{(j)}_{G_{k-1}}(v_{s-1},C) = d_{G_{k-1}}^{(j)}(v,v_{s-1})$, and that $d_{G_{k-1}}^{(j+1)}(v,u) = d_{G_{k-1}}^{(j)}(v,v_{s-1})+ \omega_{k-1}(v_{s-1},u)$. The edge $(v_{s-1},u)$ records the element $\langle C, d_{G_{k-1}}^{(j+1)}(v,u)\rangle$ to $m(u)$. If it is not recorded to $\mathcal{L}(u)$ by the end of step $j+1$, it is because there exist at least $x$ clusters that have $(j+1)$-hop bounded distance to $u$ that is smaller than $d_{G_{k-1}}^{(j+1)}(v,u)$, contradiction. Therefore, by the end of phase $i$, we have $\langle C, d_{G_{k-1}}^{(j+1)}(v,u)\rangle \in \mathcal{L}(u)$. 
	\end{proof}

	We are now ready to show that Algorithm \ref{alg parallel limited BF} can be used to detect popular clusters. 
	Recall that for the popular clusters detection, we set $x=deg_i+1$, $d=1$ and $S=P_i$.

	\begin{lemma}
		\label{lemma exporation v1}
	Given a weighted undirected graph $G_{k-1} = (V,E\cup H_{k-1},\omega_{k-1})$ on $n$ vertices, 
	a set of clusters $P_i$, distance parameter $\apdi$, degree parameter $deg_i$ and a hopbound $2\beta+1$, Algorithm \ref{alg parallel limited BF} requires $O(\beta{\log n})$ time and $O((|E|+|H_{k-1}|)\cdot n^\rho)$ processors. When the algorithm terminates, every cluster $C\in P_i$ is associated with an array $m(C)$ such that:

	\begin{enumerate}
		\item \label{item 1}{If $C$ is \textit{popular}, then the first $deg_i+1$ cells of $m(C)$ contain information regarding $deg_i+1$ clusters from $P_i$. One of them is $C$, and the other are neighboring clusters of $C$.}
		
		\item \label{item 2} {If $C$ is  \emph{unpopular} then $m(C)$ contains the identities and $(2\beta+1)$-hop bounded distances in $G_{k-1}$ between $C$ and all its neighboring clusters. In addition, the $(deg_i+1)$st  cell in $m(C)$ is empty.}

	\end{enumerate}
	\end{lemma}
	\begin{proof}

	By Corollary \ref{coro exporation rt}, executing Algorithm \ref{alg parallel limited BF} with $x= deg_i+1$, $d=1$ and $S= P_i$ requires $O(\beta{\log n})$ time using $O((|E|+|H_{k-1}|)\cdot deg_i)$ processors.
		
	Consider a cluster $C\in P_i$.
	Lemma \ref{lemma induc exploration} implies that when the propagation part terminates, for every vertex $v\in C$, the list $\mathcal{L}(v)$ contains the IDs of all clusters of $\mathcal{N}^{2\beta+1}[x](v)$, and the $(2\beta+1)$-hop bounded distances to them. 
	
	During the aggregation part of the algorithm, each vertex $v\in C$ copies $\mathcal{L}(v)$ to $m(C)$. When \cref{step sort} of Algorithm \ref{alg parallel limited BF} terminates, the first $deg_i+1$ records of the array $m(C)$ contain the IDs and the $(2\beta+1)$-hop bounded distances in $G_{k-1}$ from $C$ to the closest $deg_i+1$ clusters to $C$ (including $C$ itself).

	If the cluster $C$ has at least $deg_i$ neighboring clusters, then when the algorithm terminates, the first $deg_i+1$ cells in $m(C)$ contains records regarding at least $deg_i+1$ clusters (including $C$). Therefore,  the assertion of \cref{item 1} holds.
	
	If the cluster $C$ has less than $deg_i$ neighboring clusters, then the array $m(C)$ contains the ID and distance $\dgk(C,C')$ for each neighboring cluster $C'$ of $C$, i.e., each cluster $C'\in \Gamma(C)$. In addition, since the array $m(C)$ was sorted, the $(deg_i+1)$st cell of $m(C)$ is empty, and therefore the assertion of \cref{item 2} holds.

	\end{proof}

\subsubsection{Variant 2: A BFS Exploration to Depth $\mathbf{d}$}	\label{append var2}
In this section, we prove the correctness of Algorithm \ref{alg parallel limited BF} for the case where $d\geq 1$, $x = 1$ and $S\subseteq P_i$. For convenience, we refer to the initialization, i.e., \cref{step in1,step in2,step in3,step in4,step in5,step in6} of Algorithm \ref{alg parallel limited BF}, as pulse $0$ of the algorithm. 
We say that a cluster $C\in P_i$ is detected during a pulse $p\in [0,d]$ if $p$ is the first index such that the array $m(C)$ is not empty when pulse $p$ terminates.
Intuitively, we show that each pulse of the exploration is equivalent to one round of a BFS exploration in the virtual graph $\tilde{G}_i$. (See Section  \ref{sec superclustering} for its definition.)

For a cluster $C\in P_i$ and a collection of clusters $S\subseteq P_i$, define $d_{\tilde{G}_i}(C,S) = {\min\{d_{\tilde{G}_i}(C,C') \ | \ C'\in S \} }$.

\begin{lemma}\label{lemma bfs}
	Let $C\in P_i$. For any $p\in [0,d]$, the cluster $C$ is detected during pulse $p$ if and only if $d_{\tilde{G}_i}(C,S) = p$.
\end{lemma}
\begin{proof}
	The proof is by induction on the index of the pulse $p$. For $p=0$, the claim is trivial as only (and all)  clusters of $S$ are detected during pulse $0$.
	
	We assume that the claim holds for some $p\in [0,d-1]$ and prove it for $p+1$. 
	Consider some $C\in P_i$. 
	
	If 
	the cluster $C$ has been detected in pulse $p+1$, then the information that detected it was written by some cluster $C'$ that has been detected in pulse $p$. The cluster $C'$ wrote to its memory 
	at the beginning of pulse $p+1$
	that it has been detected. This information has traversed at most $2\beta+1$ hops and up to $\apdi$ distance before arriving at $C$. Therefore, we have $\dgk(C,C')\leq \apdi$, i.e., $d_{\tilde{G}_i}(C,C')=1$. By the induction hypothesis, since $C'$ has been detected in pulse $p$, we conclude that $d_{\tilde{G}_i}(C',S) = p$. It follows that  we have $d_{\tilde{G}_i}(C,S)=p+1$. 
	
	If 
	the distance $d_{\tilde{G}_i}(C,S)$ is equal to $ p+1$, then there exists a cluster $C'$ such that $d_{\tilde{G}_i}(C',S) = p$ and $d_{\tilde{G}_i}(C,C') = 1$, i.e., $\dgk(C,C')\leq \apdi$. Let $u\in C$ and $u'\in C'$ be a pair of vertices such that $\dgk(u,u') = \dgk(C,C')$, and let $\pi(u',u)$
	be the shortest $(2\beta+1)$-hops bounded path in $G_{k-1}$ from $u'$ to $u$. By the induction hypothesis, the cluster $C'$ has been detected during pulse $p$. Therefore, during the distribution part of pulse $p+1$, the vertex $u'$ wrote to its memory that it has been detected. 
	By arguments similar, though simpler than those used in Lemma \ref{lemma induc exploration} from Appendix \ref{append var1}, we have that during the propagation part of pulse $p+1$, all vertices of the path $\pi(u',u)$ are detected. Therefore, when the propagation part terminates, the vertex $u$ has been detected by the exploration. During the aggregation part of pulse $p+1$, the cluster $C$ learns that it has a detected vertex, and becomes detected. 
\end{proof}

Recall that by Corollary \ref{coro exporation rt}, Algorithm \ref{alg parallel limited BF} 
can be executed in $O(d\beta{\log n})$ time using $O((|E|+|H_{k-1}|)\cdot x)$ processors. Recall that here we have $x=1$. Together with Lemma \ref{lemma bfs} we derive the following corollary.

\begin{corollary}\label{coro bfs}
	Given a graph $G_{k-1}= (V,E\cup H_{k-1},\omega_{k-1})$, a set $P_i$, a subset $  S\subseteq P_i$, distance and hop threshold parameters $\apdi,2\beta_i+1$ and a depth parameter $d\geq 1$, Algorithm \ref{alg parallel limited BF} simulates a BFS exploration to depth $d$ from the set of sources $S$ in the graph $\tilde{G}_i = (P_i, \tilde{E})$ where $\tilde{E} = \{ (C,C') \ | \ \dgk(C,C')\leq \apdi \}$. The algorithm requires $O(d\beta {\log n})$ CREW PRAM time, and $O(|E|+|H_{k-1}|)$ processors. 
\end{corollary}

\section{Ruling Sets}\label{sec ruling}	 
	 In this section, we provide the details of the PRAM CREW model implementation of the algorithm of \cite{awerbuch1989network,sew,KuhnMW18} for constructing ruing sets. 
	 
	 We are given a weighted undirected graph $G_{k-1}$, a set $P_i$ of clusters, a set $W_i\subseteq P_i$ of popular clusters and distance and hop thresholds $\delta_i$ and $ h= 2\beta+1$, respectively. 
	 Recall that each vertex $v\in V$ has a unique ID in the range $\{0,1,\dots,n-1\}$. For every cluster $C\in P_i$ centered around a vertex $r_C$, let $I(C)$ be the binary representation of the ID of $r_C$ by exactly ${\log n}$ bits. Throughout this section, we will refer to the field $I(C)$ as the ID of the cluster $C$. (The notation $I(C)$ is introduced so that we can delete bits from the ID of $C$, without actually changing the ID.) 
	 Recall that  $\tilde{G}_i = (P_i, \tilde{E})$, where $\tilde{E} = \{ (C,C') \ | \ \dgk(C,C')\leq \apdi \}$.
	 The current algorithm constructs a $(3,2{\log n})$-ruling set $Q_i$ for the set $W_i$ of popular clusters, with respect to the graph $\tilde{G}_i$. (See Section \ref{sec superclustering} for the definition of $W_i$.)
	 
	 We begin with an intuitive description of the algorithm. The algorithm works recursively, using a divide-and-conquer approach. 
	 The input for the algorithm is a set of clusters $W_i$. 
	 Given a recursive invocation $\mathcal{A}$ with input $A=\{C_1,C_2,\dots\}$, let $h$ be the number of bits in $I(C_t)$ for every $C_t\in A$. Observe that for the initial invocation, we have $h= {\log n}$. 
	 If $h=0$, then return $B=A$. Otherwise, the set $A$ is partitioned into two sets $A_0,A_1$ according to the most significant bit in the ID of every cluster $C\in A$. For a cluster $C$, let $msb(C)$ be the most significant bit of $I(C)$. The set $A_0$ contains all clusters $C\in A$ with $msb(C) = 0$. The set $A_1$ contains all clusters $C\in A$ with $msb(C) = 1$. 
	 Then, for every cluster $C\in A$, the most significant bit of $I(C)$ is deleted.

	 The algorithm recursively computes ruling sets $B_0\subseteq A_0$, $ B_1\subseteq A_1$. 
	 All clusters in $B_0$ join the output set $B$. Then, a BFS exploration to depth $2$ is executed from all clusters in $B_0$ in the graph $\tilde{G}_i$, using Algorithm \ref{alg parallel limited BF} from Appendix \ref{append explorations} (with $d=2$ and $x=1$). Each cluster $C\in B_1$ that is detected by the exploration is removed from $B_1$. 
	 The clusters that remain in $B_1$ also join the output set $B$, which is then returned.

	 Let $\mathcal{T}$ be the recursion tree for the input $W_i$. For every recursive invocation $\mathcal{A}$ in the tree $\mathcal{T}$, the \emph{height} of $\mathcal{A}$ is defined to be distance from $\mathcal{A}$ to a leaf of $\mathcal{T}$ that is a descendant of $\mathcal{A}$. 
	 Observe that for every recursive invocation $\mathcal{A}$ with input $A$, the number of bits in $I(C)$ for all $C\in A$ is also the height of the invocation $\mathcal{A}$ in the tree $\mathcal{T}$.

	Note that there may be many recursive invocations that are executed in parallel. Therefore, we say that a cluster $C\in B_1$ is \textit{knocked-out} if it is detected by an exploration originated at \textit{some} cluster $C'$, not necessarily a cluster of $B_0$. Therefore, all these exploration can be executed in parallel, without additional processors. 
	See Figure \ref{fig:rectreeKnock} for an illustration. Observe that by allowing clusters to be knocked-out by other recursive invocations, the output of each invocation does not necessarily rule its input. However, we will show that throughout the algorithm, every cluster in $W_i$ has some cluster that rules it in the output of  \textit{some} recursive invocation.

	This completes the description of the algorithm. The pseudo-code of the algorithm is given in Algorithm \ref{alg pram ruling}. 
	 	 
	 \begin{figure}
	 	\centering
	 	\includegraphics[scale=0.13]{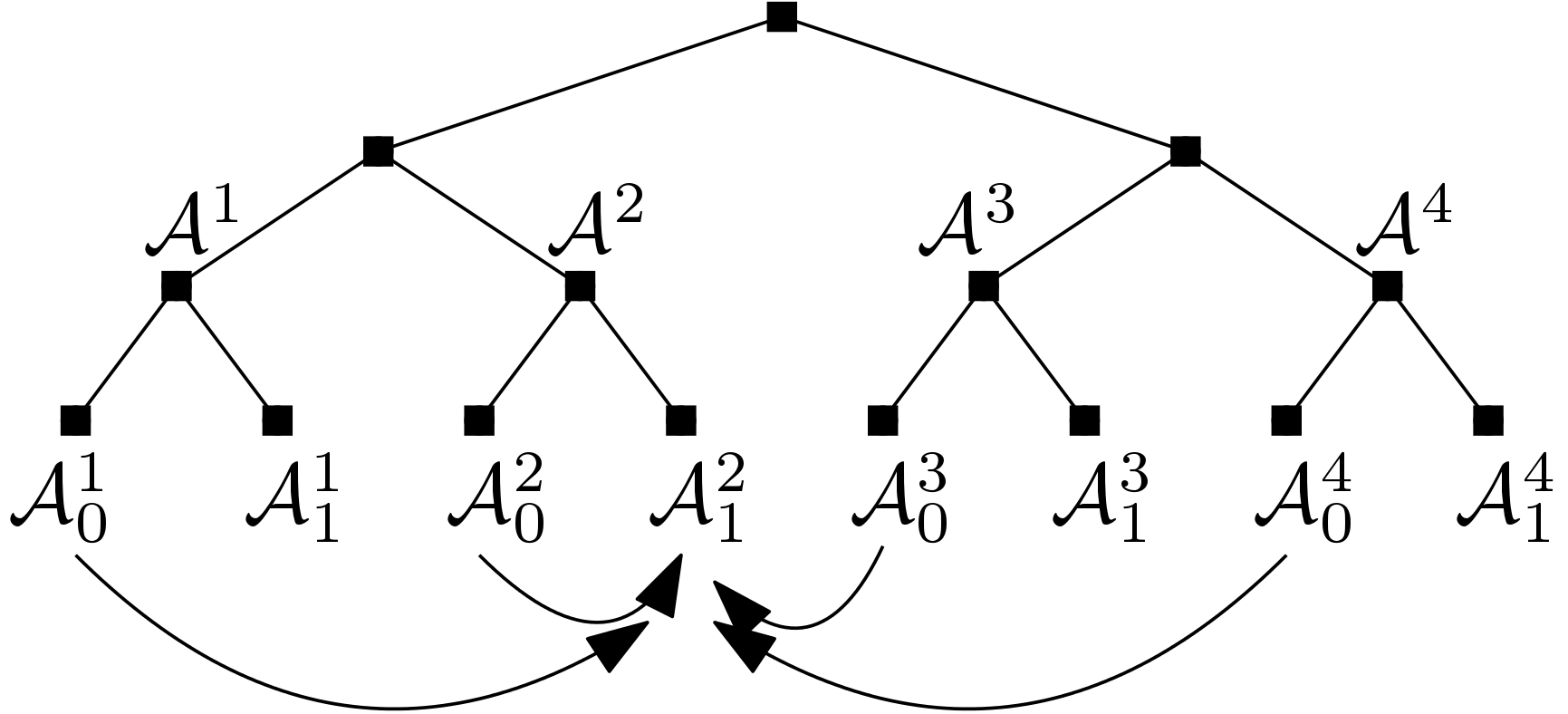}
	 	\caption{The recursion tree $\mathcal{T}$. Clusters from a recursive call $\mathcal{A}^j_1$ can be knocked-out by clusters from any recursive call $\mathcal{A}^{j'}_0$, possibly for $j\neq j'$. The arrows in the figure illustrate the knock-out messages sent from clusters in $B^1_0,B^2_0,\dots$, that can knock-out clusters in $B^2_1$.}
	 	\label{fig:rectreeKnock}
	 \end{figure}

	 \begin{algorithm}
	 	\caption{Ruling Set}
	 	\label{alg pram ruling}
	 	\begin{algorithmic}[1]
	 		\Statex \textbf{input:} The graph $G_{k-1}$ and the supergraph $\tilde{G}_i$, a collection of clusters $P_i$, a subset $A\subseteq P_i$ of popular clusters, a field $I(C)$ for every $C\in A$ and a distance parameter $\apdi$. 
	 		\Statex \textbf{output:} A set $B\subseteq A$.
	 		\State let $h$ be the number of bits in $I(C)$ for every $C\in A$
	 	\If {$h = 0$} return $A$
	 	\Else 
	 		\State $A_0 = \{C\in A \ | \textit{ the most significant bit in } I(C) \textit{ is } 0 \ \}$
	 		\State $A_1 = \{C\in A \ | \textit{ the most significant bit in } I(C) \textit{ is } 1 \ \}$
	 		\State for every $C\in A$, delete the most significant bit from $I(C)$
	 		\State compute recursively sets $B_0,B_1$ for the sets $A_0,A_1$
	 		\State execute a multiple sources BFS exploration in $\tilde{G}_i$ to depth $2$ from all clusters in $B_0$
	 		\State remove from $B_1$ all clusters detected by an exploration
	 		\State $B= B_0\cup B_1$
	 	\EndIf
	 	\end{algorithmic}
	 \end{algorithm}

	 The following lemmas summarize the properties of Algorithm \ref{alg pram ruling}. Recall that $\mathcal{T}$ is the recursion tree for the input $W_i$.
	 We begin by analyzing the running time and number of processors used by the algorithm.

	 \begin{lemma}
	 	\label{lemma ruling rt}
	 	Given the graph $G_{k-1} = (V,E\cup H_{k-1},\omega_{k-1})$, a set of clusters $P_i$, a subset $W_i\subset P_i$ and distance and hopbound parameters $\apdi, 2\beta+1$, respectively, Algorithm \ref{alg pram ruling} uses $O(|E|+|H_{k-1}|)$ processors and terminates in $O(\beta{\log^2 n})$ time. 
	 \end{lemma}
	 
	 \begin{proof}
	 There are ${\log n}+1$ levels in the recursion tree $\mathcal{T}$.
	 In every level of the recursion tree (other than the last level), the algorithm simulates a BFS exploration to depth $2$ in the graph $\tilde{G}_i$. By Corollary \ref{coro bfs}, this requires $O(\beta\cdot {\log n})$ time and $O(|E|+|H_{k-1}|)$ processors.
	 
	 In the last level of the recursion, there are $n$ recursive invocations. Each recursive invocation returns its input as its output. This can be done in $O(1)$ time using $O(n)$ processors. 
	 
	 Therefore, the running time of Algorithm \ref{alg pram ruling} is $O(\beta{\log^2 n })$, and it uses $O(|E|+|H_{k-1}|)$ processors. 
	 \end{proof}

	 The following two lemmas prove that $Q_i$ is indeed a $(3,2{\log n})$-ruling set for $W_i$ with respect to the graph $\tilde{G}_i$. (See Section \ref{sec superclustering} for the definition of $Q_i$ and $W_i$.)
	 
	 \begin{lemma}
	 	\label{lemma separated}
		For every recursive invocation $\mathcal{A}$, its output $B$ is $3$-separated w.r.t. the graph $\tilde{G}_i$. 	 	
 	\end{lemma}
	 \begin{proof}
	 	The proof is by induction on $h$, i.e., the height of the recursive invocation in the tree $\mathcal{T}$. For the base case, let $\mathcal{A}$ be a recursive invocation with height $h=0$ in the tree $\mathcal{T}$. 
	 	Observe that the path from the root of $\mathcal{T}$ to the leaf of $\mathcal{T}$ corresponds to the bits that were removed from $I(C)$ for every $C\in A$. Since $h=0$, we have that $I(C)$ contains now zero bits. Therefore, $A$ is a singleton, and the output $B=A$ is $3$-separated w.r.t. the graph $\tilde{G}_i$.

	 	Assume that the claim holds for some $h\in [0,{\log n}-1]$ and prove it for $h+1$. 
	 	
	 	Let $\mathcal{A}$ be a recursive invocation with input $A$ and height $h+1$. The invocation $\mathcal{A}$ deletes a single bit from $I(C)$ for every $C\in A$, and initiates two recursive invocations $\mathcal{A}_0$ and $\mathcal{A}_1$, each with height $h$. By the induction hypothesis, each set $B_0,B_1$ returned by an invocations $\mathcal{A}_0$, $\mathcal{A}_1$ is $3$-separated w.r.t. the graph $\tilde{G}_i$. Then, all clusters in $B_0$ initiate a BFS exploration to depth $2$ in $\tilde{G}_i$. Each cluster in $B_1$ that is detected by the exploration is removed from $B_1$. All clusters that remain in $B_1$ and all clusters in $B_0$ join the output $B$. Therefore, the output set $B$ is $3$-separated w.r.t. the graph $\tilde{G}_i$.
	 \end{proof}

	For every index $h\in [0,{\log n}]$, let $\mathcal{B}^h$ be the union of outputs of all recursive invocations with height $h$ in $\mathcal{T}$.

	\begin{lemma}
		\label{lemma ruls}
		For every index $h\in [0,\ell]$, for every cluster $C\in W_i$ there exists a cluster $C'\in \mathcal{B}^h$ such that $d_{\tilde{G}_i}(C,C') \leq 2h$.
	\end{lemma}

	 \begin{proof}
	 	We prove the claim by induction on the index $h$.

	 	For $h = 0$, for every invocation $\mathcal{A}$ with height $h=0$ the output $B$ is equal to the input $A$ of $\mathcal{A}$. Note that for every cluster $C\in W_i$ there exists a recursive invocation $\mathcal{A}$ with input $A$ and height $h=0$ such that $C\in A$. Thus the claim holds.

	 	Assume that the claim holds for some $h\in [0,{\log n}-1]$ and prove it for $h+1$. 
	 
	 	Let $C\in W_i$. By the induction hypothesis, there exists a cluster $C'\in \mathcal{B}^{h}$, with $d_{\tilde{G}_i}(C,C')\leq 2h$.
	 	
	 	If $C'\in \mathcal{B}^{h+1}$, the claim holds. 
	 	Otherwise, $C'$ has been detected by the exploration originated in a cluster $C''\in \mathcal{B}^{h+1}$. Since $C''$ knocked-out $C'$, we conclude that $d_{\tilde{G}_i}(C',C'')\leq 2$. Therefore, we have $d_{\tilde{G}_i}(C,C'')\leq d_{\tilde{G}_i}(C,C') +d_{\tilde{G}_i}(C',C'') \leq 2h+2 $. 
	 	
	 	It follows that in both cases, there exists a cluster in $\mathcal{B}^{h+1}$ with distance at most 
		$2(h+1)$ from $C$ w.r.t. the graph $\tilde{G}_i$.
	 \end{proof}
	 Observe that Lemma \ref{lemma ruls} implies that the output $Q_i$ for the initial input $W_i$ satisfies that for every cluster $C\in W_i$, there exists a cluster $C'\in Q_i$ with $d_{\tilde{G}_i}(C,C')\leq 2{\log n}$. As a corollary to Lemmas \ref{lemma ruling rt}, \ref{lemma separated} and \ref{lemma ruls} we have

	 \begin{corollary}\label{coro ruling}
	 	Given a weighted undirected graph $G_{k-1} = (V,E\cup H_{k-1},\omega_{k-1})$ on $n$ vertices, 
	 a set of clusters $P_i$, a distance parameter $\apdi$, and a hopbound $2\beta+1$, Algorithm \ref{alg pram ruling} uses $O(|E|+|H_{k-1}|)$ processors and $O(\beta{\log^{2}n})$ time, and returns a $(3,2{\log n})$-ruling set $Q_i$ for the set $W_i$ w.r.t. the graph $\tilde{G}_i = (P_i,\tilde{E})$, where $\tilde{E} = \{ (C,C') \ | \ \dgk(C,C')\leq \apdi \}$.
	 \end{corollary}

\section{Eliminating the Dependence on the Aspect Ratio}\label{sec reduc}
In this section we show that the reduction devised by Klein and Sairam \cite{KleinS97} can be used by our algorithm from Section \ref{sec hopset const} to eliminate the dependence of the hopbound $\beta$ and of the running time on the aspect ratio $\Lambda$. 
This reduction was also previously used by Elkin and Neiman \cite{ElkinN19} in their construction of hopsets, for the same purpose. 
The result of applying the reduction to our algorithm is summarized in Theorems \ref{theo reduc} and \ref{theorem reduc compute dist}. In this section we describe the algorithm for ordinary (not path-reporting) hopsets. In Appendix \ref{append red path} we extend this result also to the path-reporting case.

\subsection{Overview}
\label{sec reduc high} 

Fix a parameter $0<\epsilon \leq 1/2$. Recall that $k_0 = \lfloor{\log \beta}\rfloor$ and that $\lambda = \lceil{\log \Lambda}\rceil-1$.
For every index $k\in \krange$ we build a graph $\mathcal{G}^k$ that contains only edges with weight in the range $(\mindg,(1+\epsilon)\maxdg]$. 
This is done by deleting heavy edges, and grouping vertices that have short edges between them into supervertices, which we call \textit{nodes}. 
As a result, for every $k\in \krange$, the graph $\mathcal{G}^k$ has aspect ratio $O(n/\epsilon)$.

For every index $k\in \krange$, a $(1+\epsilon,\beta)$-hopset for the scale $(2^k,2^{k+1}]$ is computed for the graph $\mathcal{G}^k$, using our algorithm from Section \ref{sec hopset const}.
This is done in parallel for all $k\in \krange$. In addition, every node $X$ formed by the algorithm has a designated center vertex $x\str\in X$. A \textit{spanning star} centered around $x\str$ is added to a set of edges $S$, i.e., for every node $X$, the set $S$ contains the set of edges $\{ (x\str,x ) \ | \ x\in X\setminus \{x\str\} \}$. The weight of these star edges will be specified in the sequel. 

The ultimate hopset $H$ is constructed as follows. For every $k\in \krange$, let $\mch{k}$ be the single-scale $(1+\epsilon,\beta)$-hopset for the scale $(2^k,2^{k+1}]$ for the graph $\mathcal{G}^k$.
For every $k\in \krange$ and for every edge $(X,Y)$ with weight $d$ in the hopset $\mch{k}$, an edge $(x\str,y\str)$ between the respective centers of the nodes $X,Y$ is added to $H$, also with weight $d$. To ensure that the number of hops within nodes is also small, the set of (weighted) star edges $S$ is also added to $H$. This completes the overview of the reduction.

\subsection{Constructing $\mcg{k}$}
Klein and Sairam \cite{KleinS92} have shown that the graphs $\mcg{k_0},\mcg{k_0+1},\dots$ can be computed in the EREW PRAM model in $O({\log^3 n})$ time, using $O(|E|)$ processors. Their algorithm is based on combining parallel prefix computation with the connected components algorithm of Shiloach and Vishkin \cite{ShiloachV82}. As a byproduct of the connected components algorithm, for every node $U$ in a graph $\mcg{k}$, $k\in \krange$, a spanning tree $T_U$ that contains only edges of weight at most $\mindg$ is computed. 

The algorithm of Klein and Sairam does not assign centers to nodes. For our algorithm, we require each node $U$ to have a designated center $u\str\in U$. In Section \ref{sec spanning stars} we explain how the centers are selected in a way that ensures that the number of star edges is $O({n\log n})$. 

We now describe the construction of $\mcg{k}$, for every scale index $k\in \krange$. 
The nodes of $\mcg{k}$ are formed in the following way. 
Let $\mathcal{V}^{k}$ be the set of connected components in the graph $G$, after all edges of weight at least $\mindg$ are removed from it. (In other words, here we \emph{contract} all edges of weight at most $(\epsilon/n)\cdot 2^k$. )
The set $\mathcal{V}^{k}$ is the set of nodes of the graph $\mcg{k}$. 
For every pair of distinct nodes $X,Y\in \mathcal{V}^{k}$ such that there exists an edge $(x,y)\in E\cap (X\times Y)$, let $x\in X$ and $y\in Y$ be a pair of vertices such that $\omega(x,y)$ is minimal. If $\omega(x,y)\leq \maxdg$, an edge $(X,Y)$ is added to $\mcg{k}$ with weight
\begin{equation}\label{eq edge weight inter}
\mathcal{W}(X,Y) = \omega(x,y) + (|X|+|Y|)\cdot \mindg.
\end{equation}

Observe that the minimal edge weight in $\mcg{k}$ is at least $\mindg+2\mindg\geq (\epsilon/n)2^{k+1}$, and the maximal edge weight is at most $2^{k+1}+n\mindg \leq (1+\frac{\epsilon}{2})\maxdg$. Therefore, the aspect ratio of the graph $\mcg{k}$ is at most

\begin{equation}\label{eq reduc aspect}
 \frac{(1+\frac{\epsilon}{2})\maxdg}{(\epsilon/n)2^{k+1}} = O(n/\epsilon).
\end{equation} 

This completes the construction of $\mcg{k}$, for all $k\in [k_0,\lambda]$. We note that this process is equivalent to deleting edges of weight greater than $\maxdg$, and sequentially contracting edges of weight at most  $\mindg$, i.e., identifying their endpoints, while retaining the lightest among parallel edges.

\subsection{Selecting Node Centers}\label{sec spanning stars}
In this section, we select node centers in a way that guarantees that the number of star edges satisfies $|S|=O({n{\log n}})$.

For every node $X\in \mcv{k_0}$, an arbitrary vertex $x\str \in X$ is selected to be the center of the node $X$. 
For every vertex $z\in X\setminus \{x\str \}$, the edge $(x\str,z)$ is added to the set $S$. The weight of the edge will be specified in the sequel.

Consider a node $U\in \mcv{k}$, for some $k\in [k_0+1,\lambda]$. Let $X_1,X_2,\dots,X_t$ be the nodes in $\mcv{k-1}$ such that $|X_1|\geq |X_2|\geq |\dots|\geq |X_t|$ and also $U= \bigcup_{j\in [1,t]}X_{j}$. Note that the set of nodes on all different levels $k\in [k_0,\lambda]$ forms a laminar family. The center $x\str_1$ of the node $X_1$ is set to be the center of the node $U$. For every vertex $z\in U\setminus X_1$, the edge $(x_1\str,z)$ is added to the set $S$. 

We now specify the weights of star edges. 
Consider an index $k\in \krange$, and let $U\in \mcv{k}$, such that $x\str$ is the center of $U$. Let $z\in U\setminus \{x\str\}$.
Observe that since the vertices $x\str$ and $z$ both belong to the node $U$, we have that $d_G(x\str,z)\leq |U|\cdot (\epsilon/n)\cdot 2^{k}$. In their implementation of the reduction, Elkin and Neiman \cite{ElkinN19} set the weight of the edge $(x\str,z)$ to be $\mathcal{W}(x\str,z) = |U|\cdot (\epsilon/n)\cdot 2^{k}$.
This assignment of edge weight is enough for the basic variant of our algorithm. In Appendix \ref{append red path}, we show that the current reduction also enables us to obtain path-reporting hopsets with no dependence on the aspect ratio of the graph $G$. To support the path-reporting property, the weight of the star edges has to be assigned in a more careful way. Recall that a byproduct of computing a node $U$ is a spanning tree $T_U$. For every vertex $z\in U\setminus \{ x\str \} $, we set the weight of the edge $(z,x\str)$ to be the weight of the path in $T_U$ from $z$ to $x\str$. To compute the weight of this path, we use standard pointer-jumping techniques (see, e.g., \cite{ShiloachV82}). 
We note that for every $z\in U\setminus \{ x\str\}$, we have $d_{T_U}(z,x\str) < |U|\cdot \mindg$, and therefore the stretch analysis of the hopset provided in \cite{ElkinN19} also holds under the current assignment of weights.

Next, we show that this consistent way of selecting node centers ensures that $|S|\leq n{\log n}$.
For every node $U$ formed by the algorithm, let $\widehat{S}(U)$ be the set of edges with both endpoints in $U$ in the set $S$, i.e., $\widehat{S}(U) = \{(x,y)\in S\ | \ x,y\in U \}$.
The following lemma provides an upper bound on the number edges in $\widehat{S}(U)$ for every node $U$. It is later used to bound the size of the set $S$.

\begin{lemma}\label{lemma bound size s}
	For every scale $k\in \krange$, for every node $U\in \mcv{k}$, the set $\widehat{S}(U)$ contains at most $|U|\cdot {\log |U|}$ edges.
\end{lemma}

\begin{proof}
	The proof is by induction on the index $k$. For $k=k_0$ and for every node $U\in \mcv{k_0}$, the set $\widehat{S}(U)$ contains exactly $|U|-1$ edges. 
	
	We assume that the claim holds for some $k\in [k_0,{\lambda}-1]$, and prove that it also holds for $k+1$. 
	Let $U$ be a node in $\mcv{k+1}$, and let $X_1,X_2,\dots,X_t$ be the nodes in $\mcv{k}$ such that $U = \bigcup_{j\in [1,t]}X_{j}$, and $|X_1|\geq |X_2|\geq,\dots ,\geq |X_t|$. Denote $s= |U|$, and for every $j\in [1,t]$, denote $s_{j} = |X_{j}|$. 
	
	The set $\widehat{S}(U)$ contains all edges of the sets $\{\widehat{S}(X_j) \ | \ j\in [1,t] \}$, and also the edges added from the center of the node $X_1$ to all vertices of $U\setminus X_1$. 
	By the induction hypothesis, for every index $j\in [1,t]$, we have $|\widehat{S}(X_j)|\leq s_j{\log s_j}$. It follows that 
	\begin{equation}
	\begin{array}{lclclclclclclclc}
	|\widehat{S}(U)| &= &|U|-|X_1|+\sum_{j=1}^t \widehat{S}(X_{j}) \\
	&= & \sum_{j=2}^t s_j +\sum_{j=1}^t s_j{\log s_j} \\
	& \leq & s_1{\log s_1}+\sum_{j=2}^t s_j(1+{\log s_j} )\\
	& \leq & s_1{\log s_1}+\sum_{j=2}^t s_j{\log (2s_2)} & &( note: s\geq s_1+s_2\geq 2s_2) \\
	& \leq & s_1{\log s}+(s-s_1){\log s} &\leq & s{\log s}.
	\end{array}
	\end{equation} 
\end{proof}

Observe that for every $k\in \krange$, every node $U\in \mcv{k}$ is fully contained in a node of the set $\mcv{ \lambda}$. Therefore, by convexity of the function $f(x) = x{\log x}$, Lemma \ref{lemma bound size s} implies that the number of edges in $S$ satisfies

\begin{equation}
\label{eq bound s}
|S| = \sum\limits_{U\in \mcv{\lambda}}|\widehat{S}(U)|\ \leq
\sum\limits_{U\in \mcv{ \lambda}}|U|{\log |U|} \ \leq\ n{\log n}. 
\end{equation}

Note also that it follows that the total number of pairs $(u,U)$, such that $U$ is a node computed by the algorithm and $u\in U$ is $O(n{\log n})$. As a part of the computation of the set $S$ of star edges, the algorithm can also compute an array $B$ of size $O(n{\log n})$ that records all these pairs.

\subsection{Sketch of Analysis}
\label{sec reduc anal}

In this section we provide a short sketch of the analysis of the properties of the resulting hopset $H$. 
For every index $k\in \krange$, the hopset $\mch{k}$ is used to approximate distances of pairs of vertices $u,v\in V$ with $d_G(u,v)\in (2^k,2^{k+1}]$. Consider an index $k\in \krange$. If there are no edges with weight in the range $(2^k/n,2^{k+1}]$, then there is no pair of vertices with distance in the range $(2^k,2^{k+1}]$. In this case, the hopset $\mch{k}$ is \textit{redundant}. Define a scale $k$ to be \textit{relevant} if there exists an edge $(u,v)\in E$ with weight $\omega(u,v)\in (2^k/n,2^{k+1}]$.
Let $K$ be the set of all relevant scale indexes from the range $\krange$. 

Observe that every edge $(u,v)\in E$ can induce a logarithmic number of relevant scales, and so 
\begin{equation}
\label{eq bound k}
|K| = O(|E|\cdot {\log n}).
\end{equation} 
Our algorithm constructs a $(1+\epsilon,\beta)$-hopset $\mch{k}$ for the scale $(2^k,2^{k+1}]$ only for the graphs $\mcg{k}$ with $k\in K$. Recall that for every edge $(X,Y)$ with weight $d$ in some hopset $\mch{k}$, the algorithm adds to our ultimate hopset $H$ the edge $(x\str,y\str)$ between the respective centers of $X,Y$, also with weight $d$. In addition, the set $H$ also contains the set $S$ of star edges. 

Elkin and Neiman \cite{ElkinN19} have provided a detailed analysis of the reduction in the centralized model (see Section $4$ in \cite{ElkinN19}). In particular, they have shown that the hopset $H$ is a $(1+6\epsilon,6\beta+5)$-hopset for $G$. See Lemma 4.3. in \cite{ElkinN19}.

\subsubsection{Size}
For every $k\in K$, let $m_k$ denote the number of edges in $\mcg{k}$, and let $n_k$ denote the number of  nodes in $\mcv{k}$ that are not isolated in $\mcg{k}$. By arguments similar to those used in \cite{KleinS97,Cohen97,ElkinN19}, one can show that the number of non-isolated nodes in all graphs $\{\mcg{k} \ | \ k\in K \}$ is at most 
\begin{equation}\label{eq short nk}
\sum_{k \in K} n_k \leq \sum_{k \in K} |\mcv{k} | = O(n{\log n}). 
\end{equation}

Note that $\epsilon^{-1} \leq n$ (otherwise, $\beta > n$ and the result becomes meaningless). Since every edge $e\in E$ belongs to at most $O({\log (n/\epsilon)})= O({\log n})$ relevant scales, the number of edges in all graphs $\{\mcg{k} \ | \ k\in K \}$ is at most 
\begin{equation}\label{eq short nh}
\sum_{k \in K} m_k \leq  O(|E|\cdot {\log n}). 
\end{equation}

Recall that for every $k\in K$, the hopset $\mch{k}$ is a \emph{single-scale} hopset. Therefore, by \cref{eq hk concluding} we have $|\mch{k}| \leq n_k^{1+1/\kappa}$.
By \cref{eq short nk,eq bound s}, the size of the hopset $H$ satisfies 
\begin{equation}\label{eq short bound hK}
|H| =
|S| + \sum_{k\in K} |\mch{k}|
\leq n{\log n} + \sum_{k\in K} n_k^{1+1/\kappa}
\leq n{\log n} + n^{1/\kappa}\sum_{k\in K} n_k
=O\left(\nfrac\cdot {\log n} \right).
\end{equation}

\subsubsection{Computational Complexity}

Computing the graphs $\{\mcg{k}\ | \ k\in K \}$ can be done using $O(|E|)$ processors in $O({\log^3 n})$ time (see, e.g., \cite{KleinS92}). 

We next explain how to efficiently compute the edge sets of graphs $\mcg{k}$, for $k\in K$.
Every edge $(x,y)\in E$ participates in $O({\log (n/\epsilon)}) = O({\log n})$ scales, and on each scale it connects a pair of nodes $(X,Y)$, $x\in X$, $y\in Y$. We allocate an array $A$ of size $O(|E|\cdot {\log n})$, and also, designate $O({\log n})$ processors for every $e\in E$, and $O({\log n})$ slots in the array $A$ associated with $e$. 
These processors write in parallel into $A$ tuples $\langle (X,Y), k, \omega (e),e\rangle$, for all possible $(X,Y)\in \mcg{k}$, for some $k$, such that $x\in X$, $y\in Y$, and $e$ is relevant for scale $k$ (i.e., its weight is within interval $(\mindg,\maxdg]$). 

This array is then sorted according to the superedges $(X,Y)$, and within each segment that has to do with a specific superedge $(X,Y)$, 
by the scale index $k$, and within each such segment, by the weight of $e$. 
For each superedge $(X,Y)$ and scale $k$, the lightest edge $e$ is selected, and its weight $\omega(e)$ determines the weight $\mathcal{W}(X,Y)$ of the superedge $(X,Y)$ in $\mcg{k}$. 
This completes the description of this computation. It can be carried over in $O({\log n}) $ time using $\tilde{O}(n+|E|)$ processors.



Once all the graphs $\{\mcg{k}\ | \ k\in K \}$, node centers are selected. We now provide an implementation of the procedure that selects node centers. 
Observe that the nodes in $\{\mcv{k} \ | \ k\in K \}$ form a laminar family, and so there are at most $2n-1$ distinct nodes computed throughout the algorithm. A \textit{nodes graph} $\bar{\mcg{}}$ is computed. The vertices of this graph are the (up to $2n-1$) nodes formed by the algorithm. Let $\widehat{k}_0$ be the smallest index in $K$.  
For every node $U\in \bigcup_{k\in K\setminus \{ \widehat{k}_0 \} }\mcv{k}$, let $X$ be the largest lower-scale node $X\subset U$, if exists. Recall that the center of $U$ is set to be the center of $X$. The node $U$ defines $X$ as its parent, i.e., it writes to its memory $p(U)= X$. The node $U$ also adds the edge $(U,X)$ to the graph $\bar{\mcg{}}$. Every node $X\in \mcv{\widehat{k}_0}$ writes $p(X) = X$, and selects the smallest ID vertex $x\str\in X$ to be its center. See Figure \ref{fig:nodesforest} for an illustration.

Observe that every node constructed by the algorithm has (at most) one parent. Moreover, every node can be a parent of at most one node. Therefore, the maximal degree in $\bar{\mcg{}}$ is at most $2$.
In addition, the number of vertices in a node $U$ is always greater than the number of vertices of its parent $X$, and so $\bar{\mcg{}}$ does not contain cycles. Hence, $\bar{\mcg{}}$ is a forest of paths, and the maximal length of a path in $\bar{\mcg{}}$ is $2n-1$. 
\begin{figure}
	\centering
	\includegraphics[scale=0.22]{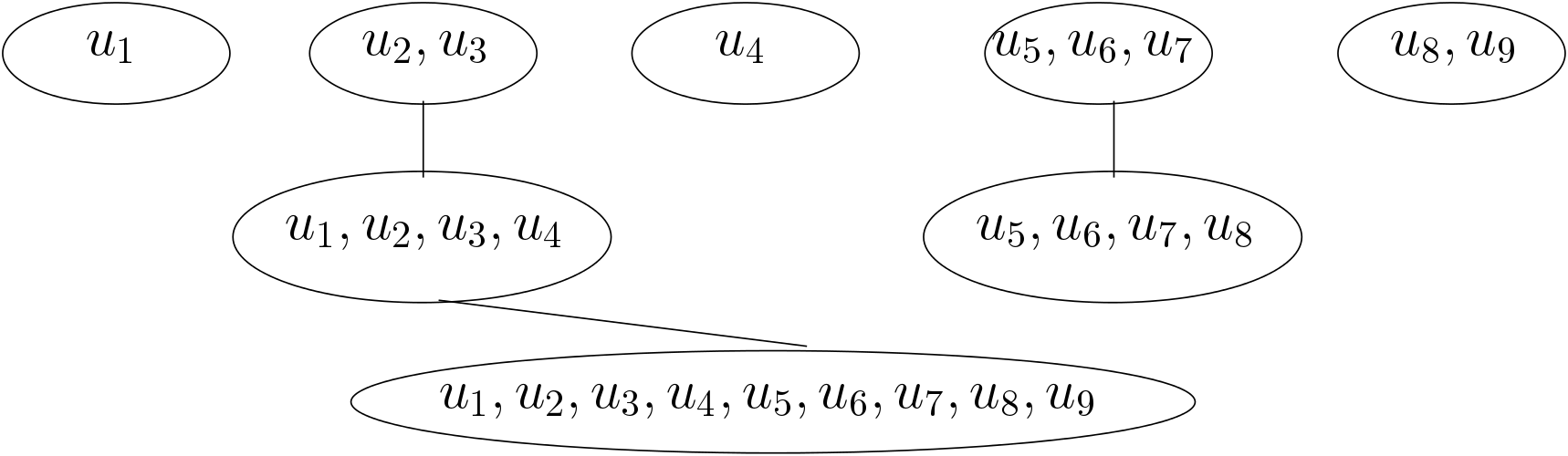}
	\caption{The nodes forest $\bar{\mcg{}}$. Each node constructed by the algorithm is a vertex in this graph. Edges are drawn from each node $U$ to the largest lower-scale node $X\subset U$, if such a node exists. }
	\label{fig:nodesforest}
\end{figure}		

A standard pointer-jumping algorithm (see, e.g., Section \ref{sec pointer} for an example) is used to compute the centers of all nodes. For ${\log{2n}}$ iterations, every node $U$ in $\bar{\mcg{}}$ writes $p(U) = p(p(U))$. When the algorithm terminates, the pointer $p(U)$ contains the identity of the root of its tree in the forest $\bar{\mcg{}}$, and the node $U$ selects the center of $p(U)$ to be its center.

The graph $\bar{\mcg{}}$ can be constructed using $O(n)$ processors in $O({\log n})$ time. The pointer-jumping procedure requires additional $O({\log n})$ time, using $O(n)$ processors. 
It follows that overall, centers can be selected in $O({\log n})$ time, using $O(n)$ processors. This completes the description of the procedure that selects node centers.

Finally, we discuss the resources used to compute the hopsets $\mch{k}$ for all $k\in K$, and the final hopset $H$. 
We note that in the centralized model, one can construct a single-scale hopset efficiently for a particulate scale $(2^k,2^{k+1}]$.
On the other hand, in parallel models, in order to construct a hopset that takes care of distances $(2^k,2^{k+1}]$ for $\mcg{k}$, one needs to first construct hopsets for all lower scales for $\mcg{k}$. 
(Note, however, that one does not need to compute hopsets for $\mcg{k_0},\mcg{k_0+1},\dots \mcg{k-1}$, before computing the hopset for $\mcg{k}$.)
These lower scales hopsets are only used to compute the hopset for the scale $(2^k,2^{k+1}]$, and do not belong to the ultimate hopset.
Since the aspect ratio of every graph $\mcg{k}$ is $O(n/\epsilon)$, 
by Lemma \ref{lemma rt}, for every $k\in K$, the hopset $\mch{k}$ for the graph $\mcg{k}$ can be constructed in $O(({\log \kappa\rho}+1/\rho)\beta{\log ^{3}n})$ time using $O((m_{k}+\nfrac_k)n^\rho)$ processors. 
By \cref{eq short nk,eq short nh} we have that the total number of processors used to compute the hopsets is at most 
\begin{equation*}
\sum_{k\in K}O((m_{k}+\nfrac_k)n^\rho)
 = O(n^\rho  \cdot (|E|\cdot {\log n}
 +
 \nfrac{\log n})).
\end{equation*}

Computing $H$ requires adding all star edges $S$ to an empty set $H$, and also for every edge $(X,Y)$ that belongs to a hopset $\mch{k}$ for some $k\in K$, the edge between the respective centers of $X,Y$ is also added to $H$. This can be performed in $O(1)$ time using $O\left(\nfrac\cdot {\log n} \right)$ processors. This completes the analysis of the computational complexity of the algorithm. Observe that the running time and work used by the entire algorithm are dominated by the respective time and work requires to compute the hopsets $\{ \mch{k}\ | \ k\in K \}$.

The following theorem summarizes the properties of the entire algorithm.

\begin{theorem}
	\label{theo reduc}
	Given a weighted undirected graph $G=(V,E,\omega)$ on $n$ vertices, and parameters $0<\epsilon<1/2$, $\kappa =2,3,\dots$, and $0<\rho<1/2$, our algorithm deterministically computes a $(1+\epsilon,\beta)$-hopset $H$ of size at most $O\left(\nfrac\cdot {\log n} \right)$ 
	in $
	O(({\log \kappa\rho}+1/\rho)\beta{\log^3 n}) $ time in the PRAM CREW model using 
	$O\left(n^\rho{\log n}\left(|E| + \nfrac\right)\right) $ processors, where 
	
	\begin{equation}
	\label{eq beta no dep}
		\beta =O\left( \frac{ {\log^2 n} ({\log \kappa\rho} + 1/\rho) }{\epsilon} \right)^{\pramell}.
	\end{equation}

\end{theorem}

Recall that the hopset $H$ can be used to compute approximate distances in the graph $G$. Given a set of sources $\mathcal{S}\subseteq V$, one can execute $|\mathcal{S}|$ parallel Bellman-Ford explorations in the graph $G\cup H$, each limited to $\beta$ hops, and solve the $(1+\epsilon)$\textit{-approximate-multiple-source-shortest distance} (aMSSD) problem. 
(When $|\mathcal{S}| =1$, this approach solves the 
$(1+\epsilon)$\textit{-approximate-single-source-shortest distance} (aSSSD) problem.)
Executing $|\mathcal{S}|$ Bellman-Ford exploration limited to $\beta$ hops can be performed in $O(\beta{\log n})$ time, using $O(|\mathcal{S}|)$ processors to simulate every vertex and every edge of $E\cup H$.

\begin{theorem}\label{theorem reduc compute dist}
	Given a weighted undirected graph $G=(V,E,\omega)$ on $n$ vertices, parameters $0<\epsilon<1/2$, $\kappa =2,3,\dots$, and $0<\rho<1/2$, and a set of sources $\mathcal{S}\subseteq V$,
	our deterministic algorithm computes $(1+\epsilon)$-approximate distances for all pairs of vertices in $\mathcal{S}\times V$ 
	in $
	O(({\log \kappa\rho}+1/\rho)\beta{\log^3 n}) $ time in the PRAM CREW model using 
	$O\left((n^\rho{\log n}+|\mathcal{S}|)\cdot \left(|E| + \nfrac\cdot {\log n}\right)\right) $ processors, where 
	$$\beta =O\left( \frac{{\log^2 n} ({\log \kappa\rho} + 1/\rho) }{\epsilon} \right)^{\pramell}.$$
\end{theorem}


\section{Path-Reporting Hopsets Without Dependency on the Aspect Ratio }
\label{append red path}

In this section, we show how to eliminate the dependency on the aspect ratio in the context of path-reporting hopsets. In particular, 
given a graph $G= (V,E)$, a source $s\in V$ and a parameter $0<\epsilon <1$, 
we compute in polylogarithmic time a \textit{$(1+\epsilon)$-approximate-single-source-shortest-path} tree (henceforth, $(1+\epsilon)$-SPT) $T= (V,E_T)$ with $E_T\subseteq E$.
For every vertex $v\in V$, the distances in the tree $T$ will satisfy 
$$d_T(s,v)\leq (1+\epsilon)d_G(s,v).$$
Moreover, every vertex $v\in V$ will know\footnote{In other words, the processor associated with the vertex $v$ will know this information.} its parent with respect to $T$ and also the distance $d_T(s,v)$ between $v$ and the source $s$.

To achieve this, we combine the Klein-Sairam reduction provided in Appendix \ref{sec reduc}, with the construction of the path-reporting hopset of Section \ref{sec path-reporting}.
Namely, we execute the algorithm from Appendix \ref{sec reduc}, but use the path-reporting hopset construction from Section \ref{sec path-reporting} instead of the construction provided in Section \ref{sec hopset const}.
Once the path-reporting $(1+\epsilon,\beta)$-hopset $H$ is computed, a Bellman-Ford exploration is executed from the source $s$ to depth $\beta$ in the graph $G\cup H$. As a result, a tree $\mathcal{T}$ is formed. For every vertex $v\in V$, let $p(v),d(v)$ be the parent and distance of $v$ from the source $s$, respectively. Note that the tree $\mathcal{T}$ may contain edges that belong to $H$, and not to the original graph $G$. A procedure that replaces edges of 
$H$ with paths from $G$ is executed, to obtain a $(1+\epsilon)$-SPT $T=(V,E_T)$, with $E_T\subseteq E$. This procedure is more complicated than the procedure used to retrieve paths in Section \ref{sec path-reporting}, because we must replace edges between centers of nodes computed by the reduction (including star edges) with edges of the original graph.

\subsection{Constructing the Hopset $H$}\label{sec const h reduc path}
As in Appendix \ref{sec reduc}, define $k$ to be \textit{relevant} if there exists an edge $(u,v)\in E$ with weight in the range $(2^k/n,2^{k+1}]$.
Let $K$ be the set of all relevant scale indexes from the range $\krange$. 
We begin by computing the graphs $\{\mcg{k}\ | \ k\in K\}$, as in Appendix \ref{sec reduc}. For every $k\in K$, we execute the algorithm provided in Section \ref{sec path-reporting} to construct a path-reporting $(1+\epsilon,\beta)$-hopset $\mch{k}$ for the graph $\mcg{k}$. In Appendix \ref{sec reduc}, the hopset was used only to estimate distances. Therefore, for every $k\in K$, it sufficed to only use the hopset constructed for the scale $(2^k,2^{k+1}]$ in $\mcg{k}$. 
To support path reporting, the hopset $\mch{k}$ must also be a hopset for all scales lower than $k$.

Recall that for every edge $(X,Y)$ with weight $d$ in the hopset $\mch{k}$, an edge $(x\str,y\str)$ between the respective centers of the nodes $X,Y$ is added to $H$, also with weight $d$.
Recall also that every edge $(X,Y)$ in the path-reporting hopset $\mch{k}$ is associated with a \textit{memory path}, which in this case is a path of nodes. Instead of maintaining the memory path as a list of nodes, we maintain the list of corresponding node centers.

Recall that the nodes of each graph $\mcg{k}$ are computed by removing all edges of weight at least $\mindg$ from $G$, and finding the connected components in the new graph. The connected components are computed using the algorithm of Shiloach and Vishkin \cite{ShiloachV82}. We note that a byproduct of this computation, is a spanning tree $T_U$ for every node $U$ computed by the algorithm. We use the procedure provided in Appendix \ref{sec spanning stars} to select a designated center for each node and compute the set of star edges $S$. The set $S$ is also added to the hopset $H$. For every node $U$ centered around a vertex $u\str$, we orient its tree $T_U$ such that every vertex $u\in U\setminus\{u\str\}$ knows 
its parent $p_{u\str}(u)$ in $T_U$.
We also use a pointer-jumping algorithm to compute for each vertex $u\in U\setminus\{u\str\}$ its distance from the center $u\str$ in the tree $T_U$.

Observe that the set $H$ contains two types of edges. The first type is edges between pairs of node centers, such that there is an edge between their corresponding nodes in some hopset $\mch{k}$, for some $k\in K$. For brevity, we will henceforth refer to edges of the first type as \emph{hop-edges}. The second type is star edges. 
 
We now analyze the stretch and size of the hopset $H$, as well as the complexity of the algorithm. In Appendix \ref{sec pr reduc retriving} we show how to employ $H$ for the construction of a $(1+\epsilon)$-SPT for the original graph $G$.

\subsubsection{Stretch}
By a similar argument to the one provided in Appendix \ref{sec reduc anal}, we have that $H$ is a $(1+\epsilon,\beta)$-hopset for the graph $G$, with $\beta$ given by \cref{eq beta no dep}. 

\subsubsection{Size}
As in Appendix \ref{sec reduc}, denote by $n_k$ the number of not isolated nodes in $\mcg{k}$, for all $k\in K$. Recall that by \cref{eq short nk} we have $\sum_{k\in K} n_k = O(n{\log n}) $.

For every $k\in K$, $\mch{k}$ is a hopset for $\mcg{k}$. 
By Theorem \ref{theorem final pr-hopset} and since the aspect ratio of $\mcg{k}$ is $O(n/\epsilon)= O(n^2)$, we have that $$|\mch{k}| \leq O(n_k^{1+1/\kappa}{\log n}).$$ 
By arguments similar to those used in \cref{eq short bound hK}, we have that 
\begin{equation}
\label{eq short bound hK path}
|H| = |S|+ \sum_{k\in K} |\mch{k}| = O(\nfrac {\log ^2 n}).
\end{equation}

\subsubsection{Computational Complexity}

Computing the graphs $\{\mcg{k}\ | \ k\in K \}$ can be done using $O(|E|)$ processors in $O({\log^3 n})$ time (see, e.g., \cite{KleinS92}). 
Selecting nodes centers can be done in $O({\log n})$ time, using $\tilde{O}(n)$ processors, using the procedure described in Appendix \ref{sec reduc} (see the discussion in Appendix \ref{sec reduc anal} for details).

By Theorem \ref{theorem final pr-hopset}, for every index $k\in K$, the path-reporting hopset $\mch{k}$ can be constructed in $O(({\log \kappa\rho}+1/\rho)\beta{\log^3 n})$ time in the PRAM CREW model using $\beta\cdot n^\rho \cdot O({\log n})^{\pramell}$ processors to simulate every vertex and every edge of the graph $\mcg{k}$, and every edge of the hopset $\mch{k}$. 
By \cref{eq short nk}, the number of nodes in all graphs $\{\mcg{k} \ | \ k\in K \}$ is at most 
$\sum_{k \in K}| \mcv{k} | \leq O(n{\log n})$.
By \cref{eq short nh}, the number of edges in all these graphs is at most $O(|E|\cdot {\log n})$. By \cref{eq short bound hK path}, the number of edges in all hopsets $\{\mch{k} \ | \ k\in K \}$ is $O(\nfrac {\log ^2 n})$. Hence, all the hopsets $\{ \mch{k} \ | \ k\in K \} $ can be constructed in parallel in $O(({\log \kappa\rho}+1/\rho)\beta{\log^3 n})$ time, 
using $(n \log n +  |E| \cdot \log n + n^{1+1/\kappa} \log^2 n) \cdot \beta n^\rho \cdot O(\log n)^{\log\kappa \rho + 1/\rho -1}) =O((|E| \cdot \log n + n^{1+1/\kappa} \cdot \log^2 n) n^\rho \cdot O(\log n)^{\log \kappa \rho + 1/\rho -1}$ processors.

Once all hopsets  $\{\mch{k} \ | \ k\in K \}$ are constructed, the hopset $H$ is computed. For every edge $(X,Y)$ in a hopset $\mch{k}$, for some $k\in K$, a single edge is added to $H$. This can be done in $O(1)$ time using $O(\nfrac {\log ^2 n})$ processors.

It follows that the number of processors used by the algorithm is bounded by 
\begin{equation*}
(|E|+\nfrac)\cdot n^\rho\cdot\left({\log n}\right) ^{O({\log \kappa\rho}+1/\rho)} . 
\end{equation*}

 The following theorem summarizes the properties of the hopset $H$.

\begin{theorem}
	\label{theo path reduc}
	Given a weighted undirected graph $G=(V,E,\omega)$ on $n$ vertices, and parameters $0<\epsilon<1/2$, $\kappa =2,3,\dots$, and $0<\rho<1/2$, our algorithm deterministically computes a path-reporting $(1+\epsilon,\beta)$-hopset $H$ of size at most $O\left(\nfrac\cdot {\log^2 n} \right)$ 
	 in $O(({\log \kappa\rho}+1/\rho)\beta{\log^3 n})$ time in
	 the PRAM CREW model using 
	$(|E|+\nfrac)\cdot n^\rho\cdot\left({\log n}\right) ^{O({\log \kappa\rho}+1/\rho)} $ processors, where 
	$$\beta =O\left( \frac{{\log^2 n} ({\log \kappa\rho} + 1/\rho) }{\epsilon} \right)^{\pramell}.$$
\end{theorem}

\subsection{Retrieving Paths}\label{sec pr reduc retriving}
In this section, we show how one can use the hopset $H$ described in Appendix \ref{sec const h reduc path} to construct a $(1+\epsilon)$-SPT $T=(V,E_T)$ where $E_T\subseteq E$. First, a parallel Bellman-Ford exploration is executed in the graph $\mathcal{G}= (V,E\cup H)$ from the source $s$ to depth $\beta$. As a result, a $(1+\epsilon)$-SPT $\mathcal{T}$ is computed. For every vertex $v\in V$, let $p(v),d(v)$ be the parent and distance of $v$ from $s$ obtained by the exploration, respectively. Observe that for all $v\in V$, we have $d(v)\leq (1+\epsilon)d_G(s,v)$ and also $d(v)=d(p(v))+\omega(p(v),v)$, where $ \omega(p(v),v)$ is the weight of the edge $(p(v),v)$.

We now replace hopset edges of $\mathcal{T}$ with paths that belong to the original graph $G$. This is done in three steps. First, we replace hop-edges from $\bigcup_{k\in K}\mch{k}$ with paths that contain edges (of $\bigcup_{k\in K}\mcg{k}$) between neighboring node centers. Then, we replace edges between neighboring node centers with star edges and original graph edges. Finally, we replace star edges with original graph edges. 
See Figures \ref{fig steps} for an illustration. 

\begin{figure}
	\begin{subfigure}{.48\textwidth}
		\centering
		\includegraphics[scale=0.15]{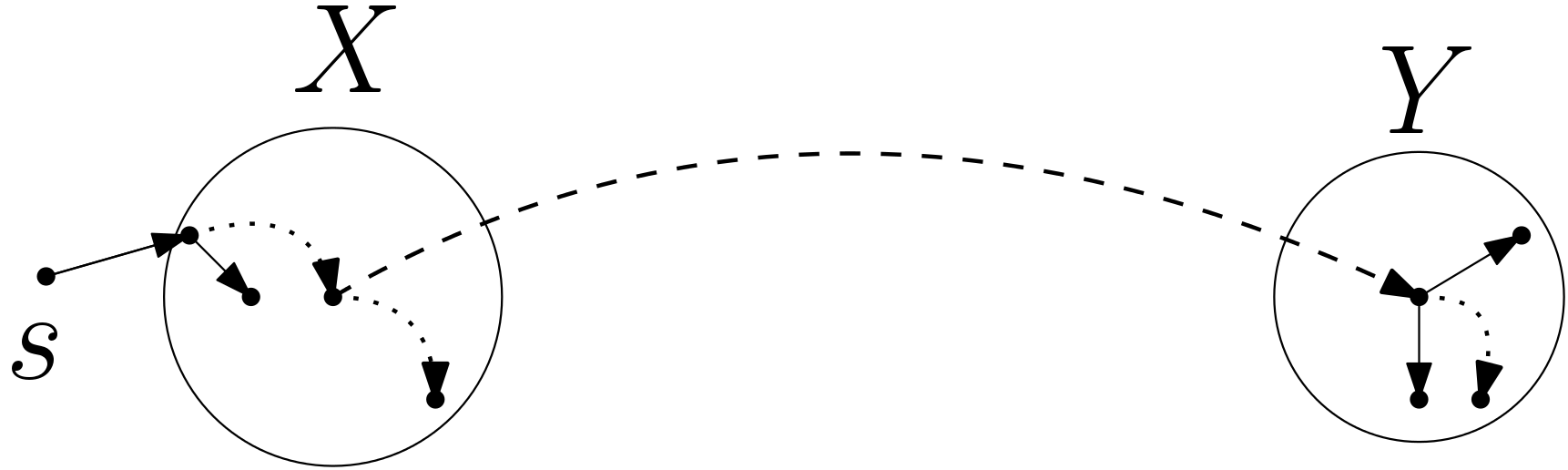}
		\caption{The output $\mathcal{T}$ of the Bellman-Ford exploration. The tree $\mathcal{T}$ contains original graph edges (solid arrows), hop-edges (dashed arrows) and star edges (dotted arrows).	}
		\label{fig step0}
	\end{subfigure} \hfill
	\begin{subfigure}{.48\textwidth}
			\centering
		\includegraphics[scale=0.15]{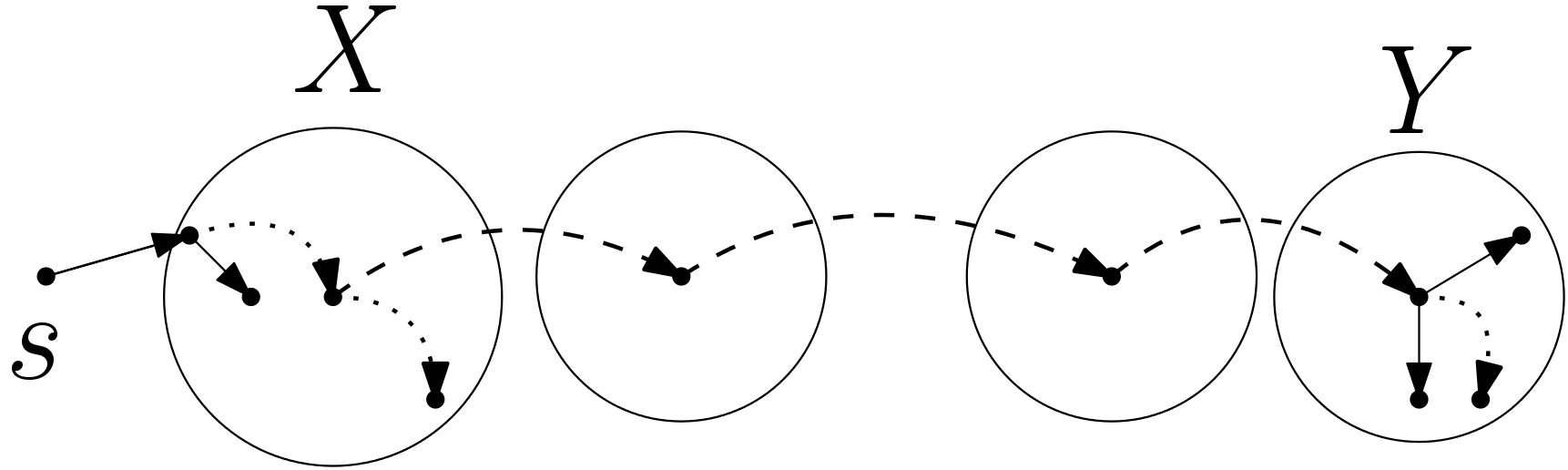}
		\caption{The tree $\mathcal{T}$ after the first step. The hopset edge between the node centers of $X,Y$ is replaced by a path that contains edges between neighboring node centers.}
		\label{fig step1}
	\end{subfigure}
	\newline
	\begin{subfigure}{.48\textwidth}
		\centering
	\includegraphics[scale=0.15]{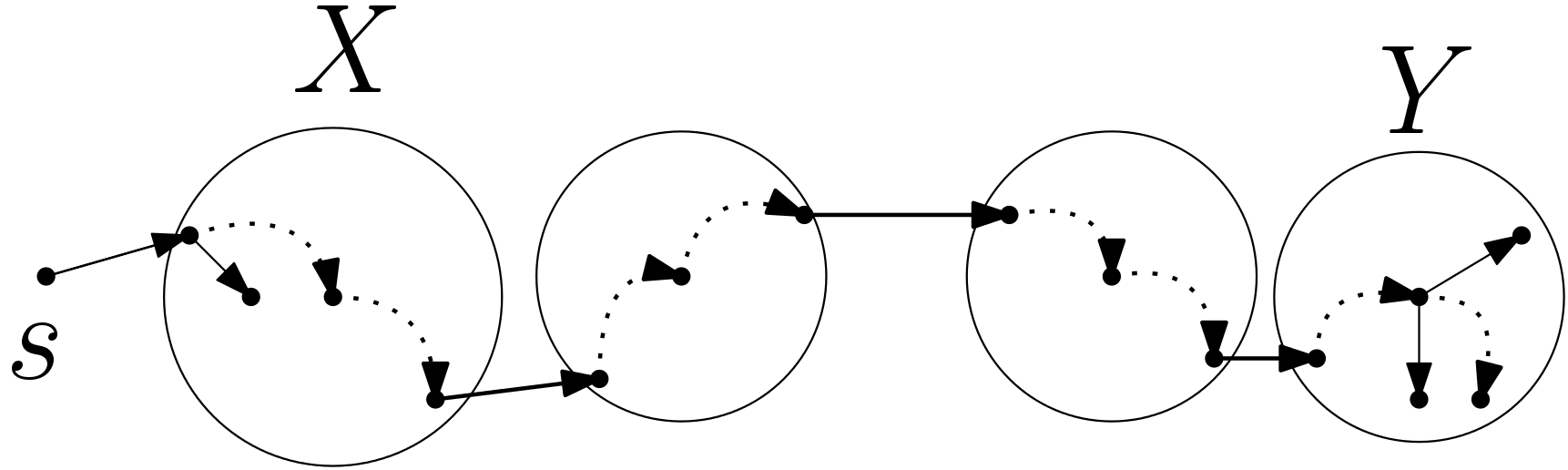}
	\caption{The tree $\mathcal{T}$ after the second step. Each edge between neighboring node centers is replaced by a path that contains two star edges (dotted arrows) and one graph edge (solid arrows).	}
	\label{fig step2}
	\end{subfigure}\hfill
	\begin{subfigure}{.48\textwidth}
			\centering
		\includegraphics[scale=0.15]{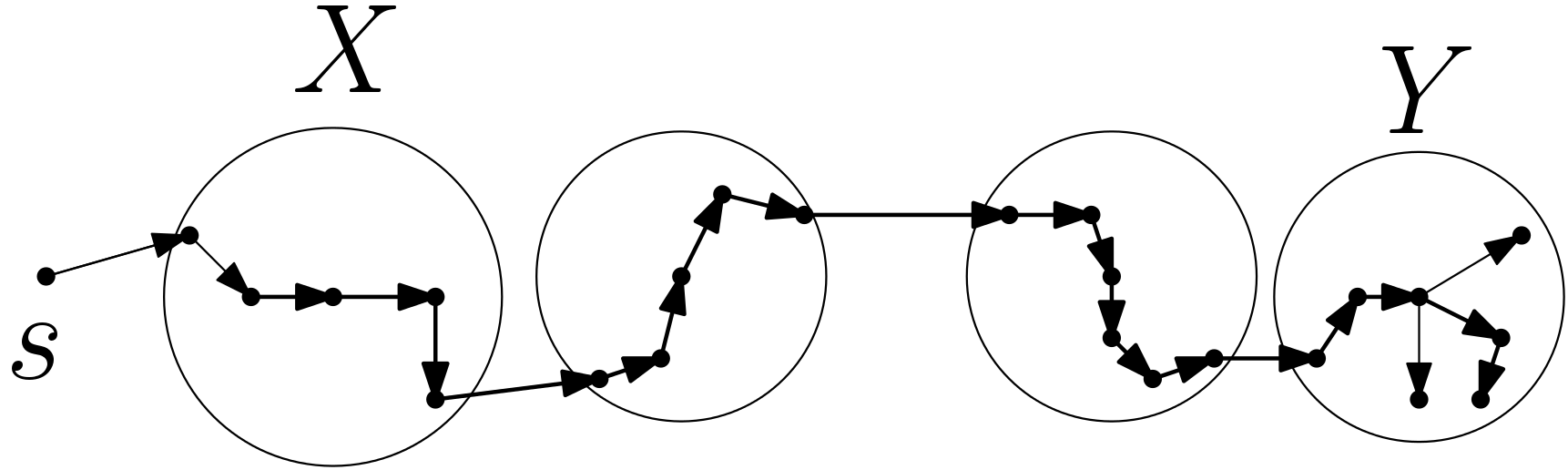}
		\caption{The tree $\mathcal{T}$ after the third (and final) step. Star edges (dotted arrows) are replaced by graph edges (solid arrows).	}
		\label{fig step3}
	\end{subfigure}
	\caption{The three steps for replacing hopset edges with graph edges.}
	\label{fig steps}
\end{figure}

\textbf{The first step} consists of executing Algorithm \ref{alg path-report} described in Section \ref{sec path-reporting} to replace each \emph{hop-edge} $(x\str,y\str)$ with the memory paths associated with it. Recall that this procedure replaces \emph{all} hop-edges, i.e., when it terminates, there are no hop-edges left in $\mathcal{T}$. 
This procedure is executed in parallel for every $k\in K$. 
When the procedure terminates, each vertex selects the best distance estimate (and parent) provided to it by the respective procedure. We note that some vertices may receive distance estimates from more then one procedure. Each vertex chooses the smallest distance estimate it receives. 
Observe that for every vertex $v\in V$ that has updated its parent during this step, we have that at the end of this step both $p(v)$ and $v$ are centers of neighboring nodes. Recall that during Algorithm \ref{alg path-report} vertices do not increase their distance estimate.
By arguments similar to those given in Section \ref{sec path-reporting}, one can show that for every vertex $v\in V\setminus \{s\}$, we have that $d(v)$ is an upper bound on the distance of $v$ from $s$ in the tree $\mathcal{T}$, and also $d(p(v))< d(v)$. Therefore, we have that at the end of the first step, $\mathcal{T}$ is still a $(1+\epsilon)$-SPT for $G\cup H$.

\textbf{The second step} consists of replacing edges between neighboring nodes centers with star edges and original graph edges. Let $(p(v),v)$ be an edge in $\mathcal{T}$, such that $X,Y$ are the nodes centered around $p(v),v$, respectively. Since $X,Y$ are neighboring nodes, by construction there exists an edge $(x,y)\in E$ such that $x\in X$ and $y\in Y$. Let $(x,y)\in E$ be the lightest edge such that $x\in X$ and $y\in Y$. (For every superedge $(X,Y)$ in $\bigcup_{k\in K}\mch{k}$, the lightest edge $(x,y)\in E\cup \{X\times Y\}$ was already computed.)
 See Figure \ref{fig rephop} for an illustration. The edge $(p(v),v)$ is now replaced by the path $P = \langle p(v)-x-y-v\rangle$. First, the vertex $v$ changes its parent and sets $p(v) = y$. Then, $v$ informs $x,y$ of the estimates they obtain from the estimate of $p(v)$ and the path $P$. 
Let $d_1= d(p(v)) + \omega(p(v),x)$, where $\omega(p(v),x)$ is the weight of the star edge between $p(v)$ and $x$. Let $d_2 = d_1 + \omega(x,y)$.

We maintain a global array $M$, in which every vertex has two designated cells. 
The vertex $v$ writes the triplets $\langle x,d_1,p(v)\rangle$ and $\langle y, d_2,x\rangle$ to its cells in $M$. 
The array $M$ is now sorted according to the first field of the triplets. Ties are broken according to the second field. Every vertex $u\in V\setminus \{s\}$ now searches for the best distance estimate that $M$ provides for it, i.e., the first triplet $\langle u,d,u'\rangle$ in $M$. If $d<d(u)$, then $u$ sets $p(u)= u'$ and $d(u)= d$. 
This concludes the description of the second step. 
When this step terminates, for every vertex $v\in V\setminus \{s\}$, the edge $(p(v),v)$ is either an edge of the original graph or a star edge. 

Observe that during the second step, vertices do not increase their distance estimates. In addition, for every vertex $v\in V$ that has updated its parent during this step, 
we have that its current parent and distance estimate $d(v),p(v)$ satisfy $d(p(v))+\omega(p(v),v)\leq d(v)$. Since all edge weights are positive, we also have $d(p(v))< d(v)$. Hence, there are no cycles in $\mathcal{T}$. Therefore, $\mathcal{T}$ is still a $(1+\epsilon)$-SPT for $G\cup H$. 

\begin{figure}		
	\centering
	\includegraphics[scale=0.23]{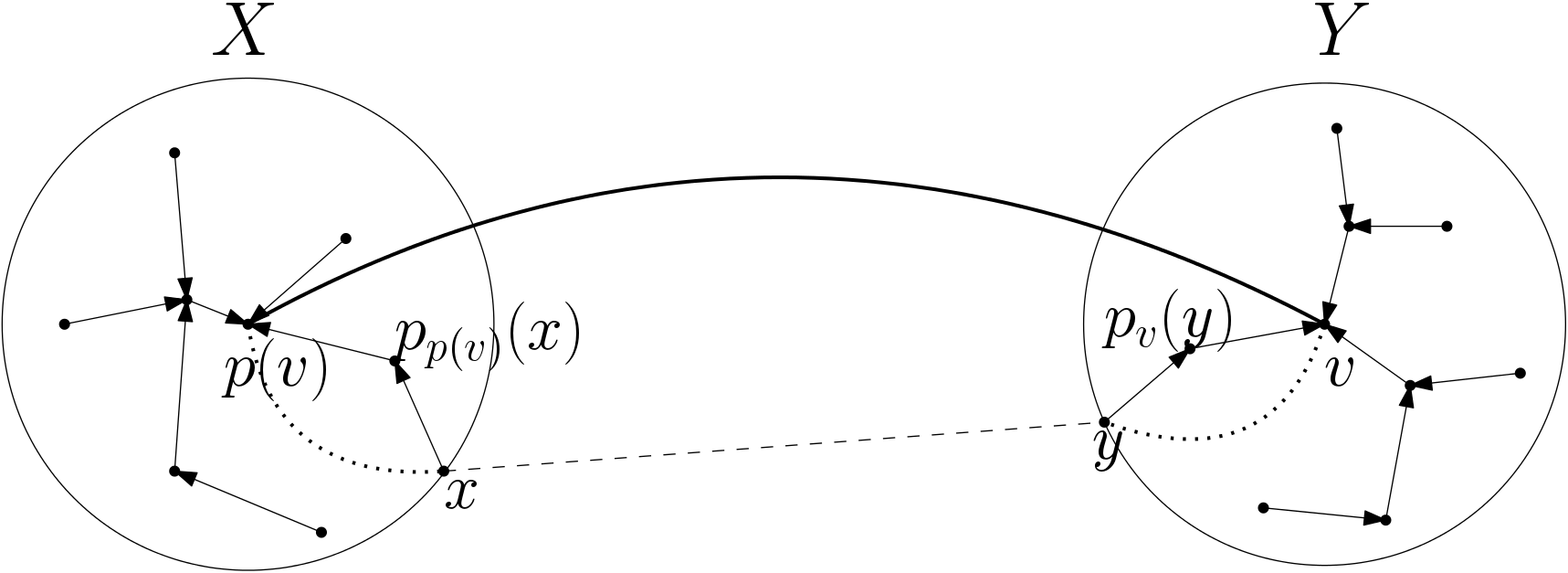}
	\caption{The second step. The curved line depicts an edge $(p(v),v)$ in the tree $\mathcal{T}$ that belongs to some graph $\mcg{k}$. The dashed edge $(x,y)$ is the lightest edge in $E$ such that $x$ belongs to the node of $p(v)$ and $y$ belongs to the node of $v$. The arrows depict the spanning tree of each node. The dotted curved edges represent star edges. In the second step, the edge $(p(v),v)$ is replaced by the path $\langle p(v)-x-y-v\rangle$, where $(p(v),x),(y,v)$ are star edges.}
	\label{fig rephop}
\end{figure}

\textbf{The third step} consists of replacing star edges from $\mathcal{T}$ with original graph edges. Recall that for each node $U$ centered around a vertex $u\str$, a spanning tree $T_U$ was computed. In addition, each vertex $u \in U\setminus \{ u\str \}$ knows its neighbor $p_{u\str}(u)$ on the path from $u$ to $u\str$ in the spanning tree $T_U$, and also the distance in $T_U$ between $u\str$ and $u$. 
We replace star edges with edges of the spanning trees of the corresponding nodes. Let $(p(v),v)$ be a star edge . We say that $(p(v),v)$ is an edge of \emph{type A} if $p(v)$ is the center of the node of $p(v)$ and $v$. Otherwise (i.e., $v$ is the center of the node of $p(v)$ and $v$), it is said to be an edge of \emph{type B}. 

Consider an edge $(p(v),v)$ of type $A$. The vertex $v$ updates its parent $p(v)$ to be $p_{p(v)}(v)$. We note that it is possible that at this point, $d(p(v))+\omega(p(v),v)> d(v)$. To remedy this, every vertex $u\in V$ now checks if there exists a node $U$ such that $u\in U$, and such that the center $u\str$ of $U$ satisfies $d(u\str)+d_{T_{U}}(u\str,u)<d(u)$. (Note that this computation can be performed in $O({\log n})$ time using $O(n{\log n})$ processors, using the array $B$ computed together with the set of star edges $S$. See the discussion that follows Lemma \ref{lemma bound size s}.) If there exists such node, then $u$ updates its parent and distance estimate accordingly. Let $U,u\str$ be the node and node center such that $u\in U$, and $d(u\str)+d_{T_{U}}(u\str,u)$ is minimal. Then, $u$ sets $d(u) = d(u\str)+d_{T_{U}}(u\str,u)$ and also $p(u) = p_{u\str}(u)$.
Observe that at this point there are no edges of type $A$, and for every $v\in V\setminus\{s\}$ we have $d(p(v))+\omega(p(v),v)\leq d(v)$. See Figure \ref{fig ustpath} for an illustration.

We now remove edges of type $B$ from $\mathcal{T}$. 
Consider an edge $(u,u\str)$ of type $B$, such that $u\str$ is the center of a node $U$ and $u\in U$. 
Intuitively, replacing this edge requires flipping the direction of the parent-child relationship along the $u-u\str$ path $\pi_U(u,u\str)$ in the spanning tree $T_U$.

This is done in the following manner. First, $u\str$ informs $u$ that the star edge must be replaced. Then, a pointer-jumping algorithm is used to inform all vertices on the path $\pi_U(u,u\str)$ that they belong to this path. Every vertex $v$ (other than $u\str$) along this path now writes to a designated array $c_U(p(v)) = v$, to inform its parent $p(v)$ that $v$ is its child along this path. The vertex $u\str$ now sets $p(u\str)= c_U(u\str)$. At this point, there are no more edges of type $B$ in $\mathcal{T}$. 
However, $d(u\str)$ may be smaller than $d(p(u\str))+\omega(p(u\str),u\str)$. To remedy this, a pointer-jumping algorithm is used again to compute for every $v\in P$ the distance $d_P(u,v)$, i.e., its distance from $u$ along the path $\pi_U(u,u\str)$. Observe that $u$ may belong to multiple paths that require flipping. For every vertex $v\in V$, let $P=\pi_U(u,u\str)$ be the path such that $d(u)+d_P(u,v)$ is minimal. If $d(u)+d_P(u,v)<d(v)$, then 
$v$ updates $p(v)= c_U(v)$ and $d(v)= d(u)+d_P(u,v)$. When this process terminates, we have $d(p(v))+\omega(p(v),v)\leq d(v)$ for every vertex $v\in V\setminus \{ s\}$.

\begin{figure}
	\centering
	\includegraphics[scale=0.13]{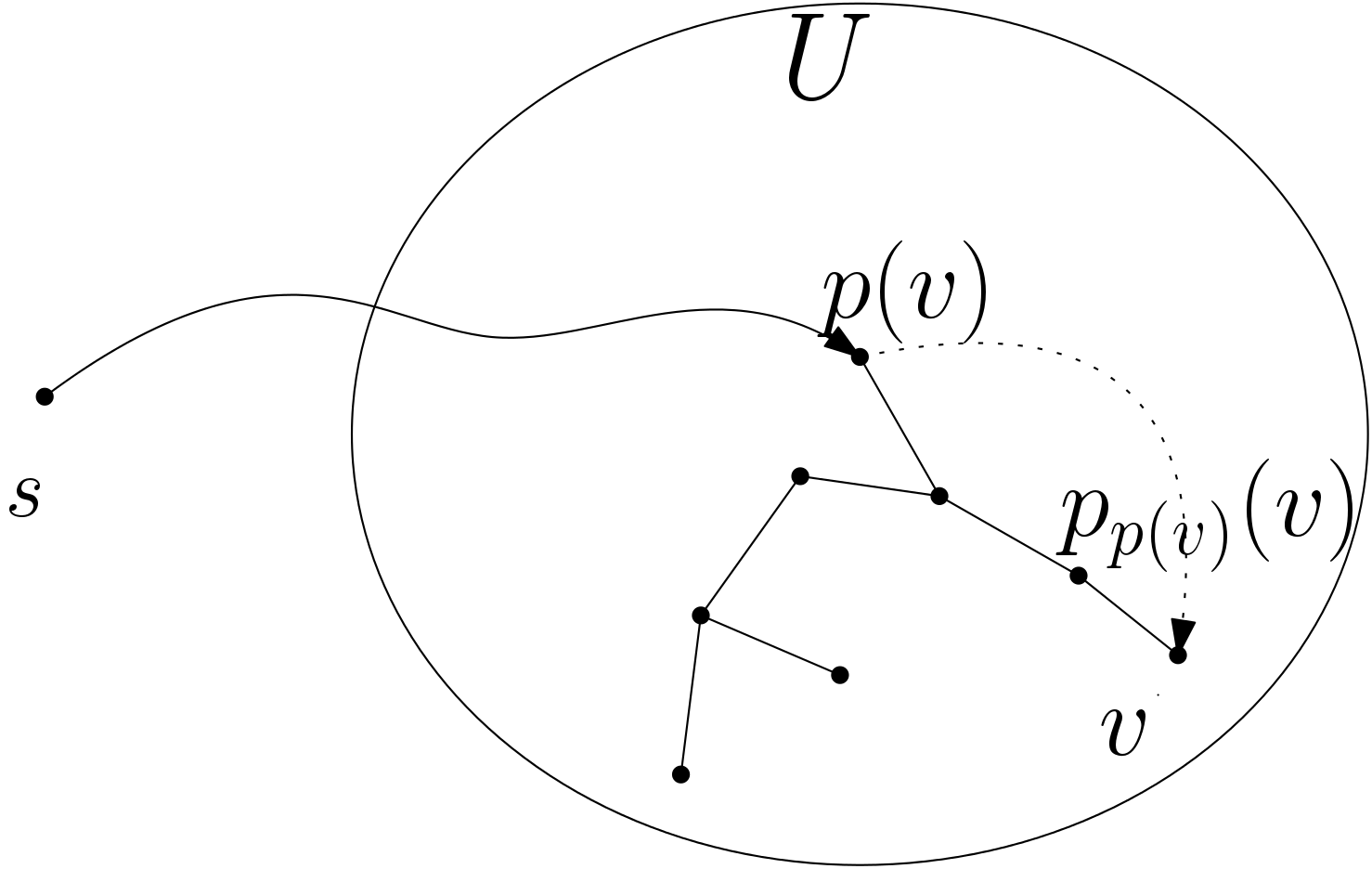}
	\caption{Replacing edges of type $A$. The dotted edge $(p(v),v)$ is an edge of type $A$, where $p(v)$ is the center of the node $U$. The straight lines depict edges of the spanning tree $T_U$ of the node $U$. 
		The vertex $v$ sets its parent to be its neighbor on the path in $T_U$ from it to $p(v)$. In addition, every vertex $u\in U$ ensures that its distance estimate is no greater than $d(p(v))+ d_{T_U}(u)$. }
	\label{fig ustpath}
\end{figure}

\begin{figure}
	\begin{subfigure}{.48\textwidth}
		\centering
		\includegraphics[scale=0.1]{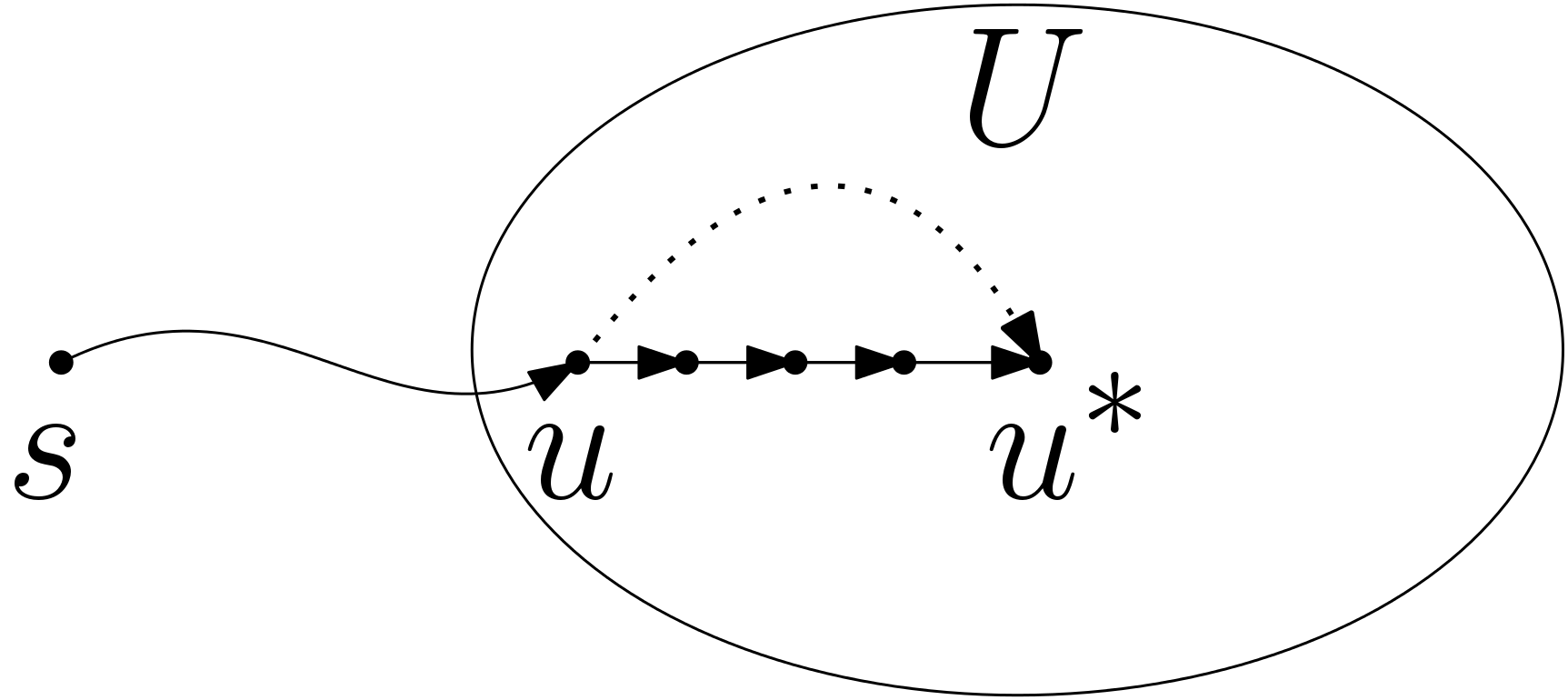}
		\caption{ 	}
		\label{fig B1}
	\end{subfigure} \hfill
	\begin{subfigure}{.48\textwidth}
		\centering
		\includegraphics[scale=0.1]{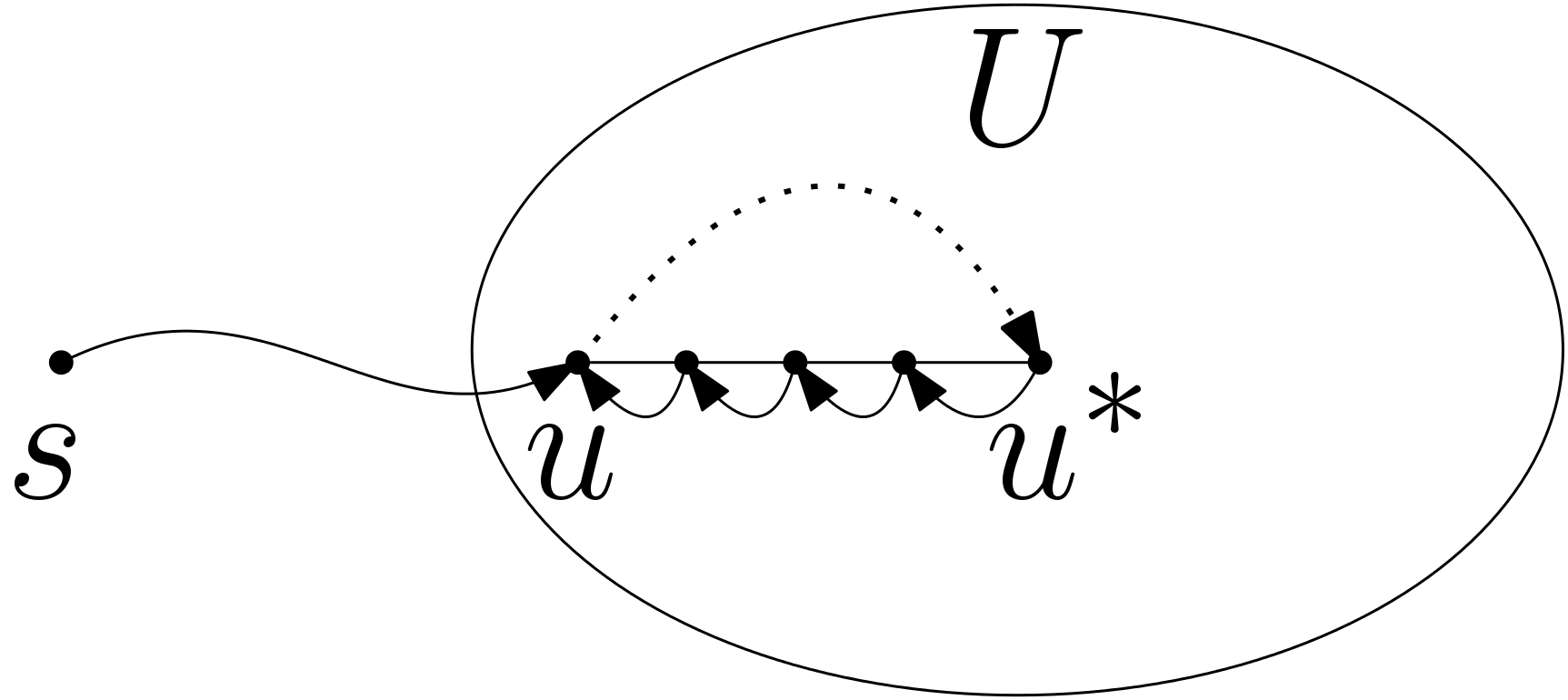}
		\caption{ }
		\label{fig B2}
	\end{subfigure}
	\caption{Replacing edges of type $B$. The dotted edge $(u,u\str)$ is an edge of type B, where $u\str$ is the center of the node $U$. In Figure \ref{fig B1}, the arrows depict the path $P$ from $u$ to $u\str$ in the spanning tree $T_U$ of a node $U$. In Figure \ref{fig B2} the arrows depict the path $P$ after the direction of all edges has been flipped.}
	\label{fig type B}
\end{figure}

This completes the description of the path-reporting algorithm. 
Define now $T=(V,E_T)$ as the tree obtained by executing steps $1-3$. 
Observe that throughout the algorithm, the distance estimate of each vertex did not increase w.r.t. the estimate it obtained by the Bellman-Ford exploration in $G\cup H$. Therefore, at the end of the exploration, for every $v\in V$ we have that $d(v)\leq (1+\epsilon)d_G(s,v)$. In addition, for every vertex $v\in V\setminus\{s\}$, we have $d(v)\geq d(p(v))+\omega(p(v),v)$, where $\omega(p(v),v)$ is the weight of the edge $(p(v),v)$ in the original graph $G$. Since the edge weights are positive, there are no cycles in $\mathcal{T}$.

When the three steps terminate, for every vertex $v\in V\setminus \{s\}$ we have $(p(v),v)\in E$. Let now $T=(V,E_T)$, where $E_T = \{ (p(v),v)\ | \ v\in V\setminus \{s\}\}$. As in Section \ref{sec path-reporting}, we now use a pointer-jumping algorithm to compute for every vertex $v\in V$ its distance from $s$ w.r.t. the tree $T$. Observe that for every vertex $v\in V$, its final distance estimate is at most the original estimate we obtained by the Bellman-Ford exploration. 
Hence $T$ is a $(1+\epsilon)$-SPT for $G$, and also $E_T\subseteq E$. 

\subsubsection{Complexity}
By Theorem \ref{theo path reduc}, the hopset $H$ can be computed in $O(({\log \kappa\rho}+1/\rho)\beta{\log^3 n})$ time using 
$(|E|+\nfrac)\cdot n^\rho\cdot\left({\log n}\right) ^{O({\log \kappa\rho}+1/\rho)} $ processors.

As in Section \ref{sec path-reporting}, the time and work required for the Bellman-Ford exploration and for the first step are dominated by the time and work required for constructing the hopset $H$. 

The second step requires some vertices to write to $O(1)$ cells in an array of length $O(n)$, sorting the array, and executing $n$ binary search processes (in parallel). This can be done in $O(\log n)$ time using $O(n)$ processors.

The third step consists of two parts, i.e., replacing edges of type $A$ and of type $B$. Replacing edges of type $A$ requires each vertex to check whether a center of its node provides it with a better distance estimate than the estimate it currently possesses. 
Observe that some vertices may be members of a linear number of nodes. Therefore, a naive implementation of this step that uses only $O(n^\rho)$ processors to simulate each vertex may require polynomial time. However, for every center $u\str$ of a node $U$ that a vertex $u$ belongs to, there is a star edge $(u,u\str)$ in $S$. This edge was allocated $O(n^\rho)$ processors. To execute this step efficiently, each star edge $(u,u\str)\in S$ contributes a single processor $proc_{(u,u\str)}$ to assist $u$ with checking whether this star edge provides it with a better distance estimate. To select the best distance estimate, the star edges processor $proc_{(u,u\str)}$ writes the distance estimate that $(u,u\str)$ provides for $u$ to a designated array, which is then sorted to find its minimum. This can be done using $\tilde{O}(n)$ processors in $O({\log n})$ time. 

The second part, requires \emph{flipping} the direction of some paths in the spanning trees of some nodes. 
For every node $U$, this can be done using $O(|U|)$ processors in $O({\log n})$ time in the CREW PRAM model. This algorithm is executed in parallel for all nodes. Every vertex $u\in U$ is simulated by the star edge from $u$ to the center of $U$. The center of $U$ is simulated by an additional processor. Recall that by \cref{eq bound s}, we have $|S|\leq n{\log n}$. Therefore, this algorithm can be executed in $O({\log n})$ time using $\tilde{O}(n)$ processors. 

It follows that the overall time and work required to complete the construction of the path-reporting hopset and to retrieve a $\oeps$-SPT in $G$ is dominated by the time and work required for constructing the hopset $H$. 
The following theorem summarizes the properties of the path-reporting algorithm.

\begin{theorem}
	\label{theo peth reporting reduc}
	Given a weighted undirected graph $G=(V,E,\omega)$ on $n$ vertices, a source vertex $s\in V$ and parameters $0<\epsilon<1$, $\kappa =2,3,\dots$, and $0< \rho<1/2$, our algorithm deterministically solves the $(1+\epsilon)$-approximate-shortest-path problem in $O(({\log \kappa\rho}+1/\rho)\beta{\log^3 n})$ time in the PRAM CREW model using $(|E|+\nfrac)\cdot n^\rho \cdot \left({\log n}\right)^{O({\log \kappa\rho}+1/\rho)}$ processors, where 
	$$\beta =O\left( \frac{{\log^2n}\cdot ({\log \kappa\rho} + 1/\rho) }{\epsilon} \right)^{\pramell}.$$
\end{theorem}

Note that $\rho \geq 1/{\log n}$, as otherwise the result is meaningless. Thus, $\beta = \left(\frac{\log n}{\epsilon}\right)^{O({\log \kappa\rho}+1/\rho)}$, and the time is bounded by this expression as well.

\end{appendices}

\bibliographystyle{alpha}
\bibliography{cite}	

\begin{thebibliography}{AGHP92}

\bibitem[ABP17]{AbboudBP17}
Amir Abboud, Greg Bodwin, and Seth Pettie.
\newblock A hierarchy of lower bounds for sublinear additive spanners.
\newblock In {\em Proceedings of the Twenty-Eighth Annual {ACM-SIAM} Symposium
  on Discrete Algorithms, {SODA} 2017, Barcelona, Spain, Hotel Porta Fira,
  January 16-19}, pages 568--576, 2017.

\bibitem[AGHP92]{AlonGHP92}
Noga Alon, Oded Goldreich, Johan H{\aa}stad, and Ren{\'{e}} Peralta.
\newblock Simple construction of almost k-wise independent random variables.
\newblock {\em Random Struct. Algorithms}, 3(3):289--304, 1992.

\bibitem[AGLP89]{awerbuch1989network}
Baruch Awerbuch, Andrew~V. Goldberg, Michael Luby, and Serge~A. Plotkin.
\newblock Network decomposition and locality in distributed computation.
\newblock In {\em 30th Annual Symposium on Foundations of Computer Science,
  Research Triangle Park, North Carolina, USA, 30 October - 1 November 1989},
  pages 364--369, 1989.

\bibitem[AGM97]{AlonGM97}
Noga Alon, Zvi Galil, and Oded Margalit.
\newblock On the exponent of the all pairs shortest path problem.
\newblock {\em J. Comput. Syst. Sci.}, 54(2):255--262, 1997.

\bibitem[AKS83]{AjtaiKS83}
Mikl{\'{o}}s Ajtai, J{\'{a}}nos Koml{\'{o}}s, and Endre Szemer{\'{e}}di.
\newblock An o(n log n) sorting network.
\newblock In {\em Proceedings of the 15th Annual {ACM} Symposium on Theory of
  Computing, 25-27 April, 1983, Boston, Massachusetts, {USA}}, pages 1--9,
  1983.

\bibitem[ASZ19]{AndoniSZ19}
Alexandr Andoni, Clifford Stein, and Peilin Zhong.
\newblock Parallel approximate undirected shortest paths via low hop emulators.
\newblock {\em CoRR}, abs/1911.01956, 2019.

\bibitem[Ber09]{Bernstein09}
Aaron Bernstein.
\newblock Fully dynamic {(2} + epsilon) approximate all-pairs shortest paths
  with fast query and close to linear update time.
\newblock In {\em 50th Annual {IEEE} Symposium on Foundations of Computer
  Science, {FOCS} 2009, October 25-27, 2009, Atlanta, Georgia, {USA}}, pages
  693--702. {IEEE} Computer Society, 2009.

\bibitem[BKKL16]{BeckerKKL16}
Ruben Becker, Andreas Karrenbauer, Sebastian Krinninger, and Christoph Lenzen.
\newblock Approximate undirected transshipment and shortest paths via gradient
  descent.
\newblock {\em CoRR}, abs/1607.05127, 2016.

\bibitem[BR11]{BernsteinR11}
Aaron Bernstein and Liam Roditty.
\newblock Improved dynamic algorithms for maintaining approximate shortest
  paths under deletions.
\newblock In Dana Randall, editor, {\em Proceedings of the Twenty-Second Annual
  {ACM-SIAM} Symposium on Discrete Algorithms, {SODA} 2011, San Francisco,
  California, USA, January 23-25, 2011}, pages 1355--1365. {SIAM}, 2011.

\bibitem[BRS94]{BergerRS94}
Bonnie Berger, John Rompel, and Peter~W. Shor.
\newblock Efficient {NC} algorithms for set cover with applications to learning
  and geometry.
\newblock {\em J. Comput. Syst. Sci.}, 49(3):454--477, 1994.

\bibitem[Coh94]{coh94}
Edith Cohen.
\newblock Polylog-time and near-linear work approximation scheme for undirected
  shortest paths.
\newblock In {\em Proceedings of the Twenty-Sixth Annual {ACM} Symposium on
  Theory of Computing, 23-25 May 1994, Montr{\'{e}}al, Qu{\'{e}}bec, Canada},
  pages 16--26, 1994.

\bibitem[Coh97]{Cohen97}
Edith Cohen.
\newblock Using selective path-doubling for parallel shortest-path
  computations.
\newblock {\em J. Algorithms}, 22(1):30--56, 1997.

\bibitem[CW87]{CoppersmithW87}
Don Coppersmith and Shmuel Winograd.
\newblock Matrix multiplication via arithmetic progressions.
\newblock In Alfred~V. Aho, editor, {\em Proceedings of the 19th Annual {ACM}
  Symposium on Theory of Computing, 1987, New York, New York, {USA}}, pages
  1--6. {ACM}, 1987.

\bibitem[EGN19]{EGN19}
Michael Elkin, Yuval Gitlitz, and Ofer Neiman.
\newblock Almost shortest paths and {PRAM} distance oracles in weighted graphs.
\newblock {\em CoRR}, abs/1907.11422, 2019.

\bibitem[EM19]{ElkinMatar}
Michael Elkin and Shaked Matar.
\newblock Near-additive spanners in low polynomial deterministic {CONGEST}
  time.
\newblock {\em CoRR}, abs/1903.00872, 2019.

\bibitem[EN17a]{ElkinN17spanners}
Michael Elkin and Ofer Neiman.
\newblock Efficient algorithms for constructing very sparse spanners and
  emulators.
\newblock In {\em Proceedings of the Twenty-Eighth Annual {ACM-SIAM} Symposium
  on Discrete Algorithms, {SODA} 2017, Barcelona, Spain, Hotel Porta Fira,
  January 16-19}, pages 652--669, 2017.

\bibitem[EN17b]{ElkinN17Hop}
Michael Elkin and Ofer Neiman.
\newblock Linear-size hopsets with small hopbound, and distributed routing with
  low memory.
\newblock {\em CoRR}, abs/1704.08468, 2017.

\bibitem[EN18]{ElkinN18rout}
Michael Elkin and Ofer Neiman.
\newblock Near-optimal distributed routing with low memory.
\newblock In Calvin Newport and Idit Keidar, editors, {\em Proceedings of the
  2018 {ACM} Symposium on Principles of Distributed Computing, {PODC} 2018,
  Egham, United Kingdom, July 23-27, 2018}, pages 207--216. {ACM}, 2018.

\bibitem[EN19]{ElkinN19}
Michael Elkin and Ofer Neiman.
\newblock Hopsets with constant hopbound, and applications to approximate
  shortest paths.
\newblock {\em {SIAM} J. Comput.}, 48(4):1436--1480, 2019.

\bibitem[EN20]{EN20}
Michael Elkin and Ofer Neiman.
\newblock Centralized and parallel multi-source shortest paths via hopsets and
  fast matrix multiplication.
\newblock {\em CoRR}, abs/2004.07572, 2020.

\bibitem[EP01]{ElkinP01}
Michael Elkin and David Peleg.
\newblock (1+epsilon, beta)-spanner constructions for general graphs.
\newblock In {\em Proceedings on 33rd Annual {ACM} Symposium on Theory of
  Computing, July 6-8, 2001, Heraklion, Crete, Greece}, pages 173--182, 2001.

\bibitem[GM97]{GalilM97}
Zvi Galil and Oded Margalit.
\newblock All pairs shortest paths for graphs with small integer length edges.
\newblock {\em J. Comput. Syst. Sci.}, 54(2):243--254, 1997.

\bibitem[GPS88]{GoldbergPS88}
Andrew~V. Goldberg, Serge~A. Plotkin, and Gregory~E. Shannon.
\newblock Parallel symmetry-breaking in sparse graphs.
\newblock {\em {SIAM} J. Discrete Math.}, 1(4):434--446, 1988.

\bibitem[GU18]{GallU18}
Francois~Le Gall and Florent Urrutia.
\newblock Improved rectangular matrix multiplication using powers of the
  coppersmith-winograd tensor.
\newblock In Artur Czumaj, editor, {\em Proceedings of the Twenty-Ninth Annual
  {ACM-SIAM} Symposium on Discrete Algorithms, {SODA} 2018, New Orleans, LA,
  USA, January 7-10, 2018}, pages 1029--1046. {SIAM}, 2018.

\bibitem[HKN16]{HenzingerKN16}
Monika Henzinger, Sebastian Krinninger, and Danupon Nanongkai.
\newblock A deterministic almost-tight distributed algorithm for approximating
  single-source shortest paths.
\newblock In {\em Proceedings of the 48th Annual {ACM} {SIGACT} Symposium on
  Theory of Computing, {STOC} 2016, Cambridge, MA, USA, June 18-21, 2016},
  pages 489--498, 2016.

\bibitem[HP19]{HuangP19}
Shang{-}En Huang and Seth Pettie.
\newblock Thorup-zwick emulators are universally optimal hopsets.
\newblock {\em Inf. Process. Lett.}, 142:9--13, 2019.

\bibitem[J{\'{a}}J92]{JaJa92}
Joseph J{\'{a}}J{\'{a}}.
\newblock {\em An Introduction to Parallel Algorithms}.
\newblock Addison-Wesley, 1992.

\bibitem[KMW18]{KuhnMW18}
Fabian Kuhn, Yannic Maus, and Simon Weidner.
\newblock Deterministic distributed ruling sets of line graphs.
\newblock In {\em Structural Information and Communication Complexity - 25th
  International Colloquium, {SIROCCO} 2018, Ma'ale HaHamisha, Israel, June
  18-21, 2018, Revised Selected Papers}, pages 193--208, 2018.

\bibitem[KP83]{KarpP83}
Richard. Karp and Nicholas Pippenger.
\newblock A time-randomness tradeoff.
\newblock In {\em AMS Conference on Probabilistic Computational Complexity},
  volume 111, 1983.

\bibitem[KR88]{KarloffR88}
Howard~J. Karloff and Prabhakar Raghavan.
\newblock Randomized algorithms and pseudorandom numbers.
\newblock In Janos Simon, editor, {\em Proceedings of the 20th Annual {ACM}
  Symposium on Theory of Computing, May 2-4, 1988, Chicago, Illinois, {USA}},
  pages 310--321. {ACM}, 1988.

\bibitem[KR90]{KarpR90}
Richard~M. Karp and Vijaya Ramachandran.
\newblock Parallel algorithms for shared-memory machines.
\newblock In Jan van Leeuwen, editor, {\em Handbook of Theoretical Computer
  Science, Volume {A:} Algorithms and Complexity}, pages 869--942. Elsevier and
  {MIT} Press, 1990.

\bibitem[KS92]{KleinS92}
Philip~N. Klein and Sairam Sairam.
\newblock A parallel randomized approximation scheme for shortest paths.
\newblock In S.~Rao Kosaraju, Mike Fellows, Avi Wigderson, and John~A. Ellis,
  editors, {\em Proceedings of the 24th Annual {ACM} Symposium on Theory of
  Computing, May 4-6, 1992, Victoria, British Columbia, Canada}, pages
  750--758. {ACM}, 1992.

\bibitem[KS97]{KleinS97}
Philip~N. Klein and Sairam Subramanian.
\newblock A randomized parallel algorithm for single-source shortest paths.
\newblock {\em J. Algorithms}, 25(2):205--220, 1997.

\bibitem[KW85]{KarpW85}
Richard~M. Karp and Avi Wigderson.
\newblock A fast parallel algorithm for the maximal independent set problem.
\newblock {\em J. {ACM}}, 32(4):762--773, 1985.

\bibitem[Li19]{Li19}
Jason Li.
\newblock Faster parallel algorithm for approximate shortest path.
\newblock {\em CoRR}, abs/1911.01626, 2019.

\bibitem[Lub86]{Luby86}
Michael Luby.
\newblock A simple parallel algorithm for the maximal independent set problem.
\newblock {\em {SIAM} J. Comput.}, 15(4):1036--1053, 1986.

\bibitem[Mad13]{Madry13}
Aleksander Madry.
\newblock Navigating central path with electrical flows: From flows to
  matchings, and back.
\newblock In {\em 54th Annual {IEEE} Symposium on Foundations of Computer
  Science, {FOCS} 2013, 26-29 October, 2013, Berkeley, CA, {USA}}, pages
  253--262. {IEEE} Computer Society, 2013.

\bibitem[NN93]{NaorN93}
Joseph Naor and Moni Naor.
\newblock Small-bias probability spaces: Efficient constructions and
  applications.
\newblock {\em {SIAM} J. Comput.}, 22(4):838--856, 1993.

\bibitem[NW88]{NisanW88}
Noam Nisan and Avi Wigderson.
\newblock Hardness vs. randomness (extended abstract).
\newblock In {\em 29th Annual Symposium on Foundations of Computer Science,
  White Plains, New York, USA, 24-26 October 1988}, pages 2--11. {IEEE}
  Computer Society, 1988.

\bibitem[SEW13]{sew}
Johannes Schneider, Michael Elkin, and Roger Wattenhofer.
\newblock Symmetry breaking depending on the chromatic number or the
  neighborhood growth.
\newblock {\em Theor. Comput. Sci.}, 509:40--50, 2013.

\bibitem[She16]{Sherman16}
Jonah Sherman.
\newblock Generalized preconditioning and network flow problems.
\newblock {\em CoRR}, abs/1606.07425, 2016.

\bibitem[Spe97]{Spencer97}
Thomas~H. Spencer.
\newblock Time-work tradeoffs for parallel algorithms.
\newblock {\em J. {ACM}}, 44(5):742--778, 1997.

\bibitem[SS99]{ss99}
Hanmao Shi and Thomas~H. Spencer.
\newblock Time-work tradeoffs of the single-source shortest paths problem.
\newblock {\em J. Algorithms}, 30(1):19--32, 1999.

\bibitem[SV82]{ShiloachV82}
Yossi Shiloach and Uzi Vishkin.
\newblock An o(log n) parallel connectivity algorithm.
\newblock {\em J. Algorithms}, 3(1):57--67, 1982.

\bibitem[TZ01]{ThorupZ01}
Mikkel Thorup and Uri Zwick.
\newblock Compact routing schemes.
\newblock In {\em Proceedings of the thirteenth annual ACM symposium on
  Parallel algorithms and architectures}, pages 1--10. ACM, 2001.

\bibitem[TZ06]{ThorupZ06}
Mikkel Thorup and Uri Zwick.
\newblock Spanners and emulators with sublinear distance errors.
\newblock In {\em Proceedings of the Seventeenth Annual {ACM-SIAM} Symposium on
  Discrete Algorithms, {SODA} 2006, Miami, Florida, USA, January 22-26, 2006},
  pages 802--809, 2006.

\bibitem[UY91]{uy91}
Jeffrey~D. Ullman and Mihalis Yannakakis.
\newblock High-probability parallel transitive-closure algorithms.
\newblock {\em {SIAM} J. Comput.}, 20(1):100--125, 1991.

\bibitem[Wil12]{Williams12}
Virginia~Vassilevska Williams.
\newblock Multiplying matrices faster than coppersmith-winograd.
\newblock In Howard~J. Karloff and Toniann Pitassi, editors, {\em Proceedings
  of the 44th Symposium on Theory of Computing Conference, {STOC} 2012, New
  York, NY, USA, May 19 - 22, 2012}, pages 887--898. {ACM}, 2012.

\bibitem[Zwi98]{Zwick98}
Uri Zwick.
\newblock All pairs shortest paths in weighted directed graphs-exact and almost
  exact algorithms.
\newblock In {\em 39th Annual Symposium on Foundations of Computer Science,
  {FOCS} '98, November 8-11, 1998, Palo Alto, California, {USA}}, pages
  310--319. {IEEE} Computer Society, 1998.

\bibitem[Zwi02]{Zwick02}
Uri Zwick.
\newblock All pairs shortest paths using bridging sets and rectangular matrix
  multiplication.
\newblock {\em J. {ACM}}, 49(3):289--317, 2002.

\end{thebibliography}
\end{document}